%% file: main.tex
\documentclass[a4paper,UKenglish,cleveref, autoref, thm-restate]{lipics-v2021}

\title{Hardness of Burning Number Problem on Regular Graphs} %TODO Please add

\author{Dhanyamol Antony}{School of Data Science, Indian Institute of Science Education and Research Thiruvananthapuram, India}{dhanyamolantony@iisertvm.ac.in}{https://orcid.org/0000-0001-7875-3457}{ANRF/ECRG/2024/001049/ENS}{}
\author{L. Sunil Chandran}{Department of Computer Science and Automation, Indian Institute of Science, Bengaluru, India}{sunil@iisc.ac.in}{https://orcid.org/0000-0001-5451-6975}{}
\author{Anita Das}{Department of Mathematics, Manipal Institute of Technology Bengaluru, Manipal Academy of Higher Education, Manipal, India}{anita.das@manipal.edu}{https://orcid.org/0000-0002-0126-4716}{}
\author{Shirish Gosavi}{Department of Computer Science and Automation, Indian Institute of Science, Bengaluru, India}{shirishgp@iisc.ac.in}{}{}
\author{Dalu Jacob}{Department of Mathematics, Indian Institute of Technology Delhi, India}{dalujacob@maths.iitd.ac.in}{}{}
\author{Shashanka Kulamarva}{Graduate School of Informatics, Kyoto University, Kyoto, Japan}{kulamarva.shashanka.3k@kyoto-u.ac.jp}{https://orcid.org/0009-0002-2982-6044}{JST ASPIRE grant JPMJAP2302}{}

\authorrunning{D. Antony et al.}
\Copyright{Dhanyamol Antony, L. Sunil Chandran, Anita Das, Shirish Gosavi, Dalu Jacob, and Shashanka Kulamarva}



\EventEditors{John Q. Open and Joan R. Access}
\EventNoEds{2}
\EventLongTitle{42nd Conference on Very Important Topics (CVIT 2016)}
\EventShortTitle{CVIT 2016}
\EventAcronym{CVIT}
\EventYear{2016}
\EventDate{December 24--27, 2016}
\EventLocation{Little Whinging, United Kingdom}
\EventLogo{}
\SeriesVolume{42}
\ArticleNo{23}
\input{preamble}
\nolinenumbers
\begin{document}

\maketitle

\begin{abstract}
The Burning Number Problem (BNP) models the spread of information or contagion in a network through a discrete-time process on a graph. At each step, one new vertex is selected as a burning source, while fire simultaneously spreads from previously burned vertices to their neighbors. The \textit{burning number} of a graph is the minimum number of steps required to burn all vertices. The decision version asks whether the burning number is at most a given integer $k$. BNP is known to be NP-complete even on restricted graph classes such as path forests.
We study BNP on connected regular graphs, a natural and previously unexplored graph class. We prove that BNP is NP-complete on connected cubic graphs, and moreover APX-hard under this restriction. We further show that BNP remains APX-hard on connected $d$-regular graphs for every fixed $d \geq 4$.
\vspace{-1cm}
\end{abstract}

\label{sec:introduction}
\input{introduction}

\section{ Preliminaries}
\label{sec:prelim}
\input{preliminaries}

\section{ Cubic Graphs}
\label{sec:cubic}
\input{cubic}

\input{dregular}

\bibliography{main}

\end{document}

%% file: preamble.tex
\usepackage{amsmath,amssymb,amsthm}
\usepackage{algorithmic}
\usepackage{algorithm2e, setspace}

\usepackage[symbol]{footmisc}
\usepackage{mathtools}
\usepackage{rotating}
\usepackage{tikz}
\usepackage{parskip}
\usepackage{lipsum}
\usepackage[inline]{enumitem}

\usepackage[normalem]{ulem}
\usepackage[numbers]{natbib}
\usepackage{graphicx} 
\usepackage[skipabove=8pt,skipbelow=4pt,innertopmargin=6pt,innerbottommargin=6pt]{mdframed}
\usepackage{titlesec}
\usepackage{orcidlink}
\usepackage{comment}
\usepackage{caption}
\usepackage{subcaption}
\usepackage{color,soul}
\usepackage{tcolorbox}
\usepackage{makecell}

\usetikzlibrary{decorations.pathreplacing,calligraphy}
\usetikzlibrary{positioning, calc,chains,fit,shapes}
\usetikzlibrary{quotes}
\usepackage[colorinlistoftodos]{todonotes} % For Comments
\titleformat{\subsection}
{\bfseries}{\thesubsection}{1em}{}
\renewcommand{\qed}{\qedsymbol}
\renewcommand\qedsymbol{$\hfill\blacksquare$}
\renewcommand\proofname{\textbf{Proof}}

\newcommand{\APH}{APX-hard}
\newcommand{\NPC}{NP-complete}

\newtheorem{construction}{Construction}

\tikzstyle{filled vertex}  = [{circle,draw=blue,fill=black!50,inner sep=1pt}]  % regular vertex
\tikzstyle{empty vertex}  = [{circle, draw, fill = white, inner sep=1.5pt, minimum width=1.5pt}]
\tikzset{
  corner/.style  = {fill=gray!30, draw= gray, thick, inner sep=2pt},
  b-vertex/.style = {fill=RubineRed, diamond, draw=blue, inner sep=1.5pt},
  a-vertex/.style = {draw=blue, fill=cyan, inner sep=2pt}
}

\newcommand{\BNP}{\textit{Burning Number Problem}}

\newcommand{\MVC}{\textit{Minimum Vertex Cover Problem}}

\newcommand {\tgadgetleftfull} [8] {

\begin{scope}[rotate around={#8:#1}]
\coordinate  (P1) at #1;
\end{scope}

\begin{scope}[rotate around={#8:(P1)}]
\coordinate  (P2) at ($(P1) + (#2, -#3)$);
\end{scope}

\begin{scope}[rotate around={#8:(P2)}]
\coordinate (P3) at ($(P2) + (0, -#4)$);
\end{scope}

\begin{scope}[rotate around={#8:(P3)}]
\coordinate  (P4) at ($(P3) + (-#2, -#3)$);
\end{scope}

\begin{scope}[rotate around={#8:(P4)}]
\coordinate  (P5) at ($(P4) + (-#2, #3)$);
\end{scope}

\begin{scope}[rotate around={#8:(P5)}]
\coordinate (P6) at ($(P5) + (0, #5)$);
\end{scope}

\begin{scope}[rotate around={#8:(P6)}]
\coordinate  (P7) at ($(P6) + (-#6, 0)$);
\end{scope}

\begin{scope}[rotate around={#8:(P7)}]
\coordinate  (P8) at ($(P7) + (0, #7)$);
\end{scope}

\begin{scope}[rotate around={#8:(P8)}]
\coordinate  (P9) at ($(P8) + (#6, 0)$);
\end{scope}

\begin{scope}[rotate around={#8:(P1)}]
\coordinate  (P10) at ($(P1) + (-#2, -#3)$);
\end{scope}

\fill [gray] (P1) -- (P2) -- (P3) -- (P4) -- (P5) -- (P6) -- (P7) -- (P8) -- (P9) -- (P10) -- cycle;
}

\newcommand {\tgadgetrightfull} [8] {
\begin{scope}[rotate around={#8:#1}]
\coordinate  (P1) at #1;
\end{scope}

\begin{scope}[rotate around={#8:(P1)}]
\coordinate  (P2) at ($(P1) + (-#2, -#3)$);
\end{scope}

\begin{scope}[rotate around={#8:(P2)}]
\coordinate (P3) at ($(P2) + (0, -#4)$);
\end{scope}

\begin{scope}[rotate around={#8:(P3)}]
\coordinate  (P4) at ($(P3) + (#2, -#3)$);
\end{scope}

\begin{scope}[rotate around={#8:(P4)}]
\coordinate  (P5) at ($(P4) + (#2, #3)$);
\end{scope}

\begin{scope}[rotate around={#8:(P5)}]
\coordinate (P6) at ($(P5) + (0, #5)$);
\end{scope}

\begin{scope}[rotate around={#8:(P6)}]
\coordinate  (P7) at ($(P6) + (#6, 0)$);
\end{scope}

\begin{scope}[rotate around={#8:(P7)}]
\coordinate  (P8) at ($(P7) + (0, #7)$);
\end{scope}

\begin{scope}[rotate around={#8:(P8)}]
\coordinate  (P9) at ($(P8) + (-#6, 0)$);
\end{scope}

\begin{scope}[rotate around={#8:(P1)}]
\coordinate  (P10) at ($(P1) + (#2, -#3)$);
\end{scope}

\filldraw[color=black!60, fill=black!5, very thick] (P1) -- (P2) -- (P3) -- (P4) -- (P5) -- (P6) -- (P7) -- (P8) -- (P9) -- (P10) -- cycle;
}

\newcommand {\emptytgadgetrightfull} [8] {
\begin{scope}[rotate around={#8:#1}]
\coordinate  (P1) at #1;
\end{scope}

\begin{scope}[rotate around={#8:(P1)}]
\coordinate  (P2) at ($(P1) + (-#2, -#3)$);
\end{scope}

\begin{scope}[rotate around={#8:(P2)}]
\coordinate (P3) at ($(P2) + (0, -#4)$);
\end{scope}

\begin{scope}[rotate around={#8:(P3)}]
\coordinate  (P4) at ($(P3) + (#2, -#3)$);
\end{scope}

\begin{scope}[rotate around={#8:(P4)}]
\coordinate  (P5) at ($(P4) + (#2, #3)$);
\end{scope}

\begin{scope}[rotate around={#8:(P5)}]
\coordinate (P6) at ($(P5) + (0, #5)$);
\end{scope}

\begin{scope}[rotate around={#8:(P6)}]
\coordinate  (P7) at ($(P6) + (#6, 0)$);
\end{scope}

\begin{scope}[rotate around={#8:(P7)}]
\coordinate  (P8) at ($(P7) + (0, #7)$);
\end{scope}

\begin{scope}[rotate around={#8:(P8)}]
\coordinate  (P9) at ($(P8) + (-#6, 0)$);
\end{scope}

\begin{scope}[rotate around={#8:(P1)}]
\coordinate  (P10) at ($(P1) + (#2, -#3)$);
\end{scope}

\draw [line width=0.5mm] (P1) -- (P2) -- (P3) -- (P4) -- (P5) -- (P6) -- (P7) -- (P8) -- (P9) -- (P10) -- cycle;
\draw [line width=1mm] (P1) -- (P2) -- (P3) -- (P4);
}

\newcommand {\tgadgetlefthalf} [9] {
\begin{scope}[rotate around={#8:#1}]
\coordinate  (P1) at #1;
\end{scope}

\begin{scope}[rotate around={#8:(P1)}]
\coordinate  (P2) at ($(P1) + (#2, -#3)$);
\end{scope}

\begin{scope}[rotate around={#8:(P2)}]
\coordinate (P3) at ($(P2) + (0, -#4-#7/2)$);
\end{scope}

\begin{scope}[rotate around={#8:(P3)}]
\coordinate (P4) at ($(P3) + (-#6-#2-#2, 0)$);
\end{scope}

\begin{scope}[rotate around={#8:(P4)}]
\coordinate (P5) at ($(P4) + (0, #7/2)$);
\end{scope}

\begin{scope}[rotate around={#8:(P5)}]
\coordinate (P6) at ($(P5) + (#6, 0)$);
\end{scope}

\begin{scope}[rotate around={#8:(P1)}]
\coordinate (P7) at ($(P1) + (-#2, -#3)$);
\end{scope}

\begin{scope}[rotate around={#8:(P3)}]
\coordinate (P8) at ($(P3) + (0, -#9)$);
\end{scope}

\begin{scope}[rotate around={#8:(P8)}]
\coordinate (P9) at ($(P8) + (0, -#4-#7/2)$);
\end{scope}

\begin{scope}[rotate around={#8:(P9)}]
\coordinate (P10) at ($(P9) + (-#2, -#3)$);
\end{scope}

\begin{scope}[rotate around={#8:(P10)}]
\coordinate (P11) at ($(P10) + (-#2, #3)$);
\end{scope}

\begin{scope}[rotate around={#8:(P11)}]
\coordinate (P12) at ($(P11) + (0, #4)$);
\end{scope}

\begin{scope}[rotate around={#8:(P12)}]
\coordinate (P13) at ($(P12) + (-#6, 0)$);
\end{scope}

\begin{scope}[rotate around={#8:(P13)}]
\coordinate (P14) at ($(P13) + (0, #7/2)$);
\end{scope}

\filldraw[color=black!60, fill=black!5, very thick] (P1) -- (P2) -- (P3) -- (P4) -- (P5) -- (P6) -- (P7) -- cycle;
\filldraw[color=black!60, fill=black!5, very thick] (P8) -- (P9) -- (P10) -- (P11) -- (P12) -- (P13) -- (P14) -- cycle;
}

\newcommand {\tgadgetrighthalf} [9] {
\begin{scope}[rotate around={#8:#1}]
\coordinate  (P1) at #1;
\end{scope}

\begin{scope}[rotate around={#8:(P1)}]
\coordinate  (P2) at ($(P1) + (#2, -#3)$);
\end{scope}

\begin{scope}[rotate around={#8:(P2)}]
\coordinate (P3) at ($(P2) + (0, -#4)$);
\end{scope}

\begin{scope}[rotate around={#8:(P3)}]
\coordinate (P4) at ($(P3) + (#6, 0)$);
\end{scope}

\begin{scope}[rotate around={#8:(P4)}]
\coordinate (P5) at ($(P4) + (0, -#7/2)$);
\end{scope}

\begin{scope}[rotate around={#8:(P5)}]
\coordinate (P6) at ($(P5) + (-#6-#2-#2, 0)$);
\end{scope}

\begin{scope}[rotate around={#8:(P1)}]
\coordinate (P7) at  ($(P1) + (-#2, -#3)$);
\end{scope}

\begin{scope}[rotate around={#8:(P6)}]
\coordinate (P8) at ($(P6) + (0, -#9)$);
\end{scope}

\begin{scope}[rotate around={#8:(P8)}]
\coordinate (P9) at ($(P8) + (0, -#4-#7/2)$);
\end{scope}

\begin{scope}[rotate around={#8:(P9)}]
\coordinate (P10) at ($(P9) + (#2, -#3)$);
\end{scope}

\begin{scope}[rotate around={#8:(P10)}]
\coordinate (P11) at ($(P10) + (#2, #3)$);
\end{scope}

\begin{scope}[rotate around={#8:(P11)}]
\coordinate (P12) at ($(P11) + (0, #4)$);
\end{scope}

\begin{scope}[rotate around={#8:(P12)}]
\coordinate (P13) at ($(P12) + (#6, 0)$);
\end{scope}

\begin{scope}[rotate around={#8:(P13)}]
\coordinate (P14) at ($(P13) + (0, #7/2)$);
\end{scope}

\filldraw[color=black!60, fill=black!5, very thick] (P1) -- (P2) -- (P3) -- (P4) -- (P5) -- (P6) -- (P7) -- cycle;
\filldraw[color=black!60, fill=black!5, very thick] (P8) -- (P9) -- (P10) -- (P11) -- (P12) -- (P13) -- (P14) -- cycle;
}

\newcommand {\pgadget} [5] {
\begin{scope}[rotate around={#5:#1}]
\coordinate  (P1) at #1;
\end{scope}

\begin{scope}[rotate around={#5:(P1)}]
\coordinate  (P2) at ($(P1) + (#2, 0)$);
\end{scope}

\begin{scope}[rotate around={#5:(P2)}]
\coordinate (P3) at ($(P2) + (#4, 0)$);
\end{scope}

\begin{scope}[rotate around={#5:(P3)}]
\coordinate  (P4) at ($(P3) + (#2, 0)$);
\end{scope}

\begin{scope}[rotate around={#5:(P4)}]
\coordinate  (P5) at ($(P4) + (-#2, -#3)$);
\end{scope}

\begin{scope}[rotate around={#5:(P5)}]
\coordinate (P6) at ($(P5) + (-#4, 0)$);
\end{scope}

\filldraw[color=black!60, fill=black!5, very thick] (P1) -- (P2) -- (P3) -- (P4) -- (P5) -- (P6) -- cycle;
}

\newcommand {\cyclegadget}[5] {
\begin{scope}[rotate around={#5:#1}]
\coordinate  (P1) at #1;
\end{scope}

\begin{scope}[rotate around={#5:#1}]
\coordinate  (P2) at #2;
\end{scope}

\begin{scope}[rotate around={#5:#2}]
\coordinate  (P3) at #3;
\end{scope}

\begin{scope}[rotate around={#5:#3}]
\coordinate  (P4) at #4;
\end{scope}

\filldraw[color=black!60, fill=black!5, very thick] (P1) .. controls (P2) and (P3) .. (P4);
}

\newcommand {\tailgadget} [3] {
\coordinate (P1) at #1;
\coordinate (P2) at ($(P1) + (#2, 0)$);
\coordinate (P3) at ($(P2) + (#2*4, #3)$);
\coordinate (P4) at ($(P3) + (#2*3, -#3)$);
\coordinate (P5) at ($(P4) + (-#2, 0)$);
\coordinate (P6) at ($(P5) + (-#2, -#3)$);
\coordinate (P7) at ($(P6) + (-#2, #3)$);
\coordinate (P8) at ($(P7) + (-#2, 0)$);
\coordinate (P9) at ($(P8) + (-#2*2, -#3)$);

\filldraw[color=black!60, fill=black!5, very thick] (P1) -- (P2) -- (P3) -- (P4) -- (P5) -- (P6) -- (P7) -- (P8) -- (P9) -- cycle;
}

\usepackage{thm-restate}
\usepackage{orcidlink}

%% file: introduction.tex
\section{ Introduction}
%\vspace{-0.25cm}
The \textit{Burning Number Problem} (BNP), introduced by Bonato et al.~\cite{Bonato2016Burning}, models the spread of influence, information, or contagion in a network through a discrete-time process on a graph. Starting with all vertices unburned, at each step $t\ge 1$, one new unburned vertex is selected as a \emph{burning source}, while simultaneously the fire spreads from every previously burned vertex to all of its unburned neighbors. Burned vertices remain burned thereafter. The \emph{burning number} of a graph $G$, denoted $b(G)$, is the minimum number of steps required to burn all vertices of $G$.
The corresponding decision problem is the following.

\textbf{\BNP}\\
\textbf{Input:} An undirected graph $G$ and a positive integer $k$.\\
\textbf{Question:} Is $b(G)\leq k$?

Graph burning is closely related to distance-domination and covering processes, but differs from static covering problems in one important aspect: it is inherently \emph{dynamic}. A burning source chosen early acquires larger effective reach in later steps, so both the \emph{location} and the \emph{activation time} of burning sources determine whether the graph can be completely burned within a given number of steps. This interaction between spatial placement and temporal propagation has led to a growing literature on structural, algorithmic, and complexity aspects of the problem.

Bonato et al.~\cite{Bonato2016Burning} proved that every connected graph of order $n$ satisfies
\(
b(G)\le 2\lceil \sqrt{n}\rceil-1,
\)
and conjectured the stronger bound $b(G)\le \lceil \sqrt{n}\rceil$, which is tight for paths and cycles. Since then, burning {of graphs} has been studied extensively on trees and structured graph classes including caterpillars, spiders, grids, interval graphs, proper interval graphs, hypercubes~\cite{Antony2025BurningBoundsHardness,Bessy2018BoundsBurning,Bonato2020BurningSurvey,Bonato2021BurningFence,Bonato2019BurningSpiderPathForest,Das2018BurningSpiders,Mitsche2018Burning}. Polynomial-time algorithms are known for several highly structured classes such as paths, cycles, cographs, split graphs, and complete binary trees~\cite{Bonato2016Burning,DBLP:journals/corr/abs-2308-02825,Kare2019ParamBurningAlgo}.
From the complexity perspective, \BNP\ is NP-complete on general graphs~\cite{Bessy2017BurningIsHard}, and remains NP-complete on several restricted {graph} classes including path forests~\cite{Bessy2017BurningIsHard}, trees of maximum degree three~\cite{Liu2020BurningCaterpillar}, caterpillars of maximum degree three~\cite{Liu2020BurningCaterpillar}, and connected proper interval graphs~\cite{Antony2025BurningBoundsHardness}. Parameterized and approximation aspects have also received considerable attention~\cite{Bonato2019ApprxAlgoBurning,Janssen2020Burning,Kamali2020BurningTwoWorlds,Kobayashi2022ParaComplexBurning,Mondal2022KBurningHardApprx}.

Despite this progress, the complexity of burning under strong degree constraints has remained unclear. In particular, regular graphs form a natural graph class with highly uniform local structure. Such symmetry often simplifies algorithmic problems, but in hardness reductions it removes many local asymmetries---such as leaves or high-degree vertices---that are often used to encode combinatorial choices. This raises the question of whether the \textit{Burning Number Problem} remains hard on regular graphs. We answer this question positively.

\begin{restatable}{theorem}{cubicthm}
\label{thm:cubicNPC}
\BNP\ is NP-complete on connected cubic graphs.
\end{restatable}

We further strengthen this in the approximation setting.

\begin{restatable}{theorem}{apxcubicthm}
\label{thm:cubicAPX}
\BNP\ is APX-hard on connected cubic graphs.
\end{restatable}

Finally, we extend the hardness to all regular graphs of fixed degrees.

\begin{restatable}{theorem}{dregularthm}
\label{thm:regular}
\BNP\ is APX-hard on connected $d$-regular graphs for every fixed integer $d\ge 4$.
\end{restatable}

\textbf{Technical overview. }
Establishing hardness on regular graphs is more delicate than for previously studied sparse graph classes. Regular graphs are highly symmetric and do not admit many local degree-irregular structures often used in reductions, such as pendant vertices, high-degree vertices, or small separators. The restriction to connected graphs is also natural for graph burning. On disconnected graphs, different components may require separate burning sources, and the problem admits partial decomposition across components. This distinction already appears in simple classes: BNP is polynomial-time solvable on paths, but NP-complete on path forests. Thus, hardness on connected regular graphs isolates the intrinsic complexity of the problem from effects caused by disconnectedness.
A further difficulty arises from the temporal aspect of graph burning. The effect of a source depends on both its location and activation time, since a source chosen earlier spreads for more rounds. Consequently, an early source may eventually burn several distant regions, interfering with the intended choices of a reduction. This effect is more pronounced on cubic graphs, where uniform degree provides several alternative propagation paths.

To address these obstacles, our reduction from \MVC\ on connected cubic graphs partitions the constructed graph into two parts. The first is a \emph{core} consisting of edge gadgets corresponding to edges of the input graph. These gadgets enforce vertex-cover constraints: for each edge $(u,v)$ in the input graph, any valid burning sequence must place an early source near one of the vertices representing $u$ or $v$ in the target graph. The second is an \emph{annex} part that regulates the allocation of burning sources. Its special structure delays propagation from the \textit{core} and forces later sources to be used almost entirely within a designated subgraph of the \textit{annex}. Consequently, every feasible (or near-optimal) burning sequence induces a vertex cover whose size equals the number of early sources in the burning.

Building on the cubic hardness construction, we then develop a regularity-lifting framework that transforms a connected cubic graph $G$ into a connected $d$-regular graph $H_d$ for any fixed $d\ge 4$, while preserving the burning number up to an additive constant. This transfers hardness from cubic graphs to all regular graphs of fixed degrees.

\textbf{Relation to prior approximation hardness.}
Mondal et al.~\cite{Mondal2022KBurningHardApprx} established APX-hardness for the \textit{generalized $k$-burning number} problem, which in particular implies APX-hardness of the standard burning number problem on unrestricted graphs (for $k=1$). Our result is substantially stronger: If $G_d$ stands for the class of connected $d$-regular graphs and \textit{Regular Graphs} is a family of graph classes $G_d, \forall d \ge 3$, we show that the standard \textit{Burning Number Problem} remains APX-hard for every single member of the \textit{Regular Graphs} family individually, and this in turn implies the hardness for the entire class of \textit{Regular Graphs}.

%% file: preliminaries.tex
All graphs considered in this paper are finite, simple, and undirected. For a graph $G$, the vertex set and edge set are denoted by $V(G)$ and $E(G)$, respectively. For basic definitions, we refer to \cite{West2001IGT}.  
Let  $u,v \in V(G)$. The vertices $u$ and $v$ are called \emph{neighbors} if there exists an edge $(u,v) \in E(G)$. Any path in $G$ with endpoints $u$ and $v$ is called a \emph{$u$--$v$ path}. The \emph{distance} between vertices $u,v \in V(G)$, denoted by $d_G(u,v)$, is
the number of edges in a shortest $u$--$v$ path in $G$. A \emph{vertex cover} of a graph $G$ is a set of vertices such that every edge of $G$ has at least one endpoint in the set. The \emph{vertex covering number} of $G$, denoted by $\beta(G)$, is the minimum integer $k$ for which $G$ admits a vertex cover of size $k$.

A \emph{binary tree} is a rooted tree in which each node has at most two children. The \emph{level} of a node $v$ in a binary tree, denoted by $l(v)$, is the number of edges on the unique path from the node to the root; the root is at level $0$. The \emph{height} of a node is the number of edges on a longest path from the node to a leaf in the subtree
rooted at that node. The height of a tree is the height of its root. A \emph{perfect binary tree} is a binary tree in which every internal node has exactly two children and all leaves are at the same level. A perfect binary tree with $n$ leaves has $\log_2 n+1$ levels.
A \emph{$d$-regular graph} is a graph in which every vertex has degree $d$. A \emph{cubic graph} is a $3$-regular graph. 

A sequence of vertices $B=(b_1,\dots,b_k)$ of a graph $G$ is a \textit{burning sequence} of $G$ if {for every vertex $v \in V(G)$}, there exists an index $i$ with $1 \le i \le k$ such that $d(v,b_i)\le k-i$, where the vertex $b_i$ is unburned at the end of step $i-1$ and is selected as a new burning source at step $i$. At each step of the process, exactly one new burning source is selected. The \emph{burning number} of $G$, denoted by $b(G)$, is the minimum integer $k$ for which $G$ admits a burning sequence of length $k$. When convenient, we treat the sequence $B$ as a set. If $|B| = b(G)$, then $B$ is called an \emph{optimal burning sequence} of $G$. For $1 \le i \le k$, we denote by $B[i]$ the burning source $b_i$, and for $1 \le i \le j \le k$, by $B[i \dotsc j] = (b_i, b_{i+1}, \dotsc, b_j)$ we denote
the subsequence of $B$ from index $i$ to $j$.

Let  $B = (b_1, b_2, \dots, b_k)$ be a burning sequence of $G$ and $v \in V(G)$ be a vertex burned at step $t$. Define $S_v = \{\, b_i \in B \mid d(b_i, v) = t - i \,\}$, where $d(b_i, v)$ denotes the distance between $b_i$ and $v$ in $G$. Thus, $S_v$ is the set of burning sources from which the fire reaches $v$ precisely at step $t$. A source $s \in B$ is said to \emph{burn} $v$, or to be \emph{responsible for burning} $v$, if $s \in S_v$. 
If $|S_v| = 1$, then $v$ is said to be \emph{burned uniquely}.

For vertices $u, v \in V(G)$, let $Steps(u,v]$ denote the minimum number of steps required for the fire originating at $u$ to burn $v$, excluding the step in which $u$ is burned but including the step in which $v$ is burned. In particular, $Steps(u,v] = d(u,v)$, where $d(u,v)$ denotes the distance between $u$ and $v$ in $G$. Similarly, $Steps[u,v]$ denotes the minimum number of steps required for the fire originating at $u$ to burn $v$, including the step in which $u$ is burned. Hence, $Steps[u,v] = d(u,v) + 1$. Throughout this paper, unless otherwise specified, $\log{(x)}$ denotes the \emph{logarithm} to the base 2, i.e., $\log_2{(x)}$.

%% file: cubic.tex
We obtain the hardness result for the \BNP\ on connected cubic graphs by a polynomial-time reduction from the \MVC\ on connected cubic graphs, which is known to be \NPC~\cite{Garey1974SimNPCProblems}.

\textbf{\MVC\ on connected cubic graphs}\\
\textbf{Input:} An undirected connected cubic graph $G$ and an integer $k$.\\
\textbf{Question:} Does there exist a vertex cover $Q \subseteq V(G)$ of size at most $k$?

We first establish bounds on the vertex covering number of connected cubic graphs.

\begin{observation}
\label{obs:vc-bounds-cubic-graph}
Let $G$ be a connected cubic graph on $n$ vertices. Then
$\frac{n}{2} \leq \beta(G) \leq \left\lceil \frac{2n}{3} \right\rceil$.
\end{observation}

\begin{proof}
Since $G$ is cubic, $|E(G)| = \frac{3n}{2}$. Each vertex can cover at most three edges, so any vertex cover $Q$ satisfies $3|Q| \ge |E(G)| = \frac{3n}{2}$, implying $|Q| \ge \frac{n}{2}$.

By Brooks' Theorem~\cite{West2001IGT}, every connected cubic graph other than $K_4$ is $3$-colorable. Hence, such graphs contain an independent set of size at least $\lceil n/3\rceil$, and therefore $\beta(G) = n - \alpha(G) \le \left\lfloor \frac{2n}{3} \right\rfloor$.
For $K_4$, $\beta(K_4)=3=\lceil 2n/3\rceil$. Hence, the result follows.
\end{proof}

We now outline our approach for proving the hardness of \BNP\ on connected cubic graphs.

\begin{itemize}
\item Starting from a connected cubic graph $G$, we construct an intermediate graph $G'$ by subdividing an arbitrary edge $(p,q)$ of $G$ twice, that is, by inserting two new vertices $x$, $y$ between $p$ and $q$, resulting  in  a path $p-x-y-q$.  Then we show that $\beta(G')=\beta(G)+1$.
\item We then replace each edge of $G'$ with a gadget called the \textit{BTP-gadget}, and add an auxiliary structure between the subdivision vertices $x$ and $y$, resulting in a connected cubic graph $H$.
\item We prove that $b(H)=\beta(G') + c n' + 3$, where $n'=|V(G')|$ and $4 \le c < 8$ is a constant depending on $n'$.
\item This establishes that \BNP\ is \APH\ on connected cubic graphs.
\end{itemize}

We now describe the construction of the intermediate graph $G'$.

\begin{construction}
\label{cons:G'}
Let $G$ be a connected cubic graph on $n$ vertices, and let $(p,q)$ be an arbitrary edge of $G$. Subdivide $(p,q)$ twice, {i.e.}, by inserting two new vertices $x$, $y$ between $p$ and $q$, resulting in a path $p-x-y-q$. The resulting graph is denoted by $G'$ on $n'=n+2$ vertices. See Figure~\ref{fig:k4subdivision} for an example.
\end{construction}

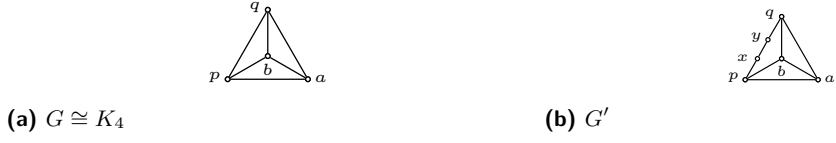
\begin{figure}
    \begin{subfigure}[b]{0.5\linewidth}
    \centering
    \input{figs/k4}
    \caption{$G\cong K_4$}
    \label{fig:k4}
  \end{subfigure}
  \begin{subfigure}[b]{0.454\linewidth}
    \centering
   \input{figs/k4_}
    \caption{$G'$}
    \label{fig:k4'}
  \end{subfigure}
\caption{ An example of construction~\ref{cons:G'} where a  graph $G'$ is obtained from a graph $G\cong K_4$ by subdividing an edge $(p,q)$ in $G$ twice. }
\label{fig:k4subdivision}
\end{figure}

The next lemma relates the vertex covering numbers of $G$ and $G'$.
\begin{lemma}
\label{lem:G'vertexcover}
$\beta(G') = \beta(G) + 1 $.
\end{lemma}

\begin{proof}
The graph $G'$ is obtained from $G$ by deleting the edge $(p,q)$ and adding the edges $(p,x)$, $(x,y)$, and $(y,q)$.

\smallskip
\noindent\textbf{($\implies$)} 
Let $Q$ be a minimum vertex cover of $G$. Since $(p,q)\in E(G)$, at least one of the vertices $p$ or $q$ belongs to $Q$.
Define
\[
Q' :=
\begin{cases}
Q \cup \{y\}, & \text{if } p \in Q,\\
Q \cup \{x\}, & \text{if } p \notin Q.
\end{cases}
\]
One can verify that $Q'$ covers all edges of $G'$.
Hence, $
\beta(G') \le |Q'| = |Q| + 1 = \beta(G) + 1$.

\noindent\textbf{($\impliedby$)} 
Let $Q'$ be a minimum vertex cover of $G'$. Consider the path $p-x-y-q$. Since this path contains three edges, at least two vertices from the set $\{p,x,y,q\}$ must belong to $Q'$. Moreover, minimality of $Q'$ implies that the only possible vertex cover configurations on this path are
$\{x,y\}, \{p,y\}, \text{or},  \{x,q\}.$

In the first case, replacing $\{x,y\}$ by $\{p\}$ or $\{q\}$ yields a vertex cover of $G$. In the second case, removing $y$ yields a vertex cover of $G$. In the third case, removing $x$ yields a vertex cover of $G$. In all cases, we obtain a vertex cover of $G$ of size $|Q'|-1$. Therefore, $\beta(G) \le |Q'| - 1,$
which implies $\beta(G') \ge \beta(G) + 1$.

Combining both inequalities, we obtain $\beta(G') = \beta(G) + 1.$
\end{proof}
Lemma~\ref{lem:G'vertexcover} allows us to translate the bounds of Observation~\ref{obs:vc-bounds-cubic-graph} to the intermediate graph $G'$, yielding the following upper bound on $\beta(G')$.
\begin{observation}
\label{obs:upperbound-on-k'}
The vertex covering number of $G'$, $\beta(G') \le \left\lceil \frac{2n'}{3} \right\rceil$ .
\end{observation}

\begin{proof}
By Observation~\ref{obs:vc-bounds-cubic-graph}, $\beta(G)\le \lceil 2n/3\rceil$. Since $n'=n+2$ and $\beta(G')=\beta(G)+1$ (Lemma~\ref{lem:G'vertexcover}), we obtain
\[
\beta(G') \le \left\lceil \frac{2n}{3} \right\rceil + 1
\le \left\lceil \frac{2(n+2)}{3} \right\rceil = \left\lceil \frac{2n'}{3} \right\rceil.
\]
\end{proof}

\subsection{Gadgets}
\label{subsec:gadgets}

We now describe the gadgets used to construct the connected cubic graph $H$ from $G'$. Each gadget contains one or more designated \emph{end vertices}, through which it is connected to the rest of the graph. In every gadget, each end vertex has degree at most two, whereas all other vertices have degree exactly three.

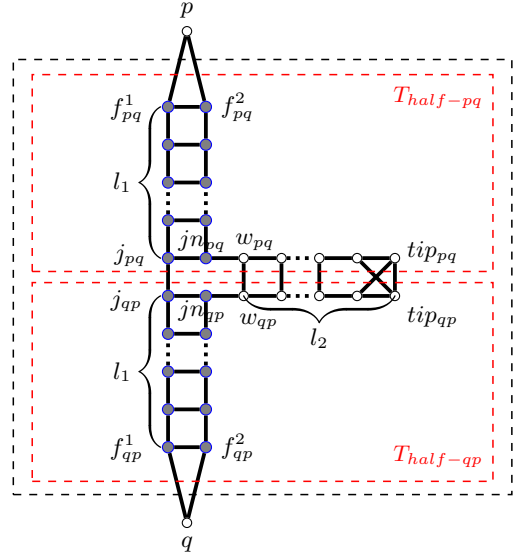
\begin{figure}[t]
\input{figs/t-gadget-vertical}
\caption{A T-gadget, $T_{pq}(l_1, l_2)$-gadget with two end vertices $p$ and $q$ where filled vertices form fixed arm and hollow vertices form floating arm. }
\label{fig:t-gadget}
\end {figure}

\textbf{1. \textit{T}-gadget}

A \emph{$T$-gadget} consists of two arms: a \emph{fixed arm} formed by the filled vertices and a \emph{floating arm} formed by the hollow vertices, as shown in Figure~\ref{fig:t-gadget}.  
The gadget has two external vertices $p$ and $q$, called the \emph{hook vertices} and four end vertices $f^1_{pq}, f^2_{pq}, f^1_{qp}$ and $f^2_{qp}$. The vertices $\mathit{tip}_{pq}$ and $\mathit{tip}_{qp}$ are referred to as the \emph{tips} of the gadget. Let $T_{{half}\text{-}pq}$ (respectively, $T_{{half}\text{-}qp}$) denote the subgraph induced by the vertices on the fixed arm starting from $p$ (respectively $q$) to $\mathit{jn}_{pq}$ (respectively $\mathit{jn}_{qp}$), together with the $l_2$ vertices on the floating arm from $\mathit{jn}_{pq}$ (respectively $\mathit{jn}_{qp}$) to $\mathit{tip}_{pq}$ (respectively $\mathit{tip}_{qp}$), as shown in Figure~\ref{fig:t-gadget}.  A $T$-gadget fixed between hook vertices $p$ and $q$ with parameters $l_1$ and $l_2$ is denoted by $T_{pq}(l_1,l_2)$.

\begin{observation}
\label{obs:T-gadget-properties}
Let $T_{pq}(l_1,l_2)$ be a $T$-gadget. Then the following are true:

\begin{enumerate}[label=\roman*.]
    \item 
    $
    {Steps}[p,q] = 2l_1+2.
    $
    \item 
    $
    {Steps}[p, tip_{pq}] = {Steps}[q, tip_{qp}] = l_1+l_2+1.
    $
    \item The total number of vertices in $T_{pq}(l_1,l_2)$ is $4l_1+2l_2$.
\end{enumerate}
\end{observation}

\textbf{2. \textit{BT}-gadget}

A \emph{$BT$-gadget} is a perfect binary tree of height $h$, denoted by $BT(h)$, as shown in Figure~\ref{fig:bt-gadget}. The root and all the leaves are end vertices in this gadget.  

\begin{figure}[h]
\input{figs/binary-tree-gadget}
\caption{$BT$-gadget which which contains the vertices inside the dashed box.}
\label{fig:bt-gadget}
\end {figure}
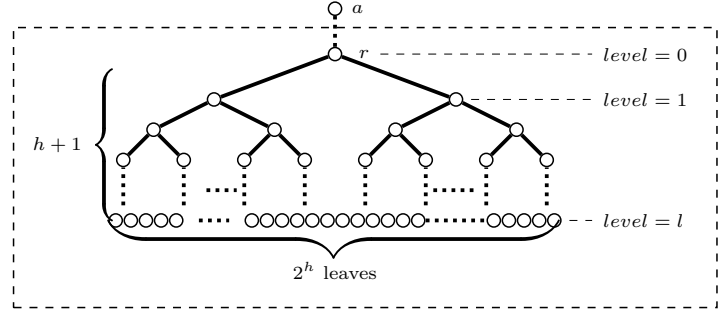
\begin{observation}
\label{obs:BT-gadget-properties}
Let $BT(h)$ be a $BT$-gadget.
\begin{enumerate}[label=\alph*.]
    \item Fire initiated at the root $r$ completely burns $BT(h)$ in exactly $h+1$ steps.
    \item The gadget $BT(h)$ has $2^{h+1}-1$ vertices,  including $2^h$ leaves.
\end{enumerate}
\end{observation}

\textbf{3. \textit{BTP}-gadget}

For each edge $(a,b)\in E(G')$, we plan to introduce a $BTP$-gadget denoted by $BTP_{ab}(h,l_1,l_2)$-gadget as shown in Figure~\ref{fig:bt-pair-gadget} for the construction of $H$. It consists of two $BT(h)$-gadgets, denoted by $BT_{ab}$ and $BT_{ba}$, with roots  $r_{ab}$ and $r_{ba}$, respectively, that are attached to $a$ and $b$ respectively. For each $1\le i\le 2^h$, a $T$-gadget $T^i_{ab}(l_1,l_2)$ connects the $i^{th}$ leaf of $BT_{ab}$ to the $i^{th}$ leaf of $BT_{ba}$. The roots $r_{ab}$ and $r_{ba}$ of $BT_{ab}$ and $BT_{ba}$, respectively, are the end vertices of the gadget. Let $Tips_{ab}$ (respectively $Tips_{ba}$) denote the set of all tips of $T^i_{half\text{-}ab}$s (respectively $T^i_{half\text{-}ba}$s)  attached to the
leaves of $BT_{ab}$ (respectively $BT_{ba}$). Note that $BTP_{ab}$ and $BTP_{ba}$ are the same gadgets.

The \emph{$a_{half}$} of the gadget, denoted $a_{{half}}(BTP_{ab})$, is an induced subgraph containing the tree $BT_{ab}$, and all $T^i_{half-ab}$s attached to the  leaves of $BT_{ab}$. The $b$-half is defined
symmetrically. The $a_{{exthalf}}(BTP_{ab})$ is an induced subgraph with the vertex $a$ and vertices of $a_{{half}}(BTP_{ab})$.  All notations associated with leaves, tips, halves, and extended halves follow Figure~\ref{fig:bt-pair-gadget}.

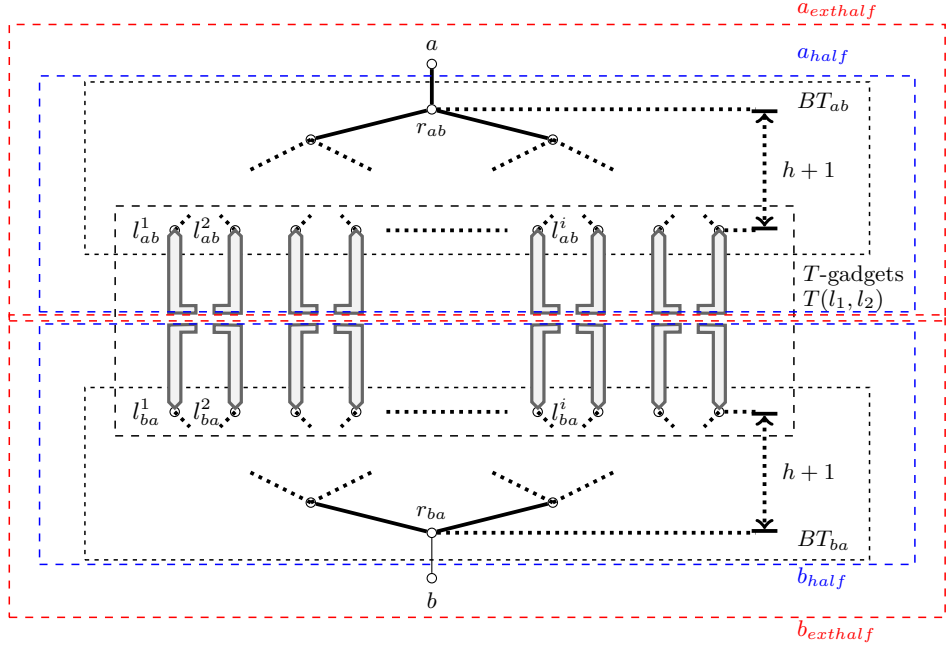
\begin{figure}[h]
\input{figs/BT-pair-gadget}
\caption{$BTP$-gadget corresponding to an  edge $(a,b)\in E(G')$ with two $BT$-gadgets $BT_{ab}$ and $BT_{ba}$ included in the upper and lower black dashed rectangle.}
\label{fig:bt-pair-gadget}
\end {figure}

Throughout the construction, the following inequalities are satisfied:
\begin{align}
l_1 + l_2 &< 2^{h-2} \label{eq:param-1}\\
l_2 &> l_1 + h + 1. \label{eq:param-2}
\end{align}

Let $BTP_{ab}(h, l_1, l_2)$ be a $BTP$-gadget with end vertices $r_{ab}$ and $r_{ba}$, as shown in Figure~\ref{fig:bt-pair-gadget}.

\begin{observation}
\label{obs:btp-gadget-num-vertices}
The number of vertices in a $BTP_{ab}(h, l_1, l_2)$-gadget is $2(2^{h+1}-1) + 2^h(4l_1+2l_2)$.
\end{observation}

\begin{proof}
By Observation~\ref{obs:BT-gadget-properties}, each $BT(h)$ contains $2^{h+1}-1$ vertices. Hence, the two binary trees together contribute $2(2^{h+1}-1)$ vertices. By Observation~\ref{obs:T-gadget-properties}, each $T(l_1,l_2)$-gadget contains $4l_1+2l_2$ vertices, and there are exactly $2^h$ such gadgets. Hence, the observation follows.
\end{proof}

We now analyze the propagation of fire inside a $BTP$-gadget.

\begin{observation}
\label{obs:btp-distances}
If a fire is initiated at $r_{ab}$, then
\begin{align}
Steps[r_{ab}, r_{ba}] &= 2h + 2l_1 + 2,\\
Steps[r_{ab}, Tips_{ab}] &= h+1+l_1+l_2.
\end{align}
\end{observation}
\begin{proof}
The fire initiated at $r_{ab}$ must traverse the binary tree $BT_{ab}$ to a leaf in $h+1$ steps and then pass through a $T(l_1,l_2)$-gadget to reach the vertex $r_{ba}$ as well as $Tips_{ab}\cup Tips_{ba}$. Hence, the stated bounds follow from Observation~\ref{obs:T-gadget-properties} and Observation~\ref{obs:BT-gadget-properties}.
\end{proof}

The next observation combines the consequences of Observation~\ref{obs:btp-distances} with the inequality~\ref{eq:param-2}.

\begin{observation}
\label{obs:btp-complete-burning}
Let $BTP_{ab}(h, l_1, l_2)$ be a $BTP$-gadget with end vertices $r_{ab}$ and $r_{ba}$. If a fire is initiated at vertex $a$ (or $b$), then, without using any additional burning sources, the number of steps required to completely burn $\{a,b\} \cup BTP_{ab}(h,l_1,l_2)$ (including $a$ (or $b$)) is exactly $h + l_1 + l_2 + 2$.
\end{observation}

\begin{proof}
Suppose a fire is initiated at $r_{ab}$. Due to the inequality~\ref{eq:param-2} and by Observation~\ref{obs:btp-distances}, we get $Steps[r_{ab},r_{ba}]<h+1+l_1+l_2$  and $Steps[r_{ab}, Tips_{ab}] = h+1+l_1+l_2$. Then it follows that if a fire is initiated at vertex $a$ (or $b$), then, without using any additional burning sources, the number of steps required to completely burn 
$\{a,b\} \cup BTP_{ab}(h,l_1,l_2)$ (including $a$ (or $b$)) is exactly $h + l_1 + l_2 + 2$.
\end{proof}

Based on the above observations, we have a lower bound on the number of burning sources required to completely burn a $BTP$-gadget.

\begin{lemma}
\label{lem:min-len-seq}
Let $B$ be a sequence of burning sources that is used to burn $BTP_{ab}(h,l_1,l_2)$. If $B$ completely burns $BTP_{ab}(h,l_1,l_2)$, then $|B|>l_1+l_2.$
\end{lemma}

\begin{proof}
Assume for  contradiction that $|B|\le l_1+l_2$. Let $Tips_{ab}$ denote the set of all tips in the $T^i_{half\text{-}ab}$s attached to the $2^h$ leaves of $BT_{ab}$. Any source placed in $V(BT_{ab})\cup V(BT_{ba})$ cannot reach any vertex of $Tips_{ab}$ within $l_1+l_2$ steps. Moreover, for distinct indices $i\neq j$, a source placed in $T^i_{ab}$ cannot reach $tip^j_{ab}$ within $l_1+l_2$ steps. Thus, each burning source can burn at most one vertex of $Tips_{ab}$. By Inequality~\ref{eq:param-1}, $l_1+l_2<2^{h-2}$, while $|Tips_{ab}|=2^h$, so not all tips can be burned, a contradiction. Hence, $|B|>l_1+l_2$.
\end{proof}

Let $B$ be a sequence that completely burns $BTP_{ab}(h, l_1, l_2)$. Let $\textit{EndBlock}$ be the subsequence of the last $l_1 + l_2$ sources of $B$. By Lemma~\ref{lem:min-len-seq}, $|B| > l_1 + l_2$. So the remaining $|B|-(l_1+l_2)$ sources are partitioned (as required) into two subsequences $\textit{StartBlock}$ and $\textit{MiddleBlock}$ with sizes $s,m \ge 0$, respectively. Thus $B$ is a concatenation of sources as follows:
\[
B=\textit{StartBlock}\circ \textit{MiddleBlock}\circ \textit{EndBlock}.
\] and $|B| = s + m + (l_1 + l_2)$, where
$\textit{StartBlock} = B[1,\ldots,s]$, $\textit{MiddleBlock} = B[s+1,\ldots,
s+m]$, and $\textit{EndBlock} = B[s+m+1,\ldots,s+m+l_1+l_2]$. Note that, $StartBlock$ ($MiddleBlock$) is empty if $s=0$ ($m=0$, respectively). The following lemma is due to Lemma~\ref{lem:min-len-seq}.

\begin{lemma}
\label{lem:max-leaves-in-endblock}
The maximum number of vertices $tip_{ab}\in Tips_{ab}$ that can be burned using the sources in
$\textit{EndBlock}$ of a burning sequence $B$ is at most $l_1+l_2$.
\end{lemma}

\begin{lemma}
\label{lem:generalized-max-leaves-in-endblock}
Let $BTP_{ab}(h,l_1,l_2)$ and $BTP_{cd}(h,l_1,l_2)$ be two vertex-disjoint BTP gadgets. The maximum number of vertices in $Tips_{ab}\cup Tips_{cd}$ that can be burned using sources in $\textit{EndBlock}$ is at most $l_1+l_2$. 
\end{lemma}

\begin{proof}
By Lemma~\ref{lem:max-leaves-in-endblock}, at most $l_1+l_2$ vertices of $Tips_{ab}$ (and similarly of $Tips_{cd}$) can be burned by $\textit{EndBlock}$. Moreover, for any $tip_{ab}\in Tips_{ab}$ and $tip_{cd}\in Tips_{cd}$, we have $d(tip_{ab},tip_{cd})>2(l_1+l_2)$. Hence, a single source in $\textit{EndBlock}$ cannot burn both. Therefore, the total number of burned tips in $Tips_{ab}\cup Tips_{cd}$
is at most $l_1+l_2$.
\end{proof}
We next bound the contribution of each source in $\textit{MiddleBlock}$ toward burning tip vertices.
\begin{lemma}
\label{lem:max-tips-from-source-in-middleblock}
Let $0\le m\le h$. For each $1\le j\le m$, a source $\textit{MiddleBlock}[j]$ can burn at most $2^{m-j}$
vertices in $Tips_{ab}$. This bound is tight and is achieved only when the source is placed at a vertex
at level $h-m+j$ in $BT_{ab}$ or $BT_{ba}$.
\end{lemma}

\begin{proof}
A source in \textit{MiddleBlock} at position $j$ has $(m-j)+(l_1+l_2)$ steps for fire propagation. Suppose the source is placed at a vertex $u$ at level $i$ in $BT_{ab}$. Any tip in the subtree rooted at $u$ is at distance $(h-i)+(l_1+l_2)$ from $u$. If $h-i > m-j$, then $(h-i)+(l_1+l_2) > (m-j)+(l_1+l_2)$, and hence, no tip can be burned. If $h-i = m-j$, then all tips in the subtree rooted at $u$ are burned. Since a node at level $i$ has $2^{h-i}$ descendant leaves, the number of burned tips is $2^{h-i} = 2^{m-j}$. Now consider the case $h-i < m-j$. Let $(m-j) = (h-i) + 2p$ for some integer $p > 0$. Then every tip in the subtree rooted at $u$ is at distance $(h-i)+(l_1+l_2) = (m-j)-2p+(l_1+l_2) < (m-j)+(l_1+l_2)$, and hence all such tips are burned. Next, consider tips that are not in the subtree rooted at $u$. Any such tip must be reached via a path that goes from $u$ to an ancestor at distance $x$ and then to the tip. The length of such a path is $2x + (h-i) + (l_1+l_2)$. This is at most $(m-j)+(l_1+l_2)$ only if $2x + (h-i) \le (m-j)$, which implies $x \le p$. Thus, the set of reachable tips is contained in the subtree rooted at the $p$-th ancestor of $u$. Let $w$ denote this ancestor. Let $v$ be the vertex at level $h-m+j$. Then $w$ is a descendant of $v$, and hence the subtree rooted at $w$ is a proper subset of the subtree rooted at $v$. Therefore, placing the source at $v$ maximizes the number of reachable tips. Since a vertex at level $h-m+j$ has exactly $2^{m-j}$ descendant leaves, the maximum number of tips that can be burned is $2^{m-j}$, and this bound is tight.
Finally, placing a source inside any $T$-gadget is suboptimal. The fire from a source located in a $T$-gadget reaches the corresponding $tip_{ab}$ strictly before the last step, and thereafter any further spread toward other tips must proceed through the attached $BT$-gadget. However, to burn other tips, the fire will have to spread in the $BT$-gadget, and by the reasoning given above, we would be better off to have a burning source in the $BT$-gadget itself. This completes the proof.
\end{proof}

\begin{lemma}
\label{lem:max-leaves-in-middleblock}
If $0\le m\le h$, then the maximum number of vertices in $Tips_{ab}$ that can be burned using all sources in $\textit{MiddleBlock}$ is at most $2^m-1$.
\end{lemma}

\begin{proof}
By Lemma~\ref{lem:max-tips-from-source-in-middleblock}, $\textit{MiddleBlock}[j]$ burns at most $2^{m-j}$ tips. Summing over all $j$,
\[
\sum_{j=1}^m 2^{m-j}=2^m-1.
\]
\end{proof}

The following lemma shows that, under the stated conditions, a burning sequence cannot completely burn three $BTP$-gadgets simultaneously.

\begin{lemma}
\label{lem:possible-to-burn-two-not-three}
Let $BTP_{ab}(h, l_1, l_2)$, $BTP_{cd}(h, l_1, l_2)$, and $BTP_{ef}(h, l_1, l_2)$ be three pairwise disjoint BTP-gadgets. Suppose $StartBlock \cap (V(BTP_{ab}) \cup V(BTP_{cd}) \cup V(BTP_{ef})) = \emptyset$, and the fire from any source in \textit{StartBlock} does not reach any vertex $v \in \{r_{ab}, r_{ba}, r_{cd}, r_{dc}, r_{ef}, r_{fe}\}$ by the end of the $(s+1)^{{th}}$ step. If $m = h+1$, then the burning sequence $B$ cannot completely burn $a_{{half}}(BTP_{ab}) \cup c_{{half}}(BTP_{cd}) \cup e_{{half}}(BTP_{ef})$.
\end{lemma}

\begin{proof}
Each gadget $BTP_{xy}$ contains exactly $2^h$ tip vertices in $x_{\mathrm{half}}(BTP_{xy})$. Since the three gadgets are pairwise disjoint, these tip sets are disjoint. Hence the total number of tips in
$a_{\mathrm{half}}(BTP_{ab})
\cup
c_{\mathrm{half}}(BTP_{cd})
\cup
e_{\mathrm{half}}(BTP_{ef})$
is $3\cdot 2^h$.
By assumption, no fire from $StartBlock$ reaches any of the six root vertices by time $s+1$. Therefore all tips in $a_{\mathrm{half}}(BTP_{ab})
\cup
c_{\mathrm{half}}(BTP_{cd})
\cup
e_{\mathrm{half}}(BTP_{ef})$ must be burned by sources in $MiddleBlock \cup EndBlock$.

Let $t_1,t_2,t_3$ denote the number of sources from $MiddleBlock$ placed in $BTP_{ab}$, $BTP_{cd}$, and $BTP_{ef}$, respectively. Since $m=h+1$, we have $t_1+t_2+t_3=h+1$. By Lemma~\ref{lem:max-leaves-in-middleblock}, the number of tips that can be burned in $BTP_{ab}$ by $t_1$ sources is at most $2^{t_1}-1$, and similarly for the other two gadgets. Hence the total number of tips burned by $MiddleBlock$ is at most
$(2^{t_1}-1)+(2^{t_2}-1)+(2^{t_3}-1) = 2^{t_1}+2^{t_2}+2^{t_3}-3$. Subject to $t_1+t_2+t_3=h+1$ and $t_i\ge 0$, the expression $2^{t_1}+2^{t_2}+2^{t_3}$ is maximized when one variable equals $h+1$ and the others equal $0$. Thus, $2^{t_1}+2^{t_2}+2^{t_3} \le 2^{h+1}+2$, and therefore $2^{t_1}+2^{t_2}+2^{t_3}-3 \le
2^{h+1}-1$.

One source in $MiddleBlock \cup EndBlock$ cannot burn $tips$ in two disjoint $BTP$-gadgets simultaneously as $d(tip_{ab}, tip_{cd}) > 2(h+1+l_1+l_2)$. By Lemma~\ref{lem:generalized-max-leaves-in-endblock}, the sources in $EndBlock$ can burn at most $l_1+l_2$ additional tips across the three gadgets.

By Inequality~\ref{eq:param-1}, $l_1+l_2<2^{h-2}.$ Consequently, the total number of tips that can be burned by $MiddleBlock \cup EndBlock$ is strictly less than $(2^{h+1}-1)+2^{h-2} = 2^{h+1}+2^{h-2}-1.$
Since $2^{h+1}+2^{h-2} = \frac{9}{4}\,2^h < 3\cdot 2^h$, it follows that $2^{h+1}+2^{h-2}-1 < 3\cdot 2^h$.
Hence, strictly fewer than $3\cdot 2^h$ tips can be burned, whereas $a_{\mathrm{half}}(BTP_{ab}) \cup c_{\mathrm{half}}(BTP_{cd}) \cup e_{\mathrm{half}}(BTP_{ef})$  together contain exactly $3\cdot 2^h$ tips. This is a contradiction.
\end{proof}

\textbf{4. \textit{P}-gadget}
\label{sec:p-gadget}

A \emph{$P$-gadget} consists of two parallel branches, each branch being a path: The \emph{major branch} has $d$ vertices $a_1, \dots, a_d$, where $a_1$ and $a_d$ are the \emph{end vertices}. The \emph{minor branch} has $d-2$ vertices denoted by $b_2, \dots, b_{d-1}$, with each $b_i$ adjacent to the corresponding vertex $a_i$ of the major branch ($2 \le i \le d-1$). We denote a $P$-gadget with $d$ vertices on its major branch as $P(d)$. Let ${Major}(P) = \{a_1, \dots, a_d\}, \quad \text{Minor}(P) = \{b_2, \dots, b_{d-1}\}$. Let $d = 2l-1$. Then $a_l$ is the \emph{middle vertex} of the major branch. 
\begin{figure}[h]
\centering
\input{figs/path-gadget}
\caption{$P$-gadget with major and minor branches. The filled vertices form major branch and hollow vertices form the minor branch.}
\label{fig:p-gadget}
\end{figure}
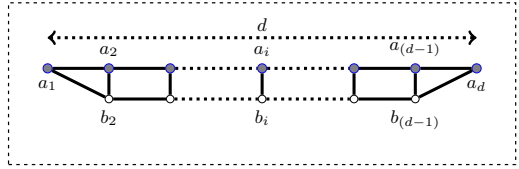

The following observations are immediate from the structure of a $P$-gadget.

\begin{observation}
\label{obs:pgadget-burntime}
If fire is initiated at $a_l$ or $b_l$, then all vertices of $P(d)$ are burned in exactly $l$ steps.
\end{observation}

\begin{observation}
\label{obs:pgadget-distances}
For any $2 \le i \le d-1$ and $1 \le j \le d$, let $a_i, a_j \in \text{Major}(P)$ and let $b_i \in \text{Minor}(P)$ be the vertex corresponding to $a_i$. Then
$ d(a_i, a_j) \le d(b_i, a_j)$.
\end{observation}

\begin{observation}
\label{obs:pgadget-num-vertices}
The total number of vertices in $P(d)$ is $2d-2$.
\end{observation}

\textbf{5. \textit{Y}-gadget}
\label{sec:y-gadget}

A \emph{$Y$-gadget} consists of three $P$-gadgets: $P_x(d_1), P_y(d_1)$, and $P_z(d_2)$, connected to a central vertex $z$, forming a \textbf{Y}-shape as shown in Figure~\ref{fig:Y-gadget}. The end vertices of a $Y$-gadget are $x_a, y_b$ and $z_b$. The gadget is inserted between the subdivision vertices $x$ and $y$ in $G'$ as follows:  
\begin{itemize}
    \item $x$ connects to $x_a$, one end of $P_x$; $y$ connects to $y_a$, one end of $P_y$.  
    \item The other ends, $x_b$ and $y_b$ of $P_x$ and $P_y$ respectively, are connected to $z$.  
    \item $z$ is connected to $z_a$, one end of $P_z$, whose other end $z_b$ serves as the end vertex of the $Y$-gadget.  
\end{itemize}

Here, $d_1 = | \text{Major}(P_x) | = | \text{Major}(P_y) |$ and $d_2 = | \text{Major}(P_z) |$. We denote this gadget as $Y(d_1, d_2)$.  
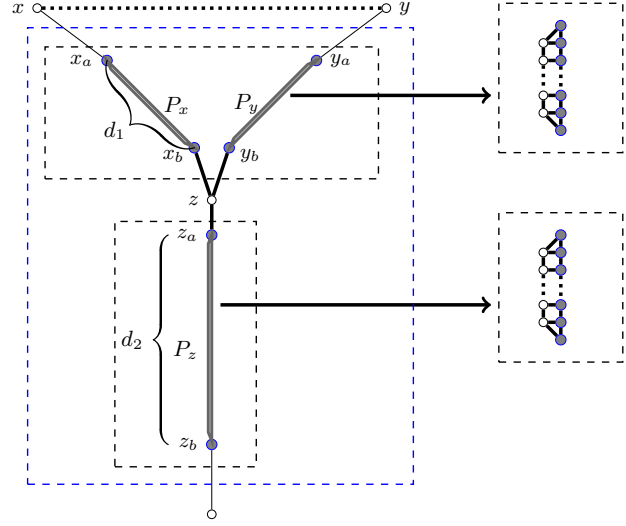
\begin{figure}[h]
\centering
\input{figs/Y-gadget}
\caption{$Y$-gadget composed of three $P$-gadgets connected to a central vertex $z$.}
\label{fig:Y-gadget}
\end{figure}
\begin{observation}
\label{obs:y-gadget-num-vertices} 
The total number of vertices in $Y(d_1, d_2)$ is 
\[
2(2d_1-2) + (2d_2-2) + 1 = 4d_1 + 2d_2 - 5.
\]
\end{observation}

\begin{observation}
\label{obs:y-gadget-burning}
If the fire starts at $s \in \{x_a, y_a\}$, the $Y(d_1, d_2)$-gadget is completely burned in 
$\max\{ 2d_1 + 1, \; d_1 + d_2 + 1 \} \text{ steps.}$
\end{observation}

\textbf{6. \textit{Tail}-gadget}
\label{sec:tail-gadget}

The \emph{$Tail$-gadget} has a \emph{spine}, which is a path on 9 vertices $(v_1, \dots, v_9)$, partitioned as:  
\begin{itemize*}
    \item $PT_1$: single vertex $v_1$  
    \item $PT_2$: path $v_2, v_3, v_4$  
    \item $PT_3$: path $v_5, \dots, v_9$.   
\end{itemize*}  Additional vertices $p,q,r$
are attached to selected spine vertices as shown in Figure~\ref{fig:tail-gadget}, without modifying the spine structure. The vertices $v_1$ and $v_9$ are the \emph{end vertices} of the gadget.  
\begin{figure}[h]
\centering
\input{figs/tail-gadget}
\caption{$Tail$-gadget consisting of a spine and three additional vertices $p,q,$ and $r$.}
\label{fig:tail-gadget}
\end{figure}
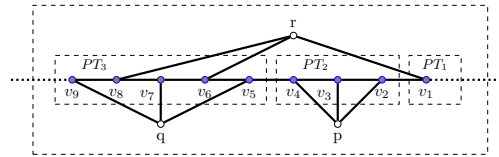

\begin{observation}
A $Tail$-gadget contains exactly 12 vertices.
\end{observation}

\textbf{7. \textit{C}-gadget}
\label{sec:c-gadget}

A \emph{$C$-gadget}, denoted by $C(m)$, is constructed by concatenating $(m-3)$ $P$-gadgets $P_m, P_{m-1}, \dots, P_4$ followed by a $Tail$-gadget, as shown in Figure~\ref{fig:C-gadget-2}. For $4 \le i \le m$, let $p_{i_1}$ and $p_{i_2}$ denote the two end vertices of $P_i$. For $4 \le i < m$, the vertex $p_{i_1}$ is adjacent to $p_{(i+1)_2}$. The vertex $p_{4_2}$ is adjacent to the vertex $v_9$ of the $Tail$-gadget.

For $4 \le i < m$, the parameter of $P_i$ is $d_i = 2i-1$, while $P_m$ has $d_m = 2m-2$. A new vertex $v_{m^2}$ is added and made adjacent to both $p_{m_1}$ and $v_1$ (the right end of the $Tail$-gadget). The vertex $v_{m^2}$ serves as the unique end vertex of $C(m)$ and is used to attach external structures.

Consider the major branches of all $P$-gadgets, the three paths $PT_1, PT_2, PT_3$ of the $Tail$-gadget, and the vertex $v_{m^2}$. The total number of vertices in these subgraphs is
\[
1 + (2m-2) + (2m-3) + (2m-5) + \cdots + 3 + 1 = m^2.
\]
These vertices induce a cycle of length $m^2$,  and is denoted by ${trunk}(C)$. All vertices on ${trunk}(C)$, except $v_{m^2}$, have degree $3$.
\begin{figure}[h]
\centering
\input{figs/cycle-gadget-new-2}
\caption{Structure of the $C$-gadget on $m^2$ vertices. }
\label{fig:C-gadget-2}
\end{figure}
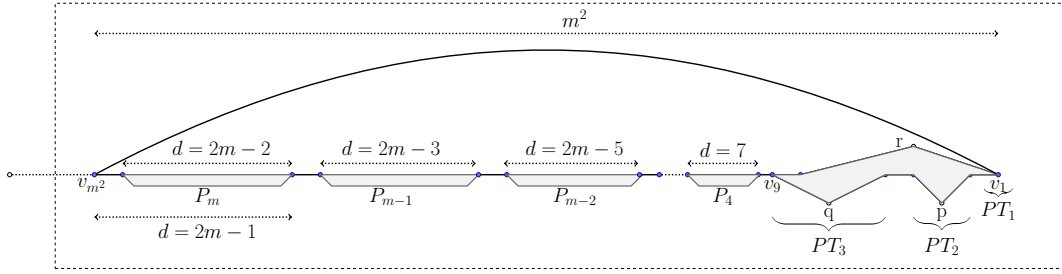
\begin{observation}
\label{obs:C-size}
A $C(m)$-gadget has $2m^2 - 2m - 1$ vertices.
\end{observation}

\begin{proof}
For $4 \le i < m$, each $P_i$ has $4i-4$ vertices, while $P_m$ has $4m-6$ vertices. Thus, all $P$-gadgets contribute $2m^2 - 2m - 14$ vertices. Adding $12$ vertices from the $Tail$-gadget and the vertex $v_{m^2}$ yields $2m^2 - 2m - 1$ vertices.
\end{proof}

Let ${trunk}'(C)$ denote the shortest cycle in $C(m)$ containing both $v_1$ and $v_{m^2}$.

\begin{observation}
\label{obs:C-short-cycle}
${trunk}'(C)$ is the shortest induced cycle in $C(m)$ containing $v_1$ and $v_{m^2}$ and has length $m^2 - 5$.
\end{observation}
Proof: Let ${trunk}'(C)$ denote the shortest cycle in $C(m)$ containing both $v_1$ and $v_{m^2}$. It is obtained by taking the path $v_9-v_8-r-v_1$ instead of $v_9-v_8-v_7\ldots$  to $v_1$. Thus, the path $v_9-v_8 \ldots v_1$ of length 9 is shortened by the path $v_9-v_8-r-v_1$ of length 4. Therefore, the length of $trunk'(C)$ is $m^2-5$.

\begin{theorem}[\cite{Bonato2016Burning}]
\label{thm:path-cycle-burning}
If $G = P_n$ or $G = C_n$, then $b(G) = \lceil \sqrt{n}~ \rceil$.
\end{theorem}
Burning a cycle of size $m^2$ using $m$ sources burns each vertex uniquely due to Theorem~\ref{thm:path-cycle-burning}. Since ${trunk}'(C)$ has $m^2 - 5$ vertices, at most $5$ vertices can be burned by more than one source.

\begin{definition}
\label{def:fixed-overlap}
The maximum number of vertices in ${trunk}'(C)$ burned by more than one source in a burning sequence of length $m$ is called {FixedOverlap}. In particular, $\textit{FixedOverlap} = 5$.
\end{definition}

\begin{remark}
\label{rem:C-trunk-suffices}
Since $b(C(m)) = m$ and ${trunk}'(C)$ is a cycle of size $m^2 - \textit{FixedOverlap}$, it suffices to analyze ${trunk}'(C)$ when ruling out candidate burning sequences for graphs containing $C(m)$.
\end{remark}

We now show that the burning number of $C(m)$ is exactly $m$.

\begin{lemma}
\label{lem:C-gadget-burning-upperbound}
$b(C(m)) \le m$.
\end{lemma}
\begin{proof}
Let
$B = (a_m, a_{m-1}, \dots, a_4, v_7, v_3, v_1)$, where $a_i$ is the middle vertex of $Major(P_i)$ for $4 \le i \le m$, and $v_7, v_3, v_1$ are the middle vertices of the paths $PT_3, PT_2$, and $PT_1$ of the $Tail$-gadget, respectively. Clearly, $|B| = m$.

Each source $a_i$ is placed at time step $(m-i+1)$ and hence has $(i-1)$ additional steps to propagate fire, which is sufficient to completely burn the corresponding $P_i$-gadget (Observation~\ref{obs:pgadget-burntime}) by the end of the $m^{\text{th}}$ step. Similarly, the sources $v_7, v_3$, and $v_1$ are placed so that the entire $Tail$-gadget is burned by the end of the $m^{\text{th}}$ step.

Since the $P$-gadgets and the $Tail$-gadget together cover all vertices of $C(m)$, the sequence $B$ burns $C(m)$ completely in $m$ steps. Therefore, $b(C(m)) \le m$.
\end{proof}

Now we show that $b(C(m)) \ge m$.

\begin{lemma}
\label{lem:C-gadget-burning-lowerbound}
$b(C(m)) \ge m$.
\end{lemma}

\begin{proof}
Suppose, for the sake of contradiction, that $C(m)$ admits a burning sequence $B = (u_1, u_2, \dots, u_k)$ with $k < m$. Each source in $B$ may lie either on a major or minor branch of some $P$-gadget,
or in the $Tail$-gadget. We transform $B$ into a sequence $B'$ which contains vertices from major branches of $P$-gadgets and the spine of $Tail$-gadget with length at most $k$ as follows: whenever a source lies on a vertex $b_i$ on the minor branch of a $P$-gadget, we replace it by its corresponding vertex $a_i$ on the major branch, unless that vertex already appears earlier in the sequence. All sources originally on major branches are kept unchanged.

By Observation~\ref{obs:pgadget-distances}, replacing a source on a minor branch by its corresponding major-branch vertex cannot delay the burning of any vertex on the major branch. Hence, under $B'$, all vertices on the major branches of the $P$-gadgets and the spine of the $Tail$-gadget are burned no later than under $B$-- note that the vertices other than spine vertices in $Tail$-gadgets will increase the induced cycle length.
If $\lvert B' \rvert < k$, we may add one arbitrary vertex on a major branch to obtain a sequence of length exactly $k$, without affecting the argument below.

Consider the cycle formed by the shortest path $v_9 - v_8 - r - v_1$ in the $Tail$-gadget together with the path from $v_{m^2}$ to $v_9$ along the major branches and the edge $v_1v_{m^2}$. This cycle has length $m^2 - 5$. Any burning sequence that burns $C(m)$ must also burn this cycle.

By Theorem~\ref{thm:path-cycle-burning}, burning a cycle of length $m^2 - 5$ requires at least
$\left\lceil \sqrt{m^2 - 5} \right\rceil = m$ steps for all $m \ge 4$. Since $k < m$, we obtain a contradiction. Therefore, $b(C(m)) \ge m$.
\end{proof}

Combining Lemmas~\ref{lem:C-gadget-burning-upperbound} and ~\ref{lem:C-gadget-burning-lowerbound}, we obtain the following result.

\begin{lemma}
\label{lem:C-gadget-burning}
$b(C(m)) = m$.
\end{lemma}
For convenience, we summarize the notations of the gadgets and associated graphs in the following Table~\ref{tab:lists}.
\input{summary-table}

\subsection{Construction of $H$ from $G'$}
\label{cons:H}

We first show that an appropriate constant factor we plan to use in the construction always exists.
\begin{lemma}
\label{lem:choose-c}
For every positive integer $n'$, there exists a constant $c$ with $4 \le c < 8$ such that $cn'$ is a power of $2$. Consequently, $\log(cn')$ is an integer.
\end{lemma}

\begin{proof}
Let $k$ be the largest integer such that $2^k \le 4n'$. If $4n' = 2^k$, then choosing $c = 4$ satisfies the claim. Otherwise, we have $2^k < 4n' < 2^{k+1}$. Dividing throughout by $n'$ gives $4 < \frac{2^{k+1}}{n'} < 8$. Let $c = \frac{2^{k+1}}{n'}$. Then $4 < c < 8$ and $cn' = 2^{k+1}$, which is a power of $2$. Hence, $\log(cn')$ is an integer.
\end{proof}

Let $G'$ be a graph obtained from a connected cubic graph $G$ by subdividing an edge $(p,q)\in E(G)$ twice,
replacing it with the path $p-x-y-q$ using Construction~\ref{cons:G'}. Let $n$ and $n'$ denote the numbers of vertices of $G$ and $G'$, and let $k$ and $k'$ denote the sizes of minimum vertex covers of $G$ and $G'$,
respectively. By Construction~\ref{cons:G'}, $n'=n+2$, and by Lemma~\ref{lem:G'vertexcover}, $k'=k+1$.
Fix a constant $c$ with $4\le c<8$ such that  {$\log(cn')$} is an integer. By lemma~\ref{lem:choose-c}, the existence of such a $c$ is guaranteed. Set $h=\log(cn')+2 $. Except for the subdivided edge $(p,q)$, there is a one-to-one correspondence between the edges of $G$ and those of $G'$.

The graph $H$ is built from $G'$ by replacing the edges with $BTP$-gadgets and attaching one additional structure. Initially, for each $v\in V(G')$ we introduce a vertex (also named $v$) in $H$. The graph $H$ consists of two parts:  $H_{{core}}$ and  $H_{{annexe}}$.

%------------------------------------------------------------
\begin{construction}[$H_{{core}}$]
\label{cons:H-core}
For every edge $(u,v)\in E(G')$, introduce a $BTP$-gadget between $u$ and $v$ with
parameters
\[
h=\log(cn')+2,\qquad
l_1=\frac{cn'-2h}{2},\qquad
l_2=\frac{cn'}{2}+2 .
\]
Since $cn'$ is a power of $2$, $l_1$ and $l_2$ are integers. We denote this gadget by
$BTP_{uv}\bigl(h,(cn'-2h)/2,\, cn'/2+2\bigr)$. Each of the two binary trees in the gadget has height $h$, hence $2^h=4cn'$ leaves, and each root--leaf path has $h+1=\log(cn')+3$ vertices. Let $H_{uv}$ denote the subgraph induced by $\{u,v\}\cup V(BTP_{uv})$.
\end{construction}

With these parameters, $l_2>l_1+h+1
\quad\text{and}\quad
l_1+l_2=cn'-h+2 < 2^{h-2}.$
Hence, all assumptions required for the properties of $BTP$-gadgets are satisfied.
%------------------------------------------------------------
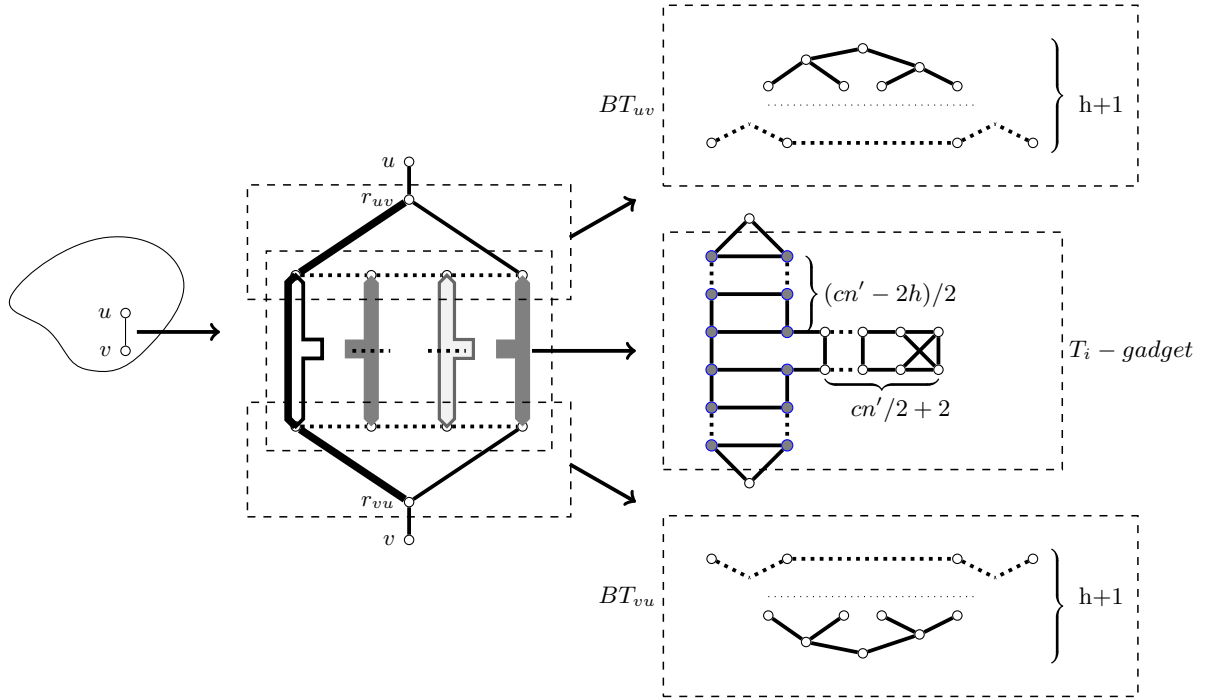
\begin{figure}[h]
\input{figs/uv-subgraph-new}
\caption{The subgraph $BTP_{uv}$ in $H_{core}$ corresponding to an edge $(u,v)\in E(G')$. }
\label{fig:uv-subgraph}
\end {figure}

\begin{construction}[$H_{{annexe}}$]
\label{cons:H-annexe}
The $H_{{annexe}}$ consists of a $Y$-gadget and a $C$-gadget. The $Y$-gadget has parameters
\[
d_1=\left\lfloor\frac{n'}{4}\right\rfloor+\frac{cn'}{2},
\qquad
d_2=\left\lceil\frac{n'}{4}\right\rceil+\frac{cn'}{2}+1 .
\]
A $C$-gadget with parameter $m=cn'+3$ is attached to the vertex $z_b$ of the
$Y$-gadget.
\end{construction}

%There is exactly one $Y$-gadget, one $C$-gadget, and one Tail-gadget (as part of the $C$-gadget) in $H_{{annexe}}$.

\begin{figure}[h]
\input{figs/xy-subgraph}
\caption{$H_{annexe}$}
\label{fig:xy-subgraph}
\end {figure}
%------------------------------------------------------------
\begin{construction}[$H$]
\label{cons:H-combined}
In $G'$, the vertices $x$ and $y$ have degree~$2$. Hence, in $H_{{core}}$, exactly two $BTP$-gadgets are incident with each of $x$ and $y$. Attach $H_{{annexe}}$ to $H_{{core}}$ by adding the edges $(x,x_a)$ and $(y,y_a)$. After this attachment, both $x$ and $y$ have degree~$3$, and the resulting graph $H$ is a connected cubic graph. This completes the construction of $H$ from $G'$.
\end{construction}
From the construction, the distance from either $x$ or $y$ to the vertex $z_b$ is at most $n'/2+cn'+3$. Since $G$ is cubic, $n$ is even (and hence $n'=n+2$ is also even), so $n'/2$ is an integer.

\begin{remark}
Each vertex of $G'$ has a unique corresponding vertex in $H$, denoted by the same label. Whether a vertex or edge is considered in $G'$ or in $H$ is always clear from the context.
\end{remark}

\begin{lemma}
\label{lem:num-vertices-in-H}
The graph $H$ has $O(n^3)$ vertices.
\end{lemma}

\begin{proof}
From Construction~\ref{cons:H-core}, we have $n' = n+2$, $h=\log(cn')+2$, $l_1=(cn'-2h)/2$, and $l_2=cn'/2+2$. By Observation~\ref{obs:btp-gadget-num-vertices}, a $BTP(h,l_1,l_2)$-gadget contains
$2(2^{h+1}-1)+2^h(4l_1+2l_2) = O(n^2)$ vertices. Since $G$ is cubic, $|E(G')|=3n/2+2$, and hence $H$ contains
$3n/2+2$ such $BTP$-gadgets. Therefore, the total number of vertices contributed by all $BTP$-gadgets
is $O(n^3)$. By Observation~\ref{obs:y-gadget-num-vertices}, the $Y$-gadget contains $O(n)$ vertices, and by Observation~\ref{obs:C-size}, the $C$-gadget contains $O(n^2)$ vertices.
Thus, $|V(H)| = O(n^3)+O(n^2)+O(n)=O(n^3)$.
\end{proof}

%------------------------------------------------------------

\begin{definition}
\label{def:domain}
\textbf{Domain of $u$.}
For a vertex $u\in V(G')$, the \emph{domain} of  $u \in V(H)$, denoted as $Dom_u$, is defined as follows.

\begin{itemize}
\item If $u\in V(G')\setminus\{x,y\}$ and $(u,a),(u,b),(u,c)\in E(G')$, then
$Dom_u$ is the subgraph of $H$ induced by
$u_{exthalf}(BTP_{ua}) \cup u_{exthalf}(BTP_{ub}) \cup u_{exthalf}(BTP_{uc}).$

\item If $u=x$, then
$Dom_x = x_{exthalf}(BTP_{xp}) \cup x_{exthalf}(BTP_{xy}) \cup P_x.$

\item If $u=y$, then
$Dom_y = y_{exthalf}(BTP_{yq}) \cup y_{exthalf}(BTP_{yx}) \cup P_y.$

\end{itemize}
\end{definition}

Note that, the domains of distinct vertices are disjoint. If a burning source $s$ lies in $Dom_u$, then $u$ is called the \emph{owner} of $s$.

Now we define
$InsideDomains = V(\bigcup_{u\in V(G')} Dom_u),
OutsideDomains = V(H)\setminus InsideDomains = \{z\}\cup V(P_z)\cup V(C).$

%------------------------------------------------------------

\begin{lemma}
\label{lem:burning_u_from_out_of_Dom_u}
Let $u\in V(H)$ correspond to a vertex $u\in V(G')$  and let $s\in V(H)\setminus Dom_u$. If the fire from $s$ burns $u\in V(H)$, then $Steps[s,u]\ge \frac{cn'}{2}+3.$
\end{lemma}

\begin{proof}
First assume $u\in V(G')\setminus\{x,y\}$. Let $(u,d),(u,e),(u,f)\in E(G')$. The fire reaches $u$ fastest when $s\in Dom_d\cup Dom_e\cup Dom_f$; without loss of generality, let $s\in Dom_d$. The shortest path from $Dom_d$ to $u$ passes through the $T$-gadgets in $BTP_{ud}$. The minimum distance is achieved when
$s=j^i_{du}$ of the $T$-gadget, yielding 
$Steps[s,u] \geq 2 + l_1 + (h+1) = l_1+h+3 = \frac{cn'}{2}+3$.

Now suppose $u=x$ (the case $u=y$ is symmetric). If $s\in Dom_p\cup Dom_y$, the above argument applies.
Otherwise, $s$ lies in the $Y$-gadget. The closest such vertex to $x$ is $z$, and
\[
Steps[z,x] = 1 + \lfloor n'/4\rfloor + \frac{cn'}{2} + 1 \ge \frac{cn'}{2}+3.
\]
\end{proof}

%------------------------------------------------------------

\begin{observation}
\label{obs:Burning-Huv}
For any $(u,v)\in E(G')$, if fire originates from $u$ or $v$ in $H$ and enters $BTP_{uv}$, then $H_{uv}$ is completely burned in $cn'+4$ steps.
\end{observation}

\begin{proof}
By Construction~\ref{cons:H-core}, $V(H_{uv})=\{u,v\}\cup V(BTP_{uv})$. Using Observation~\ref{obs:btp-complete-burning}, the required number of steps is
$h+1+l_1+l_2+1 = cn'+4$.
\end{proof}

%------------------------------------------------------------
Let $H$ be the graph obtained from $G'$ by Construction~\ref{cons:H-combined}. Let $k=\beta(G)$, $k' = k+1$, and $n' = |V(G')|$.  We are ready to prove the forward direction of the reduction.

\begin{lemma}
\label{lem:b-leq-threshold}
$b(H) \le k' + cn' + 3$
\end{lemma}

\begin{proof}
Let $Q$ and $Q'$ be minimum vertex covers of $G$ and $G'$, respectively.  By the proof of Lemma~\ref{lem:G'vertexcover}, $Q' = Q \cup \{x\}$ (or $Q \cup \{y\}$) and $|Q'|=k'$. We construct a burning sequence $B'$ of length $k' + cn' + 3$ that burns $H$.

Place the first source, $ B' [1]$, at the vertex $x\in V(H)$ if $x \in Q'$ (or at the vertex $y \in V(H)$ if $y \in Q'$). Then place the next $k'-1$ sources on the vertices of $H$ corresponding to the vertices of $Q'$ that have not yet been selected (in arbitrary order); these form the \textit{StartBlock}. The remaining $cn'+3$ sources are placed in the $C$-gadget. Recall that $H$ consists of 
\begin{enumerate*}[label=(\roman*)]
\item $\{H_{uv} \mid (u,v)\in E(G')\}$,
\item the $Y$-gadget in $H_{\mathrm{annexe}}$, and
\item the $C$-gadget in $H_{\mathrm{annexe}}$.
\end{enumerate*}

\begin{enumerate}
\item $H_{uv}, \forall (u,v) \in E(G')$: 
Since $Q'$ is a vertex cover of $G'$, for every $(u,v)\in E(G')$ at least one of $u$ or $v$ is a source in the \textit{StartBlock}. Each such source has at least $cn'+3$ remaining steps for propagation.  By Observation~\ref{obs:Burning-Huv}, this suffices to burn $H_{uv}$ completely.

\item  The $Y$-gadget. The first source is placed at $x$ (or $y$).  By Observation~\ref{obs:y-gadget-burning}, the $Y$-gadget is burned within $n'/2 + cn' + 3$ steps.  Since $n/2 \le k$ (Observation~\ref{obs:vc-bounds-cubic-graph}), $n'=n+2$, and $k'=k+1$, we obtain $n'/2 \le k'$.  Hence, the $Y$-gadget is burned in time $k' + cn' + 3$.

\item The $C$-gadget.
By Lemma~\ref{lem:C-gadget-burning}, the $C$-gadget can be burned in $cn'+3$ steps using $cn'+3$ sources.  
These are precisely the final $cn'+3$ sources of $B'$.
\end{enumerate}

Thus, $H$ is completely burned within $k' + cn' + 3$ steps, and therefore $b(H) \le k' + cn' + 3$.
\end{proof}
Now we prove the backward direction of the reduction using the following lemmas and observations.

A subgraph $H$ of a graph $G$ is called an \emph{isometric subgraph} if, for every pair of vertices $u,v \in V(H)$, we have $d_H(u,v)=d_G(u,v)$. The burning number is not guaranteed to be monotonic even on the isometric subgraphs of a graph. The following theorem shows that the burning number is monotonic on the isometric subgraphs in certain cases.

\begin{theorem}[\cite{Bonato2016Burning}]
\label{thm:isometric}
Let $H$ be an isometric subgraph of a graph $G$. Suppose that for every vertex $x \in V(G)\setminus V(H)$ and every positive integer $r$, there exists a vertex $f_r(x)\in V(H)$ such that $N_r^G[x]\cap V(H)\subseteq N_r^H[f_r(x)]$. Then $b(H)\le b(G)$.
\end{theorem}

\begin{lemma}
$b(H) \ge cn'+3$.
\end{lemma}
\begin{proof}
Consider the subgraph induced by the cycle gadget $C$ in $H$. For every $x \in V(H) \setminus V(C)$, and for any positive integer $r$, let $f_r{(x)} = v_{m^2}$. It is easy to see that $N_{r}^H[x] \cap V(C) \subseteq N^C_{r}[f_r{(x)}]$. By Theorem~\ref{thm:isometric}, $b(C) \le b(H)$ and by Lemma~\ref{lem:C-gadget-burning}, $b(C) = cn'+3$. Therefore, $b(H) \ge cn'+3$.
\end{proof}
Let $B'$ be a burning sequence of $H$ with
$|B'|\le k'+cn'+3$. Then Lemmas~\ref{lem:src-out-of-dom-cannot-burn} and~\ref{lem:Huv} show that corresponding to every edge in $G'$ there must be a burning source in the domain of one of its endpoints.
\begin{lemma}
\label{lem:src-out-of-dom-cannot-burn}
Let $B'$ be a burning sequence of $H$ with $|B'|\le k'+cn'+3$. For $(e,f)\in E(G')$, if a source $v\in B'$ satisfies $v\notin Dom_e\cup Dom_f$, then the fire from $v$ cannot burn $Tips_{ef}$ within $k'+cn'+3$ steps.
\end{lemma}

\begin{proof}
Without loss of generality, assume the fire from $v$ reaches $e$ before $f$. By Lemma~\ref{lem:burning_u_from_out_of_Dom_u}, reaching the fire at $e$ requires at least $cn'/2+3$ steps.
Burning $Tips_{ef}$ by a fire from $e$ then needs additional $cn'+3$ steps (Observation~\ref{obs:Burning-Huv}). Since $c\ge4$ implies $cn'/2>k'$, the total exceeds $k'+cn'+3$ and hence fire from $v$ cannot burn
$Tips_{ef}$ within $k'+cn'+3$ steps.
\end{proof}

%------------------------------------------------------------

\begin{lemma}
\label{lem:Huv}
Let $B'$ be a burning sequence with $|B'|\le k'+cn'+3$. For every $(u,v)\in E(G')$, there exists a burning source in $Dom_u\cup Dom_v$.
\end{lemma}

\begin{proof}
Suppose not. Then all vertices in $u_{exthalf}(BTP_{uv})$ must be burned from sources outside $Dom_u\cup Dom_v$. By Lemma~\ref{lem:src-out-of-dom-cannot-burn}, this is not possible within $k'+cn'+3$ steps.
\end{proof}

\begin{observation}
\label{obs:inequality-for-two-observation}
Let $c \ge 4$ and $n'$ be sufficiently large. For $0 \le s \le k'$, where $k' \le \left\lceil \frac{2n'}{3} \right\rceil$ and $h = \log(cn') + 2$, the following inequality holds:
\[
3 + \left\lceil \frac{n'}{4} \right\rceil + \frac{cn'}{2}
\;>\;
k' + h + 4 .
\]
\end{observation}

\begin{proof}
Since $k' \le \left\lceil \frac{2n'}{3} \right\rceil$ and  $h = \log(cn') + 2$, it suffices to prove
\begin{equation}
3 + \left\lceil \frac{n'}{4} \right\rceil + \frac{c n'}{2}
>
\left\lceil \frac{2n'}{3} \right\rceil + \log(cn') + 6 .
\label{eq:target-ineq}
\end{equation}

Using 
$\left\lceil \frac{n'}{4} \right\rceil \ge \frac{n'}{4}$ 
and 
$\left\lceil \frac{2n'}{3} \right\rceil \le \frac{2n'}{3} + 1$,
it is enough to show
\begin{equation}
3 + \frac{n'}{4} + \frac{c n'}{2} > \frac{2n'}{3} + 1 + \log(cn') + 6.
\end{equation}
After rearranging constants, this reduces to
\begin{equation}
\frac{n'}{4} + \left(\frac{c}{2} - \frac{2}{3} \right)n' > \log(cn') + 4 .
\label{eq:reduced}
\end{equation}

Observe that
$\frac{c}{2} - \frac{2}{3} = \frac{3c - 4}{6}$. Since $c \ge 4$, we have
$\frac{3c - 4}{6} \ge \frac{4}{3}$. Hence the left-hand side of~\eqref{eq:reduced} is at least
\[
\frac{n'}{4} + \frac{4}{3}n' = \frac{19}{12}n'.
\]

Thus it suffices to verify that $\frac{19}{12}n' > \log(cn') + 4$. Because $\log(cn') = O(\log n')$  grows sublinearly while $\frac{19}{12}n'$ grows linearly in $n'$, the inequality holds for all sufficiently large $n'$.

Therefore,~\eqref{eq:target-ineq} holds for sufficiently large $n'$, which completes the proof.
\end{proof}

Recall that, in Construction~\ref{cons:H-core} the parameters are chosen as
$4 \le c < 8$, $h = \log(cn') + 2$, $l_1 = \frac{cn' - 2h}{2}$, and $l_2 = \frac{cn'}{2} + 2$.
Let $B' \subseteq V(H)$ be a burning sequence of $H$ with $|B'| = s + cn' + 3$, where $s \ge 0$.  
We partition $B'$ into three consecutive subsequences as follows:
\begin{itemize*}
\item If $s \ge 1$, then $\textit{StartBlock} = B'[1,\dots,s]$; otherwise $\textit{StartBlock}$ is empty.
\item $\textit{MiddleBlock} = B'[s+1,\dots,s+h+1]$,\\
\item $\textit{EndBlock} = B'[s+h+2,\dots,s+cn'+3]$.
\end{itemize*}
Equivalently, for $1 \le p \le s$, $\textit{StartBlock}[p] = B'[p]$; for $1 \le q \le h+1$, 
$\textit{MiddleBlock}[q] = B'[s+q]$; and for $1 \le r \le cn' - h + 2$, $\textit{EndBlock}[r] = B'[s+h+1+r]$.
By the choice of parameters, $|\textit{MiddleBlock}| = h+1$, 
$|\textit{EndBlock}| = l_1 + l_2 = cn' - h + 2$, and therefore
$|\textit{MiddleBlock}| + |\textit{EndBlock}| = (h+1) + (cn' - h + 2) = cn' + 3$.

\begin{observation}
\label{obs:fire-from-h-to-c}
Let $B'$ be a sequence of burning sources of $H$ with $|B'| = s + cn' + 3$, where $0 \le s \le k'$. The fire originating from any source in $InsideDomains$ can enter the $C$-gadget only after
$|StartBlock| + |MiddleBlock| + 3$ steps.
\end{observation}

\begin{proof}
Among all vertices in $InsideDomains$, the ones closest to the $C$-gadget are $x_b$ and $y_b$. Without loss of generality, consider a source placed at $x_b$. To reach $C$, the fire must traverse the entire path gadget $P_z$. Hence,
\[
Steps[x_b, z_b] = Steps[x_b,z] + Steps(z,z_b] = 2 + |Major(P_z)|
= 3 + \left\lceil \frac{n'}{4} \right\rceil + \frac{cn'}{2}.
\]

Since $|StartBlock|\le k'\le \left\lceil \frac{2n'}{3} \right\rceil$ and $|MiddleBlock|=h+1$,
\[
|StartBlock|+|MiddleBlock| \le \left\lceil \frac{2n'}{3} \right\rceil + h + 1.
\]
By Observation~\ref{obs:inequality-for-two-observation} and due to the fact that $s\leq k'\leq \lceil \frac{2n'}{3} \rceil$,
\[
3 + \left\lceil \frac{n'}{4} \right\rceil + \frac{cn'}{2} > |StartBlock| + |MiddleBlock| + 3.
\]
Thus, the fire from any source in $InsideDomains$ can enter the $C$-gadget only after $|StartBlock| + |MiddleBlock| + 3$ steps.
\end{proof}

\begin{observation}
\label{obs:max-vertices-in-C-burned-by-insidedomain}
Let $B'$ be a burning sequence of $H$ with $|B'| = s + cn' + 3$, where $0 \le s \le k'$. Let $S \subseteq InsideDomains$ be the set of sources whose fire enters the $C$-gadget. Then the sources in $S$ can burn at most $2(|EndBlock|-3)$ vertices of $trunk'(C)$.
\end{observation}

\begin{proof}
Let $S_1 \subseteq S$ be the subset of sources whose fire reaches $C$ earliest. Any vertex of $trunk'(C)$ burned by a source in $S\setminus S_1$ is also burned by a source in $S_1$ as there is only one entry point $v_{m^2}$ to the $C$-gadget. Hence, it is enough to consider the case in which the burning source is in $S_1$.

Let $v\in S_1$. By Observation~\ref{obs:fire-from-h-to-c}, the fire from $v$ can enter $C$ only after
$|StartBlock|+|MiddleBlock|+3$ steps. Since the total length of the sequence is $|B'|$, the fire from $v$ has at most $|B'|-(|StartBlock|+|MiddleBlock|+3)=|EndBlock|-3$ steps to spread inside $C$.

Within $trunk'(C)$, the fire propagates along a path in two directions. Thus, in each step it can burn at most two new vertices. Hence, the fire from $v$ can burn at most $2(|EndBlock|-3)$ vertices of $trunk'(C)$.
\end{proof}

%------------------------------------------------------------

\begin{lemma}
\label{lem:num-sources-for-C}
Let $B'$ be a burning sequence of $H$ with $|B'| = s + cn' + 3$, where $0 \le s \le k'$.
If $StartBlock \cap OutsideDomains = \emptyset$, then every source in $MiddleBlock \cup EndBlock$ lies in $OutsideDomains$.
\end{lemma}

\begin{proof}
Since $B'$ burns $H$, it must completely burn the $C$-gadget. If no source in $InsideDomains$ enters $C$, the claim holds trivially. Thus, assume fire from some sources in $InsideDomains$ enter $C$.
By Observation~\ref{obs:max-vertices-in-C-burned-by-insidedomain}, the total number of vertices of $trunk'(C)$ burned by such sources is at most
\[
2(|EndBlock|-3) \le 2(cn'-h-1) < 2cn'-5.
\]
The total size of $trunk'(C)$ is $(cn'+3)^2$, with a FixedOverlap of $5$ vertices. Hence, the number of vertices remaining to be burned is at least
\[
(cn'+3)^2 - 5 - (2cn'-5) > (cn'+2)^2.
\]

By Theorem~\ref{thm:path-cycle-burning}, burning a path of length strictly greater than $(cn'+2)^2$ requires at least $cn'+3$ sources. Since $StartBlock$ contains no sources from $OutsideDomains$ and
$|MiddleBlock|+|EndBlock|=cn'+3$, all sources in $MiddleBlock \cup EndBlock$ must belong to $OutsideDomains$.
\end{proof}

%------------------------------------------------------------
 An edge $(u,v)\in E(G')$ is called \emph{represented} if  $\exists w \in StartBlock$ that lies in $Dom_u\cup Dom_v$, and \emph{unrepresented} otherwise.
%------------------------------------------------------------

\begin{lemma}
\label{lem:endblock-atmost-2-btp}
Let $B'$ be a burning sequence of $H$ with $|B'| = s + cn' + 3$, where $0 \le s \le k'$. Then at most two $BTP$-gadgets corresponding to unrepresented edges can be completely burned by $B'$.
\end{lemma}

\begin{proof}
Assume, for contradiction, that three distinct edges $(a,b),(e,f),(i,j)\in E(G')$ are unrepresented. Let $u\in StartBlock$ be any source from which fire reaches a vertex in $\{a,b,e,f,i,j\}$. Since $u\notin Dom_v$ for all such vertices $v$, Lemma~\ref{lem:burning_u_from_out_of_Dom_u} implies
\[
Steps[u,v] \ge \frac{cn'}{2}+3.
\]
Since $c\ge4$, this exceeds $k'+3\ge s+3$.

Thus, all conditions of Lemma~\ref{lem:possible-to-burn-two-not-three} are satisfied -- {note that it is not necessary that the edges are disjoint, as even if edges share an endpoint, their corresponding $BTP$-gadgets are disjoint} -- implying that three such $BTP$-gadgets cannot be burned within $s+cn'+3$ steps, a contradiction.
\end{proof}

%------------------------------------------------------------

\begin{lemma}
\label{lem:at-least-four-in-middleblock}
Let $B'$ be a burning sequence of $H$ with $|B'| = s + cn' + 3$, where $0 \le s \le k'$. If two $BTP$-gadgets corresponding to unrepresented edges are burned exclusively by sources in $MiddleBlock \cup EndBlock$, then at least four sources in $MiddleBlock$ belong to $InsideDomains$.
\end{lemma}

\begin{proof}
Let $BTP_{ab}$ and $BTP_{ef}$ be the two such gadgets--  recall that even if two edges share an endpoint, their corresponding $BTP$-gadgets are disjoint. Their tip sets are disjoint and satisfy
$|Tips_{ab}| = |Tips_{ef}| = 4cn'$. Moreover, the distance between any tip of $BTP_{ab}$ and any tip of $BTP_{ef}$ is greater than $2(cn'+3)$. Since $|MiddleBlock| + |EndBlock| = cn' + 3$, no source can burn tips from both gadgets. Thus, the sources in $MiddleBlock \cup EndBlock$ must burn $Tips_{ab}$ and $Tips_{ef}$ independently. Without loss of generality, assume that the earliest source in $MiddleBlock$, namely $MiddleBlock[1]$, is used to burn $BTP_{ab}$. We may assume that it is placed at $r_{ab}$, {then it will burn maximum number of $tip_{ab}$}. Within $h$ steps, all leaves of $BT_{ab}$ are burned, and within an additional $l_1+l_2$ steps, all tips in $Tips_{ab}$ are burned. Thus, all remaining sources in $MiddleBlock[2,\dotsc,h+1]\cup EndBlock$  can be devoted to burning $Tips_{ef}$.
By Lemma~\ref{lem:generalized-max-leaves-in-endblock},
sources in $EndBlock$ can burn at most $l_1 + l_2 = cn' - h + 2 < cn'$ tips. Hence, more than $3cn'$ tips of $Tips_{ef}$ must be burned by sources in $MiddleBlock$. Each source in $MiddleBlock$ can burn the number of $tip_{ab}$ up to a limiting value :
the earliest remaining source, $MiddleBlock[2]$, can burn at most $2cn'$ tips, the next source, $MiddleBlock[3]$, can burn at most $cn'$ tips, and every subsequent source burns fewer than $cn'$ tips
(by Lemma~\ref{lem:max-tips-from-source-in-middleblock}).
Therefore, to burn more than $3cn'$ tips, at least three sources from 
$MiddleBlock \setminus \{MiddleBlock[1]\}$ are required. Including $MiddleBlock[1]$, at least four sources from $MiddleBlock$ are needed. All such sources lie in 
$V(BTP_{ab}) \cup V(BTP_{ef}) \subseteq InsideDomains$.
\end{proof}

\begin{lemma}
\label{lem:middbloc-insidedomain-trunk-startblock}
Suppose $t$ sources from $MiddleBlock$ are placed at vertices in $InsideDomains$ and none of their fires infiltrate the $C$-gadget. Then at least $t^2 + 4t - 2ht + 2tcn' - 5$ vertices of $trunk'(C)$ must be burned by sources in $StartBlock$.
\end{lemma}
\begin{proof}
We first consider $trunk(C)$, a cycle on $(cn'+3)^2$ vertices. Recall that 
$|MiddleBlock| + |EndBlock| = cn' + 3$. Let $M_e = MiddleBlock \cup EndBlock$. Since 
$|M_e| = cn'+3 = b(trunk(C))$, the sources in $M_e$ can be viewed as an optimal burning sequence for $trunk(C)$, where the $i^{th}$ source burns at most $1 + 2(|M_e| - i) = 1 + 2(cn' + 3 - i)$ vertices on the cycle.
Let the $t$ sources from $MiddleBlock$ that lie in $InsideDomains$ occur at
positions
\[
1 \le pos_1 < pos_2 < \cdots < pos_t \le h+1.
\]
Since their fires do not reach the $C$-gadget, their contribution to burning $trunk(C)$ is lost. The total number of vertices they would have burned is
\[
\sum_{i=1}^{t} \bigl(1 + 2(cn' + 3 - pos_i)\bigr).
\]
This sum is maximized when the positions $pos_i$ are as small as possible, and minimized when they are as large as possible. To obtain a lower bound on the number of vertices that remain unburned, we take
\[
pos_i = (h+1) - (t - i), \qquad 1 \le i \le t.
\]
Substituting yields that at least
\[
t^2 + 4t - 2ht + 2tcn'
\]
vertices of $trunk(C)$ remain unburned. Now consider $trunk'(C)$. By the definition of $FixedOverlap$, at most $5$ vertices can be burned by more than one source and hence more than once. Hence, at least
\[
t^2 + 4t - 2ht + 2tcn' - 5
\]
vertices of $trunk'(C)$ must be burned by sources in $StartBlock$.
\end{proof}

\begin{lemma}
\label{lem:num-vertices-in-C-burned-by-startblock}
Let $B'$ be a burning sequence of $H$ with $|B'| = s + cn' + 3$, where $0 \le s \le k'$. If at least four sources in $MiddleBlock$ are placed at vertices in $InsideDomains$, then at least $6cn'$ vertices of $trunk'(C)$ are burned by sources in $StartBlock$.
\end{lemma}

\begin{proof}
Suppose that at least four sources from $MiddleBlock$ are placed at vertices in the $InsideDomains$. By applying Lemma~\ref{lem:middbloc-insidedomain-trunk-startblock} with $t=4$, the number of vertices in $trunk'(C)$ will be at least $t^2+4t-2ht+2tcn'-5 = 8cn'-8h+27$. However, in Lemma~\ref{lem:middbloc-insidedomain-trunk-startblock} it is assumed that the fire from no source in $MiddleBlock \cap InsideDomains$ will infiltrate $C$. This need not be the case. The fire from a source in $MiddleBlock \cap InsideDomains$ may infiltrate $C$ and burn some vertices of $trunk(C)$. This infiltration is maximized when $MiddleBlock[1]$ is placed at $x_b$ or $y_b$, the vertices of $InsideDomains$ closest to $C$. Without loss of generality, assume $MiddleBlock[1]$ is placed at $x_b$. Then the number of vertices of $trunk'(C)$ burned by $MiddleBlock[1]$ is $2((cn'+3) - Steps[x_b,z_b]) = 2((cn'+3) - (3 + \lceil n'/4 \rceil + cn'/2)) = cn' - n'/2$. Hence, the number of vertices of $trunk'(C)$ that must be burned by sources in $StartBlock$ is at least $ 8cn'-8h+27 - (cn' - n'/2) = 7cn' + n'/2 - 8h + 27$. 

Finally, note that $7cn' + n'/2 - 8h + 27 = 6cn' + (cn' + n'/2 - 8h + 27)$, and  $cn' + n'/2 - 8h + 27\geq 0$,  for $4 \leq c < 8$ and $n' \ge 6$. Therefore, if $|MiddleBlock \cap InsideDomains| \geq 4$, then at least $6cn'$ vertices of $trunk'(C)$ must be burned by sources in $StartBlock$.
\end{proof}
\begin{lemma}
\label{lem:four-insidedoamins-middleblock-two-outsidedomains-startblock}
Let $B'$ be a burning sequence of $H$ with $|B'| = s + cn' + 3$, where $0 \le s \le k'$.  If four or more sources in $MiddleBlock$ are placed at vertices in $InsideDomains$, then at least two sources in $StartBlock$ belong to $OutsideDomains$.
\end{lemma}

\begin{proof}
Suppose four sources in $MiddleBlock$ belong to $InsideDomains$. By Lemma~\ref{lem:num-vertices-in-C-burned-by-startblock}, at least $6cn'$ vertices of $trunk'(C)$ must be burned by sources in $StartBlock$.
If a source in $StartBlock \cap InsideDomains$ infiltrates $C$, then by 
Observation~\ref{obs:max-vertices-in-C-burned-by-insidedomain}, it can burn fewer than $2cn'$ vertices of $trunk'(C)$. Thus, at least $4cn'$ vertices must be burned by sources in
$StartBlock \cap OutsideDomains$. Now consider a source $u \in StartBlock \cap OutsideDomains$ placed at time step $t \le s$. The fire from $u$ has at most $k' + cn' + 2$ steps to propagate, and therefore burns
at most
\[
1 + 2(k' + cn' + 2)
\]
vertices of $trunk'(C)$. Since $k' \le \lceil 2n'/3 \rceil$ and $c \ge 4$, we have
\[
1 + 2(k' + cn' + 2) < 4cn'.
\]
Thus, no single source in $StartBlock \cap OutsideDomains$ can burn $4cn'$ vertices of $trunk'(C)$. It follows that at least two sources in $StartBlock$ must belong to $OutsideDomains$. If no source in $StartBlock \cap InsideDomains$ infiltrates $C$, then all $6cn'$ vertices must be burned by sources in $StartBlock \cap OutsideDomains$, and therefore the same lower bound follows immediately.
\end{proof}
\begin{observation}
\label{obs:cycle-path-source-burns-exclusive-vertex}
Let $G_t$ be a cycle or a path with $(t-1)^2 < \vert V(G_t) \vert < t^2$. Let $B = (v_1, v_2, \dotsc, v_t)$ be an optimal burning sequence for $G_t$. Let $t^2 - \vert V(G_t) \vert < 1+2d, 0 \leq d \leq (t-1)$, then $1 \leq j \leq t-d$, $v_j \in B$  burns at least one vertex exclusively. 
\end{observation}
\begin{proof}
Since $|V(G_t)| > (t-1)^2$, by Theorem~\ref{thm:path-cycle-burning}, $b(G_t) = t$. For each $1 \le i \le t$, the source $v_i$ can burn at most $1 + 2(t-i)$ vertices in $G_t$. Summing over all sources yields
\(
\sum_{i=1}^{t} \bigl(1 + 2(t-i)\bigr) = t^2.
\)
Hence, at most $t^2 - |V(G_t)|$ vertices are burned more than once. Now consider any source $v_j$ with $1 \le j \le t-d$. It can burn  at least
\(
1 + 2(t-j) \;\ge\; 1 + 2d
\) vertices. Since $t^2 - |V(G_t)| < 1 + 2d$, the total number of vertices that can be burned by multiple sources is strictly smaller than the number of vertices $v_j$ can burn. Therefore, not all vertices burned by $v_j$ can be burned by other sources, and hence $v_j$ burns at least one vertex exclusively.
\end{proof}

Now we are ready to prove the backward direction of the reduction. 
\begin{lemma}
\label{lem:b-geq-threshold}
$b(H) \ge k' + cn' + 3$.
\end{lemma}

\begin{proof}
Assume for contradiction that $b(H) < k' + cn' + 3$. Let $B'$ be a burning sequence of $H$ with $|B'| = (k'-1) + cn' + 3$. Then $|StartBlock| = k'-1$, $|MiddleBlock| = h+1$, and $|EndBlock| = cn' - h + 2$.
Since $k'$ is the vertex covering number of $G'$, there exists an edge $(e,f)\in E(G')$ that is \emph{unrepresented} in $StartBlock$. By Lemma~\ref{lem:src-out-of-dom-cannot-burn}, no source in $StartBlock$ can burn any vertex of $Tips_{ef}$ within $|B'|$ steps. Hence, the gadget $BTP_{ef}$ must be burned using sources from $MiddleBlock \cup EndBlock$. By Lemma~\ref{lem:max-leaves-in-endblock}, $EndBlock$ alone cannot burn all vertices of $Tips_{ef}$, so at least one source in $MiddleBlock$ must burn some vertex of $Tips_{ef}$. Such a source must lie in $InsideDomains$, since fire from a source in $OutsideDomains$ cannot reach $Tips_{ef}$ within  number of available steps for burning {(by Observation~\ref{obs:fire-from-h-to-c}, a source in $OutsideDomains$ reaches $e$ or $f$ only after $\lceil 2n'/3\rceil+h+1$ steps, and by Observation~\ref{obs:btp-complete-burning}, additional $h+1+l_1+l_2$ steps are required to reach any $tip_{ef}$ resulting in number of steps greater than the number of available steps)}.
Suppose first that all sources in $StartBlock$ lie in $InsideDomains$. By Observation~\ref{obs:max-vertices-in-C-burned-by-insidedomain}, these sources burn at most
\[
2(|EndBlock|-3) = 2(cn'-h-1)
\]
vertices of $trunk'(C)$. Therefore, at least
\[
(cn'+3)^2 - FixedOverlap - 2(cn'-h-1) > (cn'+2)^2
\]
vertices of $trunk'(C)$ must be burned by $MiddleBlock \cup EndBlock$. Since fire from 
$StartBlock \cap InsideDomains$ can enter the $C$-gadget only through $v_{m^2}$, the vertices of $trunk'(C)$ burned by these sources form a contiguous subpath.

Let $P$ be the complementary subpath of the unburned vertices in $trunk'(C)$ with
\[(cn'+3)^2 - FixedOverlap - 2(cn'-h-1) \leq \vert V(P) \vert.\]
Rearranging the terms, we get\[(cn'+3)^2 - \vert V(P) \vert \leq FixedOverlap + 2(cn'-h-1)\]\[\implies(cn'+3)^2 - \vert V(P) \vert \leq  1 + 2(cn'-h+1).\] 
Let $d=cn'-h+2$. Then\[(cn'+3)^2-|V(P)| < 1+2d .\]

Since fire from no source in $InsideDomains$ can reach the vertices of $P$, the burning of $P$ is entirely done by sources in $MiddleBlock \cup EndBlock$. Hence, by Observation~\ref{obs:cycle-path-source-burns-exclusive-vertex}, the first $(cn'+3)-d = h+1$ sources in $MiddleBlock \cup EndBlock$ each burn at least one vertex of $P$ exclusively. As $|MiddleBlock| = h+1$, all sources in $MiddleBlock$ are among these, and thus each burns at least one vertex of $P$ exclusively. Suppose $u$ is a vertex burned exclusively by the source placed at position $pos$ in $MiddleBlock$. If the source at $pos$ is placed at a vertex $v \in InsideDomains$, then $u$ would remain unburned as no other source in $MiddleBlock \cup EndBlock$ burns $u$, and the fire initiated at $v$ cannot reach $u$ within the available number of steps, by the definition of $P$ and the fact that $v \in InsideDomains$.

Hence, this vertex, $u$, must be burned by a source in $StartBlock \cap OutsideDomains$. Therefore, at least one source in $StartBlock$ lies in $OutsideDomains$, which implies the existence of another unrepresented edge $(a,b)\neq(e,f)$. Consequently, both $BTP_{ef}$ and $BTP_{ab}$  must be burned using $MiddleBlock \cup EndBlock$. By Lemma~\ref{lem:at-least-four-in-middleblock} and Lemma~\ref{lem:four-insidedoamins-middleblock-two-outsidedomains-startblock}, this requires at least two sources in $StartBlock \cap OutsideDomains$, implying at least three unrepresented edges. This contradicts Lemma~\ref{lem:endblock-atmost-2-btp}. Therefore, $b(H) \ge k' + cn' + 3$.
\end{proof}

Now, by Lemma~\ref{lem:b-leq-threshold} and Lemma~\ref{lem:b-geq-threshold}, we have the following lemma.

\begin{lemma}
\label{lem:b-equals-threshold}
$b(H)=k'+cn'+3$.
\end{lemma}

Now we restate the Theorem~\ref{thm:cubicNPC} and prove it.

\cubicthm*

\begin{proof}
By Lemma~\ref{lem:b-equals-threshold}, $b(H)=k'+cn'+3 = k+cn'+4$ if and only if $G$ admits a vertex cover of size $k$. Since \MVC\ is \NPC\ on connected cubic graphs, so is \BNP.
\end{proof}

Now we restate the Theorem~\ref{thm:cubicAPX} here.

\apxcubicthm*

\begin{proof}
It is known that \MVC\ on cubic graphs is APX-hard (Theorem~3.1 of~\cite{AlimontiKann2000}). We give an $L$-reduction from \MVC\ to \BNP.

Let $G$ be a connected cubic graph with $n=|V(G)|$ and $k=\beta(G)$. By Observation~\ref{obs:vc-bounds-cubic-graph}, $k \ge n/2$. Let $H$ be the graph obtained from $G$ by Construction~\ref{cons:H-combined}.
Then $n' = n+2$ and $k' = k+1$. By Lemma~\ref{lem:b-equals-threshold}, $b(H)=k'+cn'+3 = k + cn' + 4$.
Thus $b(H)$ is linear in $k$, yielding an $L$-reduction. Assume there exists a $(1+\epsilon)$-approximation algorithm $A$ for \BNP. Applied to $H$, it returns $b^*(H) \le (1+\epsilon)(k+cn'+4)$.
Define $k^* = b^*(H) - (cn'+4)$. Then
\begin{align*}
k^*
&\le (1+\epsilon)(k+cn'+4) - (cn'+4) \\
&= k + \epsilon(k+cn'+4),
\end{align*}
and therefore
\begin{align*}
\frac{k^*}{k}
&\le 1 + \epsilon + \frac{\epsilon(cn'+4)}{k} \\
&\le 1 + \epsilon + \frac{\epsilon(c(n+2)+4)}{n/2}.
\end{align*}

Since $4 \le c < 8$,
\[
\frac{c(n+2)+4}{n/2} = \frac{2(c(n+2)+4)}{n} = 2c + \frac{4c+8}{n} \le 16 + \frac{40}{n}.
\]
Hence
\[
\frac{k^*}{k} \le 1 + \epsilon\!\left(1 + 16 + \frac{40}{n}\right)
= 1 + \epsilon\!\left(17 + \frac{40}{n}\right).
\]

Since $n\ge 4$ (the smallest cubic graph is $K_4$, therefore, $n \geq 4$), we have
\[
17 + \frac{40}{4} \le 27.
\]
Hence
\[
\frac{k^*}{k}
\le
1 + 27\epsilon.
\]

Thus, setting $\delta = 27\epsilon$, we obtain $k^* \le (1+\delta)k$. Therefore a PTAS for \BNP\ would imply a PTAS for \MVC, which contradicts the APX-hardness of \MVC\ on cubic graphs.
\end{proof}

%% file: figs/k4.tex
\tikzstyle{filled vertex}  = [{circle,draw=blue,fill=black!50,inner sep=1.5pt}]  % regular vertex
\tikzstyle{empty vertex}  = [{circle, draw, fill = white, inner sep=0.5pt}]
  \centering \small
   \resizebox{0.3\linewidth}{!}{
\begin{tikzpicture}[scale=.25] %6
      \def\n{3}
      \def\radius{2}      
      \coordinate (v0) at ({90 + 180 / \n}:\radius) {};
      \foreach \i in {1,..., \n} {
        \pgfmathsetmacro{\angle}{90 - (\i - .5) * (360 / \n)}
        \node[empty vertex] (v\i) at (\angle:-\radius) {};

   }
      \node[empty vertex] (c) at  (0,0){};

      \foreach \i in {1,2,3} 
        \draw  (c) -- (v\i);
      \foreach \i in {3,2}
        \draw  let \n1 = {int(\i - 1)} in (v\n1) -- (v\i);
       \foreach \i in {3}
        \draw  let \n1 = {int(\i - 2)} in (v\n1) -- (v\i);

         \node[empty vertex, label={[xshift=-4pt,yshift=-5pt]:\tiny{$p$}}] at (v1) {};
         \node[empty vertex, label={[xshift=-4pt,yshift=-5pt]:\tiny{$q$}}] at (v2) {};
        \node[empty vertex, label={[xshift=4pt,yshift=-5pt]:\tiny{$a$}}] at (v3) {};
        \node[empty vertex, label={[yshift=-10pt]:\tiny{$b$}}] at (c) {};               
        
    \end{tikzpicture}
    }

%% file: figs/k4_.tex
\tikzstyle{filled vertex}  = [{circle,draw=blue,fill=black!50,inner sep=1.5pt}]  % regular vertex
\tikzstyle{empty vertex}  = [{circle, draw, fill = white, inner sep=0.5pt}]
  \centering \small
  \resizebox{0.3\linewidth}{!}{
\begin{tikzpicture}[scale=.25] %6
      \def\n{3}
      \def\radius{2}      
      \coordinate (v0) at ({90 + 180 / \n}:\radius) {};
      \foreach \i in {1,..., \n} {
        \pgfmathsetmacro{\angle}{90 - (\i - .5) * (360 / \n)}
        \node[empty vertex] (v\i) at (\angle:-\radius) {};

   }
      \node[empty vertex] (c) at  (0,0){};
      \node[empty vertex, label={[xshift=-4pt, yshift=-5pt]:\tiny{$y$}}] (y) at  (-0.65,0.9){};
      \node[empty vertex, label={[xshift=-5pt, yshift=-5pt]:\tiny{$x$}}] (x) at  (-1.15,0){};   
      \foreach \i in {1,2,3} 
        \draw (c) -- (v\i);
      \foreach \i in {3}
        \draw let \n1 = {int(\i - 1)} in (v\n1) -- (v\i);
       \foreach \i in {3}
        \draw let \n1 = {int(\i - 2)} in (v\n1) -- (v\i);
        
        \draw (v1) -- (x)--(y) -- (v2);

        \node[empty vertex, label={[xshift=-4pt,yshift=-5pt]:\tiny{$p$}}] at (v1) {};
        \node[empty vertex, label={[xshift=-4pt,yshift=-5pt]:\tiny{$q$}}] at (v2) {};
        \node[empty vertex, label={[xshift=4pt,yshift=-5pt]:\tiny{$a$}}] at (v3) {};  
        \node[empty vertex, label={[yshift=-10pt]:\tiny{$b$}}] at (c) {};       
        
    \end{tikzpicture}
    }

%% file: figs/t-gadget-vertical.tex
\tikzstyle{filled vertex}  = [{circle,draw=blue,fill=black!50,inner sep=1.5pt}]  % regular vertex
\tikzstyle{empty vertex}  = [{circle, draw, fill = white, inner sep=1.25pt}]
\centering \small
%\resizebox{}{!}
{\begin{tikzpicture}[scale=.5] %6
  {
    \def\n{4}
    \node[empty vertex, label=below:\small{$q$}] (v1) at  (2.5,0){};  
    
	\node[filled vertex, label={[xshift=-5.5mm,yshift=-4mm]:\small{$f^1_{qp}$}}] (v2) at  (2,2){}; 
 	\node[filled vertex, label=right:\small{$f^2_{qp}$}] (v2') at  (3,2){}; 
    \draw [line width=0.5mm] (v1) --(v2);
    \draw [line width=0.5mm] (v1) --(v2');
    \draw [line width=0.5mm] (v2) --(v2');

	\node[filled vertex] (v3) at  (2,3){}; 
 	\node[filled vertex] (v3') at  (3,3){};
    \draw [line width=0.5mm] (v3) --(v3');
    \draw [line width=0.5mm] (v2) --(v3);
    \draw [line width=0.5mm] (v2') --(v3');    

	\node[filled vertex] (v4) at  (2,4){}; 
 	\node[filled vertex] (v4') at  (3,4){};
    \draw [line width=0.5mm] (v4) --(v4');
    \draw [line width=0.5mm] (v3) --(v4);
    \draw [line width=0.5mm] (v3') --(v4'); 

	\node[filled vertex] (v5) at  (2,5){}; 
 	\node[filled vertex] (v5') at  (3,5){};
    \draw[line width=0.5mm] (v5) --(v5');
    \draw [dotted, line width=0.5mm](v4) --(v5);
    \draw [dotted, line width=0.5mm] (v4') --(v5');    
    
	\node[filled vertex, label={[xshift=-5.5mm,yshift=-4mm]:\small{$j_{qp}$}}] (v6) at  (2,6){}; 
    
 	\node[filled vertex] (v6') at  (3,6){};
    \node[ label=below:\small{$jn_{qp}$}] at  (2.90,6.35){};
    \draw [line width=0.5mm] (v6) --(v6');
    \draw [line width=0.5mm] (v5) --(v6);
    \draw [line width=0.5mm] (v5') --(v6');     

	\node[filled vertex, label={[xshift=-5mm,yshift=-3mm]:\small{$j_{pq}$}}] (v7) at  (2,7){}; 
 	\node[filled vertex] (v7') at  (3,7){};
    \node[label=above:\small{$jn_{pq}$}]  at  (2.90,6.7){};
    \draw [line width=0.5mm] (v7) --(v7');
    \draw [line width=0.5mm] (v6) --(v7);

    \node[filled vertex] (v8) at  (2,8){}; 
 	\node[filled vertex] (v8') at  (3,8){};
    \draw [line width=0.5mm] (v8) --(v8');
    \draw [line width=0.5mm] (v7) --(v8);
    \draw [line width=0.5mm] (v7') --(v8');  

	\node[filled vertex] (v9) at  (2,9){}; 
 	\node[filled vertex] (v9') at  (3,9){};
    \draw [line width=0.5mm] (v9) --(v9');
    \draw [dotted, line width=0.5mm] (v8) --(v9);
    \draw [dotted, line width=0.5mm](v8') --(v9');   

	\node[filled vertex] (v10) at  (2,10){}; 
 	\node[filled vertex] (v10') at  (3,10){};
    \draw [line width=0.5mm] (v10) --(v10');
    \draw [line width=0.5mm] (v9) --(v10);
    \draw [line width=0.5mm] (v9') --(v10');   
    
	\node[filled vertex,label={[xshift=-5.5mm,yshift=-4mm]:\small{$f^1_{pq}$}}] (v11) at  (2,11){}; 
 	\node[filled vertex,label=right:\small{$f^2_{pq}$}] (v11') at  (3,11){};
    \draw [line width=0.5mm] (v11) --(v11');
    \draw [line width=0.5mm] (v10) --(v11);
    \draw [line width=0.5mm] (v10') --(v11');       
 
	\node[empty vertex, label=above:\small{$p$}] (v12) at  (2.5,13){}; 
    
    \draw [line width=0.5mm] (v11)--(v12);
    \draw [line width=0.5mm] (v11')--(v12);

	\node[empty vertex] (s1) at  (4,6){};
    \node[label=below:\small{$w_{qp}$}] at  (4.4,6){};
	\node[empty vertex] (s1') at (4,7){}; 
    \node[label=below:\small{$w_{pq}$}]  at (4.3,8.1){};
    \draw  [line width=0.5mm] (v6')--(s1);
    \draw [line width=0.5mm] (v7')--(s1');
    \draw [line width=0.5mm] (s1)--(s1');

	\node[empty vertex] (s2) at  (5,6){};
	\node[empty vertex] (s2') at (5,7){};    
    \draw [line width=0.5mm] (s1)--(s2);
    \draw [line width=0.5mm] (s1')--(s2');
    \draw [line width=0.5mm] (s2)--(s2');    

	\node[empty vertex] (s3) at  (6,6){};
	\node[empty vertex] (s3') at (6,7){};    
    \draw [dotted, line width=0.5mm](s2)--(s3);
    \draw [dotted, line width=0.5mm](s2')--(s3');
    \draw [line width=0.5mm] (s3)--(s3');    

	\node[empty vertex] (s4) at  (7,6){};
	\node[empty vertex] (s4') at (7,7){};    
    \draw [line width=0.5mm] (s3)--(s4);
    \draw [line width=0.5mm] (s3')--(s4');  
          
	\node[empty vertex,label={[xshift=5mm, yshift=-6mm]:\small{$tip_{qp}$}}] (s5) at  (8,6){};
	\node[empty vertex,label={[xshift=5mm, yshift=-2.5mm]:\small{$tip_{pq}$}}] (s5') at (8,7){};    
    
    \draw [line width=0.5mm] (s4)--(s5);
    \draw [line width=0.5mm] (s4)--(s5');
    \draw [line width=0.5mm] (s4')--(s5);
    \draw [line width=0.5mm] (s4')--(s5');    
    \draw [line width=0.5mm] (s5)--(s5');    

% decorations
	\draw[decoration={calligraphic brace, amplitude=8pt},decorate, line width=1pt] (1.80,7) -- (1.80,11);
    \node[label=left:\small{$l_1$}] at (1.50,9) {};

	\draw[decoration={calligraphic brace, amplitude=8pt},decorate, line width=1pt] (1.80,2) -- (1.80,6);
    \node[label=left:\small{$l_1$}] at (1.50,4) {};   

	\draw[decoration={calligraphic brace, amplitude=10pt},decorate, line width=1pt] (8,6) -- (4,6);
    \node[label=below:\small{$l_2$}] at (6,5.6) {}; 

% coveing rectangle
    \node (top-left) at (-0.5, 11.75){};
    \node (bottom-right) at (9.50,1.25) {};
    
    \node[rectangle, draw,dashed,inner sep=0pt, line width = 0.2mm][fit= (top-left) (bottom-right),  inner xsep=5ex, inner ysep=1ex] {}; 

% rectangles to mark half portions
    \node (h1-left) at (0, 11.35){};
    \node (h1-right) at (9, 7.15) {};
    
    \node[rectangle, draw,dashed,inner sep=0pt, line width = 0.2mm, color=red][fit= (h1-left) (h1-right),  inner xsep=5ex, inner ysep=1ex] {}; 

    \node (h2-left) at (0, 5.85){};
    \node (h2-right) at (9, 1.60) {};
    
    \node[rectangle, draw,dashed,inner sep=0pt, line width = 0.2mm, color=red][fit= (h2-left) (h2-right),  inner xsep=5ex, inner ysep=1ex] {};    

    % add labels for Thalf-pq and Thalf-qp

   \node[color=red, label= right:\footnotesize{\textcolor{red}{$T_{half-pq}$}}] (half-1) at  (7.5,11.25){};

    \node[color=red, label= right:\footnotesize{\textcolor{red}{$T_{half-qp}$}}] (half-2) at  (7.5,1.75){};
  }
\end{tikzpicture}}

%% file: figs/binary-tree-gadget.tex
\tikzstyle{filled vertex}  = [{circle,draw=blue,fill=black!50,inner sep=1.5pt}]  % regular vertex
\tikzstyle{empty vertex}  = [{circle, draw, fill = white, inner sep=1.75pt, line width=0.2mm}]
\centering 
%\resizebox{}{!}
{\begin{tikzpicture}[scale=.2] %6
    \node[empty vertex, label= right:\scriptsize{$\ r$}] (u_0_1) at  (18.5,99){};

    \node [empty vertex, label=right:\scriptsize{$a$}] (u_0) at (18.5, 102){};
    \draw [dotted, thick, line width=0.5mm] (u_0_1) -- (u_0);
    %level l+1
    %vertices of $B_{uv}$ starts
    \node[empty vertex] (u_1_1) at  (10.5,96){};  
   
   \draw[line width=0.5mm] (u_1_1) -- (u_0_1);				
   \node[empty vertex] (u_1_2) at  (26.5,96){};
   
    \draw[line width=0.5mm] (u_1_2) -- (u_0_1);
    \node[empty vertex] (u_2_1) at  (6.5,94){};
    \draw[line width=0.5mm] (u_2_1) -- (u_1_1);				
    \node[empty vertex] (u_2_2) at  (14.5,94){};
    \draw[line width=0.5mm] (u_2_2) -- (u_1_1);			
    \node[empty vertex] (u_2_3) at  (22.5,94){};
    \draw[line width=0.5mm] (u_2_3) -- (u_1_2);				
    \node[empty vertex] (u_2_4) at  (30.5,94){};
    \draw[line width=0.5mm] (u_2_4) -- (u_1_2);		

    \node[empty vertex] (u_3_1) at  (4.5,92){};
    \draw[line width=0.5mm] (u_3_1) -- (u_2_1);				
    \node[empty vertex] (u_3_2) at  (8.5,92){};
    \draw[line width=0.5mm] (u_3_2) -- (u_2_1);			
    \node[empty vertex] (u_3_3) at  (12.5,92){};
    \draw[line width=0.5mm] (u_3_3) -- (u_2_2);				
    \node[empty vertex] (u_3_4) at  (16.5,92){};
    \draw[line width=0.5mm] (u_3_4) -- (u_2_2);		
    \node[empty vertex] (u_3_5) at  (20.5,92){};
    \draw[line width=0.5mm] (u_3_5) -- (u_2_3);				
    \node[empty vertex] (u_3_6) at  (24.5,92){};
    \draw[line width=0.5mm] (u_3_6) -- (u_2_3);			
    \node[empty vertex] (u_3_7) at  (28.5,92){};
    \draw[line width=0.5mm] (u_3_7) -- (u_2_4);				
    \node[empty vertex] (u_3_8) at  (32.5,92){};
    \draw[line width=0.5mm] (u_3_8) -- (u_2_4);

   \node[empty vertex] (u_4_2) at  (4,88){};		
   \node[empty vertex] (u_4_3) at  (5,88){};				
    \node[empty vertex] (u_4_4) at  (6,88){};		
    \node[empty vertex] (u_4_5) at  (7,88){};			
    \node[empty vertex] (u_4_6) at  (8,88){};;			

    \node[empty vertex] (u_4_11) at  (13,88){};			
    \node[empty vertex] (u_4_12) at  (14,88){};			
    \node[empty vertex] (u_4_13) at  (15,88){};				
    \node[empty vertex] (u_4_14) at  (16,88){};	
    \node[empty vertex] (u_4_15) at  (17,88){};				
    \node[empty vertex] (u_4_16) at  (18,88){};	
    \draw[dotted,line width=0.5mm] (9.5, 88) -- (11.5, 88);
    
    \draw [dotted,line width=0.5mm] (4.5,89) -- (4.5,91.5);	
    \draw [dotted,line width=0.5mm] (8.5,89) -- (8.5,91.5);
    \draw [dotted,line width=0.5mm] (12.5,89) -- (12.5,91.5);
    \draw [dotted,line width=0.5mm] (16.5,89) -- (16.5,91.5);
    
    \draw[dotted,line width=0.5mm] (10,90) -- (12,90);				

    \node[empty vertex] (u_4_17) at  (19,88){};		
    \node[empty vertex] (u_4_18) at  (20,88){};		
    \node[empty vertex] (u_4_19) at  (21,88){};				
    \node[empty vertex] (u_4_20) at  (22,88){};		
    \node[empty vertex] (u_4_21) at  (23,88){};			
    \node[empty vertex] (u_4_22) at  (24,88){};;			

    \node[empty vertex] (u_4_27) at  (29,88){};			
    \node[empty vertex] (u_4_28) at  (30,88){};			
    \node[empty vertex] (u_4_29) at  (31,88){};				
    \node[empty vertex] (u_4_30) at  (32,88){};	
    \node[empty vertex] (u_4_31) at  (33,88){};				
    \draw[dotted,line width=0.5mm] (24.5, 88) -- (28.5, 88);

    \draw [dotted,line width=0.5mm] (20.5,89) -- (20.5,91.5);
    \draw [dotted,line width=0.5mm] (24.5,89) -- (24.5,91.5);
    \draw [dotted,line width=0.5mm] (28.5,89) -- (28.5,91.5);
    \draw [dotted,line width=0.5mm] (32.5,89) -- (32.5,91.5);

    \draw[dotted,line width=0.5mm] (25,90) -- (28,90);
 
   \draw[decoration={calligraphic brace, amplitude=6pt}, decorate, line width=0.5mm](3.75,88) -- (3.75,98) ;
   \node [label=left:\scriptsize{$h+1$}] (v3) at (3, 93) {};
 
 \draw[decoration={calligraphic brace, amplitude=14pt}, decorate, line width=0.5mm] (33.3,88) -- (3.5,88);     
 \node [label=below:\scriptsize{$2^h$ leaves}] (v1) at (18.5, 86.5) {};

    \node[label=right:\scriptsize{$level=0$}] (u_5_1) at  (35,99){};
   \draw [dashed,line width=0.1mm] (21.5, 99) -- (35.5, 99);    
    \node[label=right:\scriptsize{$level=1$}] (u_5_2) at  (35,96){};
    \draw [dashed,line width=0.1mm] (27.5, 96) -- (35.5, 96);      
    \node[label=right:\scriptsize{$level=l$}] (u_5_3) at  (35,88){};
    \draw [dashed,line width=0.1mm] (34, 88) -- (35.5, 88);  

    \node (top-left) at (-1, 99){};
    \node (bottom-right) at (42, 84) {};
    
   \node[rectangle, draw,dashed,inner sep=0pt, line width = 0.2mm][fit= (top-left) (bottom-right),  inner xsep=1.5ex, inner ysep=1.5ex] {};
    
   \end{tikzpicture}}

%% file: figs/BT-pair-gadget.tex
\tikzstyle{filled vertex}  = [{circle,draw=blue,fill=black!50,inner sep=1.5pt}]  % regular vertex
\tikzstyle{empty vertex}  = [{circle, draw, fill = white, inner sep=1.25pt}]
\centering \small
%\resizebox{}{!}
{\begin{tikzpicture}[scale=.4] %6
    \node[empty vertex, label= above:\footnotesize{$a$}] (u_0) at  (18.5,99.5){};
    \node[empty vertex, label= below:\footnotesize{$r_{ab}$}] (u_0_1) at  (18.5,98){};
    \draw [line width=0.5mm] (u_0) -- (u_0_1);
    %level l+1
    %vertices of $B_{uv}$ starts
    \node[empty vertex] (u_1_1) at  (14.5,97){};
    
    \draw [line width=0.5mm] (u_1_1) -- (u_0_1);				
    \node[empty vertex] (u_1_2) at  (22.5,97){};
    
    \draw [line width=0.5mm] (u_1_2) -- (u_0_1);
    \draw [dotted, line width=0.5mm] (14.5, 97) -- (12.5, 96);
    \draw [dotted, line width=0.5mm] (14.5, 97) -- (16.5, 96);
    \draw [dotted, line width=0.5mm] (22.5, 97) -- (20.5, 96);
    \draw [dotted, line width=0.5mm] (22.5, 97) -- (24.5, 96);

    \node[empty vertex, label=left:$l^1_{ab}$] (u_l_1) at  (10,94){};		
    \node[empty vertex, label=left:$l^2_{ab}$] (u_l_2) at  (12,94){};		
    \node[empty vertex] (u_l_3) at  (14,94){};	
    \node[empty vertex] (u_l_4) at  (16,94){};	   
    
    \node[empty vertex, label=right:$l^i_{ab}$] (u_l_5) at  (22,94){};		
    \node[empty vertex] (u_l_6) at  (24,94){};			
    \node[empty vertex] (u_l_7) at  (26,94){};
    \node[empty vertex] (u_l_8) at  (28,94){};

    \draw[dotted, line width=0.5mm] (10, 94) -- (10.5, 94.5);
    \draw[dotted, line width=0.5mm] (12, 94) -- (11.5, 94.5);   
    \draw[dotted, line width=0.5mm] (14, 94) -- (14.5, 94.5);
    \draw[dotted, line width=0.5mm] (16, 94) -- (15.5, 94.5);
    \draw[dotted, line width=0.5mm] (22, 94) -- (22.5, 94.5);
    \draw[dotted, line width=0.5mm] (24, 94) -- (23.5, 94.5);   
    \draw[dotted, line width=0.5mm] (26, 94) -- (26.5, 94.5);
    \draw[dotted, line width=0.5mm] (28, 94) -- (27.5, 94.5);
    \draw[dotted, line width=0.5mm] (17, 94) -- (21, 94);

    \node (B1_lu_corner) at (8, 98.1) {};
    \node (B1_rl_corner) at (32, 94) {};
    \node[label= right:\footnotesize{$BT_{ab}$}] (b_1) at  (30,98.25){};    
    \node  (B1-top-row) at (29.5, 98) {};
    \node  (B1_bottom_row) at (29.5, 94) {};
    \draw [dotted, line width=0.5mm] (u_0_1) -- (B1-top-row);
    \draw [dotted, line width=0.5mm] (u_l_8) -- (B1_bottom_row); 
    \draw[dotted, |<->|, line width=0.5mm] (29.5, 98) -- (29.5, 94);
    \node[label= right:\footnotesize{$h+1$}] at  (29.5,96){};    

    \node[rectangle, draw, ultra thick, dotted,inner sep=0pt, line width = 0.2mm][fit= (B1_lu_corner) (B1_rl_corner),  inner xsep=2ex, inner ysep=1.5ex] {};  
    
    % % Inverted tree
     \node[empty vertex, label= below:\footnotesize{$b$}] (v_0) at  (18.5,82.5){};   
    \node[empty vertex, label= above:\footnotesize{$r_{ba}$}] (v_0_1) at  (18.5,84){};
    \draw (v_0) -- (v_0_1);
    
    \node[empty vertex] (v_1_1) at  (14.5,85){};   

    \draw [line width=0.5mm] (v_1_1) -- (v_0_1);		

    \node[empty vertex] (v_1_2) at  (22.5,85){};
    
    \draw [line width=0.5mm] (v_1_2) -- (v_0_1);
    \draw [dotted, line width=0.5mm] (14.5, 85) -- (12.5, 86);
    \draw [dotted, line width=0.5mm] (14.5, 85) -- (16.5, 86);
    \draw [dotted, line width=0.5mm] (22.5, 85) -- (20.5, 86);
    \draw [dotted, line width=0.5mm] (22.5, 85) -- (24.5, 86);

    \node[empty vertex, label=left:$l^1_{ba}$] (v_u_1) at  (10,88){};		
    \node[empty vertex,  label=left:$l^2_{ba}$] (v_u_2) at  (12,88){};		
    \node[empty vertex] (v_u_3) at  (14,88){};	
    \node[empty vertex] (v_u_4) at  (16,88){};	   
    
    \node[empty vertex, label=right:$l^i_{ba}$] (v_u_5) at  (22,88){};		
    \node[empty vertex] (v_u_6) at  (24,88){};			
    \node[empty vertex] (v_u_7) at  (26,88){};
    \node[empty vertex] (v_u_8) at  (28,88){};

    \draw [dotted, line width=0.5mm] (10, 88) -- (10.5, 87.5);
    \draw [dotted, line width=0.5mm] (12, 88) -- (11.5, 87.5);   
    \draw [dotted, line width=0.5mm] (14, 88) -- (14.5, 87.5);
    \draw [dotted, line width=0.5mm] (16, 88) -- (15.5, 87.5);
    \draw [dotted, line width=0.5mm] (22, 88) -- (22.5, 87.5);
    \draw [dotted, line width=0.5mm] (24, 88) -- (23.5, 87.5);   
    \draw [dotted, line width=0.5mm] (26, 88) -- (26.5, 87.5);
    \draw[dotted, line width=0.5mm] (28, 88) -- (27.5, 87.5);
    \draw [dotted, line width=0.5mm] (17, 88) -- (21, 88);
    
    \node (B2_ll_corner) at (8, 83.9) {};
    \node (B2_rr_corner) at (32, 88) {};

    \node[label= right:\footnotesize{$BT_{ba}$}] (b_2) at  (30,83.75){};    
    \node (B2-top-row) at (29.5, 88) {};
    \node (B2-bottom-row) at (29.5, 84) {};
    \draw [dotted, line width=0.5mm] (v_0_1) -- (B2-bottom-row);
    \draw [dotted, line width=0.5mm] (v_u_8) -- (B2-top-row); 
    \draw[dotted, |<->|, line width=0.5mm] (29.5, 88) -- (29.5, 84);
    \node[label= right:\footnotesize{$h+1$}] at  (29.5,86){};    

    \node[rectangle, draw, ultra thick, dotted ,inner sep=0pt, line width = 0.2mm][fit= (B2_ll_corner) (B2_rr_corner),  inner xsep=2ex, inner ysep=1.5ex] {};  

    % Add t-gadgets
   \tgadgetrighthalf{(10,94)}{0.20}{0.20}{2.3}{2.40}{0.5}{0.5}{0}{0.4};
   \tgadgetlefthalf{(12,94)}{0.20}{0.20}{2.3}{2.40}{0.5}{0.5}{0}{0.4};
   \tgadgetrighthalf{(14,94)}{0.20}{0.20}{2.3}{2.40}{0.5}{0.5}{0}{0.4};
   \tgadgetlefthalf{(16,94)}{0.20}{0.20}{2.3}{2.40}{0.5}{0.5}{0}{0.4};

   \tgadgetrighthalf{(22,94)}{0.20}{0.20}{2.3}{2.40}{0.5}{0.5}{0}{0.4};
   \tgadgetlefthalf{(24,94)}{0.20}{0.20}{2.3}{2.40}{0.5}{0.5}{0}{0.4};
   \tgadgetrighthalf{(26,94)}{0.20}{0.20}{2.3}{2.40}{0.5}{0.5}{0}{0.4};
   \tgadgetlefthalf{(28,94)}{0.20}{0.20}{2.3}{2.40}{0.5}{0.5}{0}{0.4};
    
    \node (tgadgets_lu_corner) at (9, 94) {};
    
    \node (tgadgets_rb_corner) at (29.5, 88) {};
    \node[rectangle, draw,dashed,inner sep=0pt, line width = 0.2mm][fit= (tgadgets_lu_corner) (tgadgets_rb_corner),  inner xsep=2ex, inner ysep=1.5ex] {};  

    \node[label= right:\footnotesize{$T$-gadgets}] at  (30.2,92.5){};    
    \node[label= right:\footnotesize{$T(l_1, l_2)$}] at  (30.2,91.7){};    

    \node (uhalf_lu_corner) at (5.5, 100) {};  
    \node (uhalf_rb_corner) at (34.5, 91.8) {};
    \node[rectangle, draw,dashed,inner sep=0pt, color =red, line width = 0.2mm][fit= (uhalf_lu_corner) (uhalf_rb_corner),  inner xsep=2ex, inner ysep=1.5ex] {};  
    \node[color=red, label= right:\footnotesize{\textcolor{red}{$a_{exthalf}$}}] (a_1) at  (30,101.25){}; 

\node (uhalf_lu1_corner) at (6.5, 98.3) {};  
    \node (uhalf_rb1_corner) at (33.5, 92.1) {};
    \node[rectangle, draw,dashed,inner sep=0pt, color =blue, line width = 0.2mm][fit= (uhalf_lu1_corner) (uhalf_rb1_corner),  inner xsep=2ex, inner ysep=1.5ex] {};  
     \node[color=blue, label= right:\footnotesize{\textcolor{blue}{$a_{half}$}}] (a_1) at  (30,99.75){};   
   
    \node (vhalf_lu_corner) at (5.5, 90.4) {};  
    \node (vhalf_rb_corner) at (34.5, 82) {};
    \node[rectangle, draw,dashed,inner sep=0pt, color =red, line width = 0.2mm][fit= (vhalf_lu_corner) (vhalf_rb_corner),  inner xsep=2ex, inner ysep=1.5ex] {};  
   \node[label= right:\footnotesize{\textcolor{red}{$b_{exthalf}$}}] (b_1) at  (30,80.75){};       

   \node (vhalf_lu_corner) at (6.5, 90.1) {};  
    \node (vhalf_rb_corner) at (33.5, 83.75) {};
    \node[rectangle, draw,dashed,inner sep=0pt, color =blue, line width = 0.2mm][fit= (vhalf_lu_corner) (vhalf_rb_corner),  inner xsep=2ex, inner ysep=1.5ex] {};  
   \node[label= right:\footnotesize{\textcolor{blue}{$b_{half}$}}] (b_1) at  (30,82.5){};       
    \end{tikzpicture}}

%% file: figs/path-gadget.tex
\resizebox{0.5\linewidth}{!}{\tikzstyle{filled vertex}  = [{circle,draw=blue,fill=black!50,inner sep=1.5pt}]  % regular vertex
\tikzstyle{empty vertex}  = [{circle, draw, fill = white, inner sep=1.25pt}]
\centering \small
%\resizebox{}{!}
{\begin{tikzpicture}[scale=.5] %6		
    \node[filled vertex, label= below:\footnotesize{$a_1$}] (u_1_1) at  (3, 3){};
    \node[filled vertex, label= above:\footnotesize{$a_2$}] (u_1_2) at  (5, 3){};    
    \node[filled vertex] (u_1_3) at  (7, 3){};        
    \node[filled vertex, label= above:\footnotesize{$a_i$}] (u_1_4) at  (10, 3){};          
    \node[filled vertex] (u_1_5) at  (13, 3){};        
    \node[filled vertex, label= above:\footnotesize{$a_{(d-1)}$}] (u_1_6) at  (15, 3){};        
    \node[filled vertex, label= below:\footnotesize{$a_d$}] (u_1_7) at  (17, 3){};         
    
    \node[empty vertex,  label= below:\footnotesize{$b_2$}] (u_2_2) at  (5, 2){};    
    \node[empty vertex] (u_2_3) at  (7, 2){};       
    \node[empty vertex, label= below:\footnotesize{$b_i$}] (u_2_4) at  (10, 2){};       
    \node[empty vertex] (u_2_5) at  (13,2){};          
    \node[empty vertex,label= below:\footnotesize{$b_{(d-1)}$} ] (u_2_6) at  (15, 2){};      

    \draw [line width=0.5mm] (u_1_1) --  (u_1_2);    
    \draw [line width=0.5mm] (u_1_2) --  (u_1_3);    
    \draw [line width=0.5mm] [dotted] (u_1_3) --  (u_1_4);    
    \draw [dotted, line width=0.5mm] (u_1_4) --  (u_1_5);    
    \draw [line width=0.5mm] (u_1_5) --  (u_1_6);  
    \draw [line width=0.5mm] (u_1_6) --  (u_1_7);  
    
    \draw [line width=0.5mm] (u_2_2) --  (u_2_3);    
    \draw [dotted, line width=0.5mm] (u_2_3) --  (u_2_4);    
    \draw [dotted, line width=0.5mm] (u_2_4) --  (u_2_5);  
    \draw [line width=0.5mm] (u_2_5) --  (u_2_6); 
    
    \draw [line width=0.5mm] (u_1_1) -- (u_2_2);
    \draw [line width=0.5mm] (u_1_2) -- (u_2_2);
    \draw [line width=0.5mm] (u_1_3) -- (u_2_3);
    \draw [line width=0.5mm] (u_1_4)-- (u_2_4);
    \draw [line width=0.5mm] (u_1_5) -- (u_2_5);
    \draw [line width=0.5mm] (u_1_6) -- (u_2_6);
    \draw [line width=0.5mm] (u_1_7) -- (u_2_6);    
    
   \draw [dotted, <->, line width=0.5mm] (3, 4) -- (17, 4);
   \node[label= above:\footnotesize{$d$}] at  (10,3.75){};        

    \node (Left-top-corner) at (2.5, 4.5) {};
     \node (Right-bottom-corner) at (17.5, 0.5) {};   
     \node[rectangle, draw, ultra thick, dotted,inner sep=0pt, line width = 0.2mm][fit= (Left-top-corner) (Right-bottom-corner),  inner xsep=2ex, inner ysep=1.5ex] {};   
\end{tikzpicture}} }

%% file: figs/Y-gadget.tex
\resizebox{0.6\linewidth}{!}{\tikzstyle{filled vertex}  = [{circle,draw=blue,fill=black!50,inner sep=1.5pt}]  % regular vertex
\tikzstyle{empty vertex}  = [{circle, draw, fill = white, inner sep=1.25pt}]
\centering \small
%\resizebox{}{!}
{\begin{tikzpicture}[scale=.5] %6
    \node[empty vertex, label= left:\footnotesize{$x$}] (x) at  (3,6.5){};
    \node[empty vertex, label= right:\footnotesize{$y$}] (y) at  (13,6.5){};
    \draw[dotted, line width=0.5mm] (x) -- (y);    

    \node[filled vertex, label= left:\footnotesize{$x_a$}] (x_1) at  (5,5){};
    \node[filled vertex, label= {[xshift=-3mm, yshift=-4mm]:\footnotesize{$x_b$}}] (x_2) at  (7.5,2.5){};    
    \node[filled vertex, label= right:\footnotesize{$y_a$}] (y_1) at  (11,5){};
    \node[filled vertex, label= {[xshift=3mm, yshift=-4mm]:\footnotesize{$y_b$}}] (y_2) at  (8.5,2.5){};    
    
    \draw (x) -- (x_1);
    \draw (y) -- (y_1);   
  
    \pgadget{(5,5)}{0.25}{0.1}{3}{-45};
    \node [label=$P_x$] at (7,3) {};
  
  \draw[decoration={calligraphic brace, amplitude=12pt}, decorate, line width=1pt] (7.5,2.5) -- (5, 5); 
    \node [label=left:{$d_1$}] at (6,3){};
  
    \pgadget{(11,5)}{0.25}{0.1}{3}{-135};    
    \node [label=$P_y$] at (9,3) {};
   \node [empty vertex, label= left:\footnotesize{$z$}] (z) at (8, 1) {};
   \draw [line width=0.5mm] (x_2) -- (z);
   \draw [line width=0.5mm] (y_2) -- (z);
    \node[filled vertex, label= left:\footnotesize{$z_a$}] (z_a) at  (8, 0){};
    \draw [line width=0.5mm] (z) -- (z_a);
    
    \node[filled vertex, label= left:\footnotesize{$z_b$}] (z_b) at  (8,-6){};
    
    \pgadget{(8,0)}{0.25}{0.1}{5.40}{270};      
     \node [label=$P_z$] at (7.25,-4) {}; 

  \draw[decoration={calligraphic brace, amplitude=6pt}, decorate, line width=1pt] (6.75, -6) -- (6.75, 0); 

   \node [label=left:$d_2$] at (6.50, -3){};
    
 %   \node[empty vertex, label= below:$z_c$] (z_c) at  (10.5,-6){};
    \node[empty vertex] (z_c) at  (8,-8){};    
  %  \draw [line width=0.5mm] (z_b) -- (z_c);    
    \draw (z_b) -- (z_c); 
    
    \node (u_y_l_u) at (4, 4.75) {};
    \node (u_y_r_b) at (12, 2.25) {};
     \node[rectangle, draw,dashed,inner sep=0pt, line width = 0.2mm][fit= (u_y_l_u) (u_y_r_b),  inner xsep=2ex, inner ysep=1.5ex] {};  
    
    \node (l_y_l_u) at (6, -0.25) {};
    \node (l_y_r_b) at (8.5, -6.0) {};
     \node[rectangle, draw,dashed,inner sep=0pt, line width = 0.2mm][fit= (l_y_l_u) (l_y_r_b),  inner xsep=2ex, inner ysep=1.5ex] {};  

    \draw[->, line width=0.5mm] (10.25, 4) -- (16, 4);
    \draw[->, line width=0.5mm] (8.25, -2) -- (16, -2);
%    \draw (12, -6.25) --(12, -8); 

   % Path corresponding to upper portion of Y
 %  \node[empty vertex] (p1_r_1) at  (18, 6.5){};
   \node[filled vertex] (p1_r_2) at  (18, 6){};
   \node[filled vertex] (p1_r_3) at  (18, 5.5) {};
   \node[filled vertex] (p1_r_4) at  (18, 5){};
   \node[filled vertex] (p1_r_5) at  (18, 4){};   
   \node[filled vertex] (p1_r_6) at  (18, 3.5){};
   \node[filled vertex] (p1_r_7) at  (18, 3){};
%   \node[empty vertex] (p1_r_8) at  (18, 2.5){}; 
   \node[empty vertex] (p1_l_3) at  (17.5, 5.5) {};
   \node[empty vertex] (p1_l_4) at  (17.5, 5){};
   \node[empty vertex] (p1_l_5) at  (17.5, 4){};   
   \node[empty vertex] (p1_l_6) at  (17.5, 3.5){};  
 %  \draw (p1_r_1) -- (p1_r_2);
   \draw [line width=0.5mm] (p1_r_2) -- (p1_r_3);   
   \draw [line width=0.5mm] (p1_r_3) -- (p1_r_4);      
   \draw [dotted, line width=0.5mm] (p1_r_4) -- (p1_r_5);      
   \draw [line width=0.5mm] (p1_r_5) -- (p1_r_6);   
   \draw [line width=0.5mm] (p1_r_6) -- (p1_r_7);      
%   \draw (p1_r_7) -- (p1_r_8);  
   \draw [line width=0.5mm] (p1_l_3) -- (p1_l_4);      
   \draw [dotted, line width=0.5mm] (p1_l_4) -- (p1_l_5);      
   \draw [line width=0.5mm] (p1_l_5) -- (p1_l_6);    
   % horizontal lines
   \draw [line width=0.5mm] (p1_l_3) -- (p1_r_3);   
   \draw [line width=0.5mm] (p1_l_4) -- (p1_r_4);   
   \draw [line width=0.5mm] (p1_l_5) -- (p1_r_5);   
   \draw [line width=0.5mm] (p1_l_6) -- (p1_r_6); 
   %oblique lines
   \draw [line width=0.5mm] (p1_l_3) -- (p1_r_2);
   \draw [line width=0.5mm] (p1_l_6) -- (p1_r_7);   

    \node (p1_l_u) at (17, 6) {};
    \node (p1_r_b) at (19, 3) {};
    \node[rectangle, draw,dashed,inner sep=0pt, line width = 0.2mm][fit= (p1_l_u) (p1_r_b),  inner xsep=2ex, inner ysep=1.5ex] {};  
    
   %Path correspondoing to lower portion of Y
 %  \node[empty vertex] (p2_r_1) at  (18, 0.5){};
   \node[filled vertex] (p2_r_2) at  (18, 0){};
   \node[filled vertex] (p2_r_3) at  (18, -0.5) {};
   \node[filled vertex] (p2_r_4) at  (18, -1){};
   \node[filled vertex] (p2_r_5) at  (18, -2){};   
   \node[filled vertex] (p2_r_6) at  (18, -2.5){};
   \node[filled vertex] (p2_r_7) at  (18, -3){};
 %  \node[empty vertex] (p2_r_8) at  (18, -3.5){}; 
   \node[empty vertex] (p2_l_3) at  (17.5, -0.5) {};
   \node[empty vertex] (p2_l_4) at  (17.5, -1){};
   \node[empty vertex] (p2_l_5) at  (17.5, -2){};   
   \node[empty vertex] (p2_l_6) at  (17.5, -2.5){};  
 %  \draw (p2_r_1) -- (p2_r_2);
   \draw [line width=0.5mm] (p2_r_2) -- (p2_r_3);   
   \draw [line width=0.5mm] (p2_r_3) -- (p2_r_4);      
   \draw [dotted, line width=0.5mm] (p2_r_4) -- (p2_r_5);      
   \draw [line width=0.5mm] (p2_r_5) -- (p2_r_6);   
   \draw [line width=0.5mm] (p2_r_6) -- (p2_r_7);      
 %  \draw (p2_r_7) -- (p2_r_8);  
   \draw [line width=0.5mm] (p2_l_3) -- (p2_l_4);      
   \draw [dotted, line width=0.5mm] (p2_l_4) -- (p2_l_5);      
   \draw [line width=0.5mm] (p2_l_5) -- (p2_l_6);    
 %  horizontal lines
   \draw [line width=0.5mm] (p2_l_3) -- (p2_r_3);   
   \draw [line width=0.5mm] (p2_l_4) -- (p2_r_4);   
   \draw [line width=0.5mm] (p2_l_5) -- (p2_r_5);   
   \draw [line width=0.5mm] (p2_l_6) -- (p2_r_6); 
  % oblique lines
   \draw [line width=0.5mm] (p2_l_3) -- (p2_r_2);
   \draw [line width=0.5mm] (p2_l_6) -- (p2_r_7);   

    \node (p2_l_u) at (17, 0) {};
    \node (p2_r_b) at (19, -3) {};
    \node[rectangle, draw,dashed,inner sep=0pt, line width = 0.2mm][fit= (p2_l_u) (p2_r_b),  inner xsep=2ex, inner ysep=1.5ex] {};  

    \node (top_left) at (3.5, 5.3) {};
    \node (bottom_right) at (13, -6.5) {};
     \node[rectangle, draw,dashed,inner sep=0pt, line width = 0.2mm, color=blue][fit= (top_left) (bottom_right),  inner xsep=2ex, inner ysep=1.5ex] {};  
          
  \end{tikzpicture}}}

%% file: figs/tail-gadget.tex
\resizebox{0.5\linewidth}{!}{\tikzstyle{filled vertex}  = [{circle,draw=blue,fill=black!50,inner sep=1.5pt}]  % regular vertex
\tikzstyle{empty vertex}  = [{circle, draw, fill = white, inner sep=1.5pt}]
  \centering \small
\resizebox{.8\textwidth}{!}
{\begin{tikzpicture} %6
{
   \node[filled vertex, label=below:\large{$v_9$}] (v9)  at  (15,10){};
   \node[filled vertex,, label=below:\large{$v_8$}] (v8) at  (16,10){};
   \node[filled vertex, label={[xshift=-3mm, yshift=-6.5mm]:\large$v_7$}] (v7) at  (17,10){};
   \node[filled vertex, label=below:\large{$v_6$}] (v6) at  (18,10){};
   \node[filled vertex, label=below:\large{$v_5$}] (v5) at  (19,10){};
   \node[filled vertex, label=below:\large{$v_4$}] (v4) at  (20,10){};
   \node[filled vertex, label={[xshift=-3mm, yshift=-6.5mm]:\large$v_3$}] (v3) at  (21,10){};
   \node[filled vertex,, label=below:\large{$v_2$}] (v2) at  (22,10){};
   \node[filled vertex,label=below:\large{$v_1$}] (v1) at  (23,10){};

    \node[empty vertex, label=below:\large{q}] (q) at  (17,9){};
    \node[empty vertex, label=below:\large{p}] (p) at  (21,9){};
    \node[empty vertex,label=above:\large{r}] (r) at  (20,11){};

    \foreach \i in {9,..., 2} {
        \pgfmathsetmacro{\p}{int(\i - 1)}
            \draw [line width=0.5mm] (v\i) -- (v\p);
     }
            
    \draw [line width=0.5mm] (q) -- (v9);
    \draw [line width=0.5mm] (q) -- (v7);      
    \draw [line width=0.5mm] (q) -- (v5);    

    \draw [line width=0.5mm] (r) -- (v8);
    \draw [line width=0.5mm] (r) -- (v6);      
    \draw [line width=0.5mm] (r) -- (v1);    

    \draw [line width=0.5mm] (p) -- (v2);
    \draw [line width=0.5mm] (p) -- (v3);      
    \draw [line width=0.5mm] (p) -- (v4);       
          
    \node[] (a)  at  (13.5,10){};
    \node[] (b) at  (25,10){};    

    \draw [dotted, line width=0.5mm] (a) -- (v9);
    \draw [dotted, line width=0.5mm] (v1) -- (b);   
    \node (tail_left_upper) at (14.5, 11.5) {};
    \node (tail_right_lower) at (24, 8.5) {};

    \node[rectangle, draw,dashed,inner sep=0pt, line width = 0.2mm][fit= (tail_left_upper) (tail_right_lower),  inner xsep=2ex, inner ysep=0.5ex] {};     

    \node (pt3_lu) at (15, 10.3) {};
    \node (pt3_rb) at (19, 9.7) {};
   \node[rectangle, draw,dashed,inner sep=0pt, line width = 0.2mm][fit= (pt3_lu) (pt3_rb),  inner xsep=2ex, inner ysep=1ex] {};  
   \node[label=below:{$PT_3$}] at (15.5, 10.7) {};

    \node (pt2_lu) at (20, 10.3) {};
    \node (pt2_rb) at (22, 9.7) {};
    \node[rectangle, draw,dashed,inner sep=0pt, line width = 0.2mm][fit= (pt2_lu) (pt2_rb),  inner xsep=2ex, inner ysep=1ex] {};  
    \node[label=below:{$PT_2$}] at (20.5, 10.7) {};  

    \node (pt1_lu) at (23, 10.3) {};
    \node (pt1_rb) at (23.4, 9.7) {};
    \node[rectangle, draw,dashed,inner sep=0pt, line width = 0.2mm][fit= (pt1_lu) (pt1_rb),  inner xsep=2ex, inner ysep=1ex] {};  
   \node[label=below:{$PT_1$}] at (23.2, 10.7) {};  

}\end{tikzpicture}}}

%% file: figs/cycle-gadget-new-2.tex
\tikzstyle{filled vertex}  = [{circle,draw=blue,fill=black!50,inner sep=1.5pt}]  % regular vertex
\tikzstyle{empty vertex}  = [{circle, draw, fill = white, inner sep=1.5pt}]
  \centering \small
\resizebox{\textwidth}{!}{\begin{tikzpicture} %6
{
 %   \node[empty vertex, label=left:\footnotesize{$c_{n^2}$}] (c_1) at  (13, -2){};
 \node [empty vertex] (c_0) at (10, -2) {};
   \node[filled vertex, label={[xshift = -1mm, yshift = -8mm]:\huge{$v_{m^2}$}}] (c_1) at  (13, -2){};
   \draw [dotted, line width=0.5mm] (c_0) -- (c_1);
    \node[filled vertex] (c_2) at  (20, -2){};
    \node[filled vertex] (c_1') at  (14, -2){};
    \draw [line width=0.5mm] (c_1) -- (c_1');
   
    \pgadget{(14,-2)}{0.4}{0.4}{5.2}{0};

    \draw [dotted, <->, line width=0.5mm] (14, -1.6) -- (20, -1.6);
    \node[label=above:\huge{$d = 2m-2$}] (p_1_a) at  (17.5, -1.6){};    
    \node[label=below:\huge{$P_{m}$}] (p_1_b) at  (17, -2.25){};
    \draw [dotted, <->, line width=0.5mm] (13, -3.5) -- (20, -3.5);
    \node[label=below:\huge{$d=2m-1$}] (p_1_bb) at  (17, -3.5){};   
    
    \node[filled vertex] (c_3) at  (21, -2){};
    \draw [line width=0.5mm] (c_2) -- (c_3); 
 
   \pgadget{(21,-2)}{0.4}{0.4}{4.8}{0};

    \draw [dotted, <->, line width=0.5mm] (21, -1.6) -- (26.5, -1.6);
    \node[label=above:\huge{$d=2m-3$}] (p_2_a) at  (23.5, -1.6){};   
    \node[label=below:\huge{$P_{m-1}$}] (p_2_b) at  (23.5,-2.25){};
    
    \node[filled vertex] (c_4) at  (26.6, -2){};   
    \node[filled vertex] (c_5) at  (27.6, -2){};  
    \draw [line width=0.5mm] (c_4) -- (c_5);
    \pgadget{(27.6,-2)}{0.4}{0.4}{4}{0};
    \draw [dotted, <->, line width=0.5mm] (27.5, -1.6) -- (32.25, -1.6);
    \node[label=above:\huge{$d=2m-5$}] (p_3_a) at  (30, -1.6){};   
    \node[label=below:\huge{$P_{m-2}$}] (p_3_b) at  (30, -2.25){};  
    
    \node[filled vertex] (c_6) at  (32.3, -2){};
     \node[filled vertex] (c_7) at  (33, -2){};   
    \draw [line width=0.5mm] (c_6) -- (c_7); 

    \node[filled vertex] (c_8) at  (34, -2){}; 
     \draw [dotted, line width=0.5mm] (c_7) -- (c_8);  
    \node[filled vertex] (c_9) at  (36.5, -2){}; 
     \draw [line width=0.5mm] (c_8) -- (c_9);  
   \pgadget{(34,-2)}{0.4}{0.4}{1.75}{0};
    \draw [dotted, <->, line width=0.5mm] (34, -1.6) -- (36.5, -1.6);
    \node[label=above:\huge{$d=7$}] (p_4_a) at  (35.25, -1.6){};   
    \node[label=below:\huge{$P_4$}] (p_4_b) at  (35.25,-2.25){};   
     
   % \node[filled vertex] (c_10) at  (37, -2){}; 
    \node[filled vertex, label=below:\huge{$v_9$}] (c_10) at  (37, -2){};  
     \draw [line width=0.5mm] (c_9) -- (c_10);  

     \node[filled vertex] (c_11) at  (38, -2){};       
     \draw [line width=0.5mm] (c_10) -- (c_11);   

      \node[filled vertex] (c_12) at  (39, -2){};       
      \draw [line width=0.5mm] (c_11) -- (c_12);       

     \node[filled vertex] (c_13) at  (40, -2){};       
    \draw [line width=0.5mm] (c_12) -- (c_13);       

     \node[filled vertex] (c_14) at  (41, -2){};       
     \draw [line width=0.5mm] (c_13) -- (c_14);   

    \node[filled vertex] (c_15) at  (42, -2){};       
    \draw [line width=0.5mm] (c_14) -- (c_15);       

    \node[filled vertex] (c_16) at  (43, -2){};       
    \draw [line width=0.5mm] (c_15) -- (c_16);       

    \node[filled vertex] (c_17) at  (44, -2){};       
    \draw [line width=0.5mm] (c_16) -- (c_17);       

    \node[filled vertex, label=below:\huge{$v_1$}] (c_18) at  (45, -2){};       
    \draw [line width=0.5mm] (c_17) -- (c_18);       

    \node[empty vertex, label=below:\huge{q}] (c_12_b) at  (39, -3){};       
     \draw [line width=0.5mm] (c_12_b) -- (c_12);       
     \draw [line width=0.5mm] (c_12_b) -- (c_10);
     \draw [line width=0.5mm] (c_12_b) -- (c_14);     

 \draw[decoration={calligraphic brace, amplitude=14pt}, decorate, line width=1pt] (41, -3.5) -- (37,-3.5);      
     \node[label=below:\huge{$PT_3$}]  at (39, -4){};

     \node[empty vertex, label=below:\huge{p}] (c_16_b) at  (43, -3){};       
     \draw [line width=0.5mm] (c_16_b) -- (c_16);       
     \draw [line width=0.5mm] (c_16_b) -- (c_15);
     \draw [line width=0.5mm] (c_16_b) -- (c_17);       
 
    \draw[decoration={calligraphic brace, amplitude=13pt}, decorate, line width=1pt] (44, -3.5) -- (42,-3.5);      
     \node[label=below:\huge{$PT_2$}]  at (43, -4){};

     \draw[decoration={calligraphic brace, amplitude=8pt}, decorate, line width=1pt] (45.5, -2.5) -- (44.5,-2.5);      
     \node[label=below:\huge{$PT_1$}]  at (45, -2.75){};    

     \node[empty vertex, label={[xshift = -5mm, yshift = -2mm]:\huge{r}}] (c_15_u) at  (42, -1){};       
     \draw [line width=0.5mm] (c_15_u) -- (c_13);       
     \draw [line width=0.5mm] (c_15_u) -- (c_11);
     \draw [line width=0.5mm] (c_15_u) -- (c_18); 

     \tailgadget{(37,-2)}{1}{1};  
     \draw [line width=0.5mm] (c_1) to[bend left=28] (c_18);
     \draw [dotted, line width=0.5mm, <->] (13, 3) -- (45, 3);
     \node[label=above:\huge{$m^2$}] (cycle_size) at  (30, 3){};

      \node (cycle_left_upper) at (12, 3.75) {};
      \node (cycle_right_lower) at (47, -5) {};

     \node[rectangle, draw,dashed,inner sep=0pt, line width = 0.2mm][fit= (cycle_left_upper) (cycle_right_lower),  inner xsep=2ex, inner ysep=1.5ex] {}; 
}\end{tikzpicture}}

%% file: summary-table.tex
\begin{table}[]
    \centering
\small
\setlength{\tabcolsep}{3pt}
\renewcommand{\arraystretch}{0.5}
\begin {tabular} {|l|l|}
\hline
\multicolumn{1}{|c|}{Notation} & \multicolumn{1}{c|} {Definition/Description} \\
\hline

Steps[u,v]&\makecell[tl]{Number of steps required for the fire originating \\from $u$ to burn $v$ including $u$ and $v$}\\

\hline

\makecell[tl]{a)$T_{half-ab}$ \\\\b) $T^i_{half-ab}$} &\makecell[tl] {a)The subgraph induced by the vertices on the fixed arm starting \\from $a$ to $\mathit{jn}_{pq}$, together with the $l_2$ vertices on the \\floating arm from $\mathit{jn}_{pq}$ to $\mathit{tip}_{pq}$ (Figure~\ref{fig:t-gadget})\\
b) The $T_{half-ab}$ in the $i^{th}$ $T$-gadget of the $BTP_{ab}$ gadget } \\
\hline
\makecell[tl]{a) $tip_{ab} $\\b) $tip^i_{ab} $\\c) $Tips_{ab}$\\d) $TIPS_{ab}$} &\makecell[tl]{a) Any tip in $T_{half-ab}$\\b)  Any tip in $T^i_{half-ab}$ \\c)The set of all tips of $T$-gadgets in $T^i_{half-ab}$, for $1\leq i\leq 2^h$ \\d $Tips_{ab} \cup Tips_{ba}$. Note: $TIPS_{ab} = TIPS_{ba}$} \\
\hline
\makecell[tl]{a) $a_{half}$ in $BTP_{ab}$\\ b) $a_{exthalf}$ in $BTP_{ab}$}&\makecell[tl]{a) The subgraph induced by $V(BT_{ab}\cup T_{half-ab})$\\ b)The subgraph induced by  $V(\{a,b\}\cup a_{half}$)} \\
\hline
\makecell[tl]{a) StartBlock\\ b) MiddleBlock\\c) EndBlock}&\makecell[tl]{a) The first $s$ sources of a burning sequence $B$\\b) The next $m$ sources after the StartBlock of $B$ \\c) The remaining burning sources after MiddleBlock \\(typically the last $l_1+l_2$ sources) in $B$}\\

\hline
\makecell[tl]{a) $trunk(C)$\\b) $trunk'(C)$}&\makecell[tl]{a) The longest induced cycle of $m^2$ vertices in a cycle gadget $C(m)$\\b) The shortest cycle containing $v_1$ and $v_{m^2}$ in $C(m)$} 
\\
\hline
$FixedOverlap$&\makecell[tl]{The maximum number of vertices in $trunk'(C)$ that may get\\burned more than once by a burning sequence that burns $trunk(C)$\\ which is equal to 5.}
\\
\hline
\makecell[tl]{a) $H_{core}$\\\\b) $H_{annexe}$}&\makecell[tl]{a) Union of all $BTP_{uv}$-gadgets along with the pairs of vertices $u,v\in V(H)$\\ corresponding to $u,v\in V(G')$ in which they are attached\\b) Union of the $Y$-gadget and $C$-gadget in $H$} 
\\
\hline
\makecell[tl]{a) $Dom_u$\\\\b) Owner($v$)}&\makecell[tl]{a)If $u$ is not $x$ or $y$, then union of $u_{exthalf}$ for all the three  $BTP$-gadgets\\ associated with $u$. If $u$ is $x$ (or $y$), then union of $u_{exthalf}$ of two \\associated $BTP$-gadgets and $P_x$ (resp $P_y$)\\b) If $v \in Dom_u$, then $u$ is owner of $v$. }
\\
\hline
\makecell[tl]{a) InsideDomains\\b) OutsideDomain}&\makecell[tl]{a) Union of domains of all vertices (for which domain is defined)\\b) The vertices of $H$ that do not belong to InsideDomains} 
\\
\hline
\makecell[tl]{a) Represented edge\\\\b) Unrepresented edge}&\makecell[tl]{a) $(u,v)\in E(G')$ such that there is a burning source in $Dom_u$ or $Dom_v$ \\in the burning sequence \\
b) When an edge $(u,v)\in E(G')$ is not represented} 
\\
\hline
\makecell[tl]{a) $Proj(z_{(i,j)}, G_k)$\\b) $Proj(B, G_k)$}&\makecell[tl]{a) Projection of $z_{(i,j)}$ onto $G_k$, that is, $z_{(i,k)}$\\b) Sequence formed by taking projection of every vertex in $B$\\  onto $G_k$ in the same order} 
\\
\hline
\makecell[tl]{a) BL(G,B)\\b) UB(G, B)}&\makecell[tl]{a) The set of vertices in $G$ burned in the last step by\\ the burning sequence $B$\\b) The set of vertices in $G$ such that each vertex is burned \\by a unique burning source} \\
\hline
\end {tabular}
\\
\caption{Summary of important notations with brief descriptions}
\label{tab:lists}
\end{table}

%% file: figs/uv-subgraph-new.tex
\tikzstyle{filled vertex}  = [{circle,draw=blue,fill=black!50,inner sep=1.5pt}]  % regular vertex
\tikzstyle{empty vertex}  = [{circle, draw, fill = white, inner sep=1.25pt}]
\centering \small
%\resizebox{}{!}
{\begin{tikzpicture}[scale=.5] %6

\draw plot [smooth cycle, tension=.8] 
    coordinates {(-2.5,1.5)(-1.5,.5)(-.5,-.5)(1,0)(2,2)(0,3)};
    \node[empty vertex, label= left:\footnotesize{$u$}] (u_g) at  (0.5,1){};
    \node[empty vertex, label= left:\footnotesize{$v$}] (v_g) at  (0.5,0){};  
    \draw (u_g) -- (v_g);

    \draw[->, line width=0.5mm] (0.8, 0.5) -- (3, 0.5);

    \node[empty vertex, label= left:\footnotesize{$u$}] (u_h) at  (8,5){};
    \node[empty vertex, label= left:\footnotesize{$v$}] (v_h) at  (8,-5){};   
    \node[empty vertex, label= left:\footnotesize{$r_{uv}$}] (u_v) at  (8,4){};
    \node[empty vertex, label= left:\footnotesize{$r_{vu}$}] (v_u) at  (8,-4){}; 
    \draw [line width=0.5mm] (u_h) -- (u_v);
    \draw [line width=0.5mm] (v_h) -- (v_u);

    \node[empty vertex] (b_1_1) at  (5,2){};
    \node[empty vertex] (b_1_2) at  (7,2){};
    \node[empty vertex] (b_1_3) at  (9,2){};    
    \node[empty vertex] (b_1_4) at  (11,2){};
    
    \node[empty vertex] (b_2_1) at  (5,-2){};
    \node[empty vertex] (b_2_2) at  (7,-2){};
    \node[empty vertex] (b_2_3) at  (9,-2){};    
    \node[empty vertex] (b_2_4) at  (11,-2){};

    \draw [line width=1mm] (u_v) -- (b_1_1);
    \draw [dotted, line width=0.5mm] (b_1_1) -- (b_1_2);
    \draw [dotted, line width=0.5mm] (b_1_2) -- (b_1_3);    
    \draw [dotted, line width=0.5mm] (b_1_3) -- (b_1_4);      
    \draw [line width=0.5mm] (u_v) -- (b_1_4);     
     
    \draw  [line width=1mm] (v_u) -- (b_2_1);
    \draw [dotted, line width=0.5mm] (b_2_1) -- (b_2_2);
    \draw [dotted, line width=0.5mm] (b_2_2) -- (b_2_3);    
    \draw [dotted, line width=0.5mm] (b_2_3) -- (b_2_4);        
    \draw [line width=0.5mm] (v_u) -- (b_2_4);   

    \emptytgadgetrightfull{(5,2)}{0.2}{0.2}{3.6}{1.6}{0.5}{0.5}{0};  
    
    \tgadgetleftfull{(7,2)}{0.2}{0.2}{3.6}{1.6}{0.5}{0.5}{0};    
    \tgadgetrightfull{(9,2)}{0.2}{0.2}{3.6}{1.6}{0.5}{0.5}{0};
    \tgadgetleftfull{(11,2)}{0.2}{0.2}{3.6}{1.6}{0.5}{0.5}{0}; 
    
    \draw [dotted, line width=0.5mm] (6.5, 0) -- (7.5, 0);
    \draw [dotted, line width=0.5mm] (8.5, 0) -- (9.5, 0);    

    \node (b_1_u_l) at (4.5, 3.75) {};
    \node (b_1_l_r) at (11.5, 2) {};    
   \node[rectangle, draw,dashed,inner sep=0pt, line width = 0.2mm][fit= (b_1_u_l) (b_1_l_r),  inner xsep=2ex, inner ysep=1.5ex] {};  
   
    \node (b_2_l_l) at (4.5, -3.75) {};
    \node (b_2_u_r) at (11.5, -2) {};    
   \node[rectangle, draw,dashed,inner sep=0pt, line width = 0.2mm][fit= (b_2_l_l) (b_2_u_r),  inner xsep=2ex, inner ysep=1.5ex] {};  

    \node (t_l_u) at (5, 2) {};
    \node (t_r_l) at (11, -2) {};    
   \node[rectangle, draw,dashed,inner sep=0pt, line width = 0.2mm][fit= (t_l_u) (t_r_l),  inner xsep=2ex, inner ysep=1.5ex] {};

    \draw[->, line width=0.5mm] (12.25, 3) -- (14, 4);
    \draw[->, line width=0.5mm] (12.25, -3) -- (14, -4);
    \draw[->, line width=0.5mm] (11.25, 0) -- (14, 0);

    \node [empty vertex] (r1) at (20, 8) {};
    \node [empty vertex] (r11) at (18.5, 7.7) {};
    \node [empty vertex] (r12) at (21.5, 7.5) {};
    \node [empty vertex] (r111) at (17.5, 7) {};
    \node [empty vertex] (r112) at (19.5, 7) {};
    \node [empty vertex] (r121) at (20.5, 7) {};
    \node [empty vertex] (r122) at (22.5, 7) {};

    \draw [line width=0.5mm]  (r1) -- (r11);
    \draw [line width=0.5mm]  (r1) -- (r12);
    \draw [line width=0.5mm]  (r11) -- (r111);
    \draw [line width=0.5mm]  (r11) -- (r112);
    \draw [line width=0.5mm] (r12) -- (r121);
    \draw [line width=0.5mm]  (r12) -- (r122);

    \draw [dotted] (17.5, 6.5) -- (23, 6.5);

    \node [empty vertex] (r1l1) at (16, 5.5) {};
    \node [empty vertex] (r1l2) at (18, 5.5) {};
    \node [empty vertex] (r1l3) at (22.5, 5.5) {};      
    \node [empty vertex] (r1l4) at (24.5, 5.5) {};   

    \draw [dotted, line width=0.5mm] (r1l1) -- (17, 6);
    \draw [dotted, line width=0.5mm] (r1l2) -- (17, 6);    
    \draw [dotted, line width=0.5mm] (r1l3) -- (23.5, 6);
    \draw [dotted, line width=0.5mm] (r1l4) -- (23.5, 6);      
    \draw [dotted, line width=0.5mm] (r1l2) -- (r1l3);

   \draw[decoration={calligraphic brace, amplitude=4pt}, decorate, line width=1pt] (25, 8.25) -- (25, 5.25); 
    \node [label=right:{h+1}] at (25.25, 6.5){};

    \node (r_1_u_l) at (15.5, 8.5) {};
    \node (r_1_b_r) at (26.5, 5) {};    
   \node[rectangle, draw,dashed,inner sep=0pt, line width = 0.2mm][fit= (r_1_u_l) (r_1_b_r),  inner xsep=2ex, inner ysep=1.5ex] {};  
   \node[label=left:$BT_{uv}$] at (15, 6.5){};   
    % r2
    \node [empty vertex] (r2) at (20, -8) {};
    \node [empty vertex] (r21) at (18.5, -7.7) {};
    \node [empty vertex] (r22) at (21.5, -7.5) {};
    \node [empty vertex] (r211) at (17.5, -7) {};
    \node [empty vertex] (r212) at (19.5, -7) {};
    \node [empty vertex] (r221) at (20.5, -7) {};
    \node [empty vertex] (r222) at (22.5, -7) {};

    \draw [line width=0.5mm]  (r2) -- (r21);
    \draw [line width=0.5mm]  (r2) -- (r22);
    \draw [line width=0.5mm]  (r21) -- (r211);
    \draw [line width=0.5mm]  (r21) -- (r212);
    \draw [line width=0.5mm]  (r22) -- (r221);
    \draw [line width=0.5mm]  (r22) -- (r222);

    \draw [dotted] (17.5, -6.5) -- (23, -6.5);

    \node [empty vertex] (r2l1) at (16, -5.5) {};
    \node [empty vertex] (r2l2) at (18, -5.5) {};
    \node [empty vertex] (r2l3) at (22.5, -5.5) {};      
    \node [empty vertex] (r2l4) at (24.5, -5.5) {};   

    \draw [dotted, line width=0.5mm] (r2l1) -- (17, -6);
    \draw [dotted, line width=0.5mm] (r2l2) -- (17, -6);    
    \draw [dotted, line width=0.5mm] (r2l3) -- (23.5, -6);
    \draw [dotted, line width=0.5mm] (r2l4) -- (23.5, -6);      
    \draw [dotted, line width=0.5mm] (r2l2) -- (r2l3);

   \draw[decoration={calligraphic brace, amplitude=4pt}, decorate, line width=1pt] (25, -5.25) -- (25, -8.25); 
    \node [label=right:{h+1}] at (25.25, -6.5){};
    
    \node (r_2_l_l) at (15.5, -8.5) {};
    \node (r_2_u_r) at (26.5, -5) {};    
   \node[rectangle, draw,dashed,inner sep=0pt, line width = 0.2mm][fit= (r_2_l_l) (r_2_u_r),  inner xsep=2ex, inner ysep=1.5ex] {};  

  \node[label=left:$BT_{vu}$] at (15, -6.5){};

    % t-gadget
  \node [empty vertex] (f1) at (17, 3.5) {}; 
  \node [filled vertex] (fl1) at (16, 2.5) {};   
  \node [filled vertex] (fl2) at (16, 1.5) {};
  \node [filled vertex] (fl3) at (16, 0.5) {};  
  \node [filled vertex] (fl4) at (16, -0.5) {}; 
  \node [filled vertex] (fl5) at (16, -1.5) {};   
  \node [filled vertex] (fl6) at (16, -2.5) {};
  \node [empty vertex] (f2) at (17, -3.5) {}; 
  \node [filled vertex] (fr1) at (18, 2.5) {};   
  \node [filled vertex] (fr2) at (18, 1.5) {};
  \node [filled vertex] (fr3) at (18, 0.5) {};  
  \node [filled vertex] (fr4) at (18, -0.5) {}; 
  \node [filled vertex] (fr5) at (18, -1.5) {};   
  \node [filled vertex] (fr6) at (18, -2.5) {};  
     
  \draw [line width=0.5mm] (f1) -- (fl1);
  \draw [line width=0.5mm] (f1) -- (fr1);  
  \draw [line width=0.5mm] (fl1) -- (fr1);

  \draw [dotted, line width=0.5mm] (fl1) -- (fl2);  
  \draw [dotted, line width=0.5mm] (fr1) -- (fr2);
  \draw [line width=0.5mm] (fl2) -- (fr2);

  \draw [line width=0.5mm] (fl2) -- (fl3);  
  \draw [line width=0.5mm] (fr2) -- (fr3);
  \draw [line width=0.5mm] (fl3) -- (fr3);  

  \draw[decoration={calligraphic brace, amplitude=4pt}, decorate, line width=1pt] (18.5, 2.5) -- (18.5,0.5); 
 \node [label=right:{$(cn'-2h)/2$}] at (18.5, 1.5){};  

  \draw [line width=0.5mm] (fl3) -- (fl4);

  \draw [line width=0.5mm] (fl4) -- (fl5);  
  \draw [line width=0.5mm] (fr4) -- (fr5);
  \draw [line width=0.5mm] (fl4) -- (fr4);  
  \draw [line width=0.5mm] (fl5) -- (fr5);  

  \draw [dotted, line width=0.5mm] (fl5) -- (fl6);  
  \draw [dotted, line width=0.5mm] (fr5) -- (fr6);
  \draw [line width=0.5mm] (fl6) -- (fr6);

  \draw [line width=0.5mm] (fl6) --  (f2);
  \draw [line width=0.5mm] (fr6) --  (f2); 

  \node [empty vertex] (hu1) at (19, 0.5) {};  
  \node [empty vertex] (hb1) at (19, -0.5) {}; 
  \node [empty vertex] (hu2) at (20, 0.5) {};  
  \node [empty vertex] (hb2) at (20, -0.5) {}; 
  \node [empty vertex] (hu3) at (21, 0.5) {};  
  \node [empty vertex] (hb3) at (21, -0.5) {}; 
  \node [empty vertex] (hu4) at (22, 0.5) {};  
  \node [empty vertex] (hb4) at (22, -0.5) {}; 

  \draw[decoration={calligraphic brace, amplitude=4pt}, decorate, line width=1pt] (22, -0.75) -- (19,-0.75); 
 \node [label=below:{$cn'/2+2$}] at (21, -0.75){};  
  
  \draw [line width=0.5mm] (fr3) -- (hu1);
  \draw [line width=0.5mm] (fr4) -- (hb1); 
  \draw [line width=0.5mm] (hu1) -- (hb1);

  \draw [dotted, line width=0.5mm] (hu1) -- (hu2);
  \draw [dotted, line width=0.5mm]   (hb1) -- (hb2);   
  \draw [line width=0.5mm] (hu2) -- (hb2);  
  
  \draw [line width=0.5mm] (hu2) -- (hu3);
  \draw [line width=0.5mm] (hb2) -- (hb3);   

  \draw [line width=0.5mm] (hu3) -- (hu4);
  \draw [line width=0.5mm] (hb3) -- (hb4);   
  \draw [line width=0.5mm] (hu4) -- (hb4); 
  \draw [line width=0.5mm] (hu3) -- (hb4);
  \draw [line width=0.5mm] (hb3) -- (hu4);

   \node (t_l_u) at (15.5, 2.5) {};
   \node (t_r_l) at (24.5, -2.5) {};    
   \node[rectangle, draw,dashed,inner sep=0pt, line width = 0.2mm][fit= (t_l_u) (t_r_l),  inner xsep=2ex, inner ysep=1.5ex] {};     
   
   \node[label=right:$T_{i}-gadget$] at (25, 0){};  
 \end{tikzpicture}}

%% file: figs/xy-subgraph.tex
\tikzstyle{filled vertex}  = [{circle,draw=blue,fill=black!50,inner sep=1.5pt}]  % regular vertex
\tikzstyle{empty vertex}  = [{circle, draw, fill = white, inner sep=1.25pt}]
\centering \small
\resizebox{\textwidth}{!}
{\begin{tikzpicture}[scale=.7] %6

\draw plot [smooth cycle, tension=.8] 
    coordinates {(-4.5,1.5)(-3.5,.5)(-2.5,-.5)(-1,0)(0,2)(-2,3)};
    \node[empty vertex, label= left:\footnotesize{$x$}] (x_g) at  (-1.5,1){};
    \node[empty vertex, label= left:\footnotesize{$y$}] (y_g) at  (-1.5,0){};  
    \draw [line width=0.5mm] (x_g) -- (y_g);

    \draw[->, line width=0.5mm] (-0.1, 0.9) -- (1, 0.9);

    \node[empty vertex, label= left:\footnotesize{$x$}] (x) at  (4,6){};
    \node[empty vertex, label= right:\footnotesize{$y$}] (y) at  (12,6){};
    \draw[dotted, line width=0.5mm] (x) -- (y);

    \node[filled vertex, label= left:\footnotesize{$x_a$}] (x_1) at  (5,5){};
    \node[filled vertex, label= {[xshift=-3mm, yshift=-4mm]:\footnotesize{$x_b$}}] (x_2) at  (7.5,2.5){};    
    \node[filled vertex, label= right:\footnotesize{$y_a$}] (y_1) at  (11,5){};
    \node[filled vertex, label= {[xshift=3mm, yshift=-4mm]:\footnotesize{$y_b$}}] (y_2) at  (8.5,2.5){};    
 %   \draw[dotted] (x_h) -- (y_h);    
    \draw [line width=0.5mm] (x) -- (x_1);
    \draw [line width=0.5mm] (y) -- (y_1);   
  
    \pgadget{(5,5)}{0.25}{0.1}{3}{-45};
    \node [label=$P_x$] at (7,3) {};
  
  \draw[decoration={calligraphic brace, amplitude=12pt}, decorate, line width=1pt] (7.5,2.5) -- (5, 5); 
    \node [label=left:\tiny{$\lfloor n'/4 \rfloor+cn'/2$}] at (6,3){};
    
    \pgadget{(11,5)}{0.25}{0.1}{3}{-135};    
    \node [label=$P_y$] at (9,3) {};
   \node [empty vertex, label= left:\footnotesize{$z$}] (z) at (8, 1) {};
   \draw (x_2) -- (z);
   \draw (y_2) -- (z);
    \node[filled vertex, label= left:\footnotesize{$z_a$}] (z_a) at  (8, 0){};
    \draw [line width=0.5mm] (z) -- (z_a);
    
    \node[filled vertex, label= left:\footnotesize{$z_b$}] (z_b) at  (8,-6){};
    % \draw [ultra thick] (8, -2.5) -- (z_b);    
    \pgadget{(8,0)}{0.25}{0.1}{5.40}{270};      
     \node [label=$P_z$] at (7.25,-4) {}; 

  \draw[decoration={calligraphic brace, amplitude=6pt}, decorate, line width=1pt] (6.75, -6) -- (6.75, 0); 

   \node [label=left:\tiny{$\lceil n'/4 \rceil+cn'/2+1$}] at (6.50, -3){};
    
    \node[filled vertex, label= below:\footnotesize{$v_{(cn'+3)^2}$}] (z_c) at  (10.5,-6){};
    \draw [line width=0.5mm] (z_b) -- (z_c);    

    \node[filled vertex, label= below:\footnotesize{$v_1$}] (c_1) at  (15.5,-6){};   

  %  \draw [ultra thick] (z_c) -- (c_1);

  \cyclegadget {(10.5, -6)}{(11.5, -5)} {(14.5, -5)} {(15.5, -6)}{0};
    
    \node (u_y_l_u) at (2.0, 4.75) {};
    \node (u_y_r_b) at (12, 2.25) {};
     \node[rectangle, draw,dashed,inner sep=0pt, line width = 0.2mm][fit= (u_y_l_u) (u_y_r_b),  inner xsep=2ex, inner ysep=1.5ex] {};  
    
    \node (l_y_l_u) at (2.0, -0.25) {};
    \node (l_y_r_b) at (7.9, -6.0) {};
     \node[rectangle, draw,dashed,inner sep=0pt, line width = 0.2mm][fit= (l_y_l_u) (l_y_r_b),  inner xsep=2ex, inner ysep=1.5ex] {};  

    \node (cycle_l_u) [label={[xshift=1mm, yshift=-2mm]$C'$}] at (10.05, -5) {};
    \node (cycle_r_b) at (15.5, -6.5) {};
     \node[rectangle, draw,dashed,inner sep=0pt, line width = 0.2mm][fit= (cycle_l_u) (cycle_r_b),  inner xsep=2ex, inner ysep=1.5ex] {};  

   \draw[->, line width=0.5mm] (10.25, 4) -- (15.75, 4);
    \draw[->, line width=0.5mm] (8.25, -2) -- (15.75, -2);
    \draw [line width=0.5mm] (12, -6.25) --(12, -8); 
   \draw[->, line width=0.5mm] (12, -8) --(16.25, -8);

   % Path corresponding to upper portion of Y
 %  \node[empty vertex] (p1_r_1) at  (18, 6.5){};
   \node[filled vertex] (p1_r_2) at  (18, 6){};
   \node[filled vertex] (p1_r_3) at  (18, 5.5) {};
   \node[filled vertex] (p1_r_4) at  (18, 5){};
   \node[filled vertex] (p1_r_5) at  (18, 4){};   
   \node[filled vertex] (p1_r_6) at  (18, 3.5){};
   \node[filled vertex] (p1_r_7) at  (18, 3){};
%   \node[empty vertex] (p1_r_8) at  (18, 2.5){}; 
   \node[empty vertex] (p1_l_3) at  (17.5, 5.5) {};
   \node[empty vertex] (p1_l_4) at  (17.5, 5){};
   \node[empty vertex] (p1_l_5) at  (17.5, 4){};   
   \node[empty vertex] (p1_l_6) at  (17.5, 3.5){};  
 %  \draw (p1_r_1) -- (p1_r_2);
   \draw [line width=0.5mm] (p1_r_2) -- (p1_r_3);   
   \draw [line width=0.5mm] (p1_r_3) -- (p1_r_4);      
   \draw [dotted, line width=0.5mm] (p1_r_4) -- (p1_r_5);      
   \draw [line width=0.5mm] (p1_r_5) -- (p1_r_6);   
   \draw [line width=0.5mm] (p1_r_6) -- (p1_r_7);      
%   \draw (p1_r_7) -- (p1_r_8);  
   \draw [line width=0.5mm] (p1_l_3) -- (p1_l_4);      
   \draw [dotted, line width=0.5mm] (p1_l_4) -- (p1_l_5);      
   \draw [line width=0.5mm] (p1_l_5) -- (p1_l_6);    
   % horizontal lines
   \draw [line width=0.5mm] (p1_l_3) -- (p1_r_3);   
   \draw [line width=0.5mm] (p1_l_4) -- (p1_r_4);   
   \draw [line width=0.5mm] (p1_l_5) -- (p1_r_5);   
   \draw [line width=0.5mm] (p1_l_6) -- (p1_r_6); 
   %oblique lines
   \draw [line width=0.5mm] (p1_l_3) -- (p1_r_2);
   \draw [line width=0.5mm] (p1_l_6) -- (p1_r_7);   

    \node (p1_l_u) at (17, 6) {};
    \node (p1_r_b) at (19, 3) {};
    \node[rectangle, draw,dashed,inner sep=0pt, line width = 0.2mm][fit= (p1_l_u) (p1_r_b),  inner xsep=2ex, inner ysep=1.5ex] {};  
    
   %Path correspondoing to lower portion of Y
 %  \node[empty vertex] (p2_r_1) at  (18, 0.5){};
   \node[filled vertex] (p2_r_2) at  (18, 0){};
   \node[filled vertex] (p2_r_3) at  (18, -0.5) {};
   \node[filled vertex] (p2_r_4) at  (18, -1){};
   \node[filled vertex] (p2_r_5) at  (18, -2){};   
   \node[filled vertex] (p2_r_6) at  (18, -2.5){};
   \node[filled vertex] (p2_r_7) at  (18, -3){};
 %  \node[empty vertex] (p2_r_8) at  (18, -3.5){}; 
   \node[empty vertex] (p2_l_3) at  (17.5, -0.5) {};
   \node[empty vertex] (p2_l_4) at  (17.5, -1){};
   \node[empty vertex] (p2_l_5) at  (17.5, -2){};   
   \node[empty vertex] (p2_l_6) at  (17.5, -2.5){};  
 %  \draw (p2_r_1) -- (p2_r_2);
   \draw [line width=0.5mm] (p2_r_2) -- (p2_r_3);   
   \draw [line width=0.5mm] (p2_r_3) -- (p2_r_4);      
   \draw [dotted, line width=0.5mm] (p2_r_4) -- (p2_r_5);      
   \draw [line width=0.5mm] (p2_r_5) -- (p2_r_6);   
   \draw [line width=0.5mm] (p2_r_6) -- (p2_r_7);      
 %  \draw (p2_r_7) -- (p2_r_8);  
   \draw [line width=0.5mm] (p2_l_3) -- (p2_l_4);      
   \draw [dotted, line width=0.5mm] (p2_l_4) -- (p2_l_5);      
   \draw [line width=0.5mm] (p2_l_5) -- (p2_l_6);    
   %horizontal lines
   \draw [line width=0.5mm] (p2_l_3) -- (p2_r_3);   
   \draw [line width=0.5mm] (p2_l_4) -- (p2_r_4);   
   \draw [line width=0.5mm] (p2_l_5) -- (p2_r_5);   
   \draw [line width=0.5mm] (p2_l_6) -- (p2_r_6); 
 %  oblique lines
   \draw [line width=0.5mm] (p2_l_3) -- (p2_r_2);
   \draw [line width=0.5mm] (p2_l_6) -- (p2_r_7);   

    \node (p2_l_u) at (17, 0) {};
    \node (p2_r_b) at (19, -3) {};
    \node[rectangle, draw,dashed,inner sep=0pt, line width = 0.2mm][fit= (p2_l_u) (p2_r_b),  inner xsep=2ex, inner ysep=1.5ex] {};

   % Cycle attached at the end of Y
   \node[filled vertex, label={[xshift=-1mm, yshift=1mm]$v_{m^2}$}] (c_1) at  (18, -8){};
    \node[filled vertex] (c_2) at  (20, -8){};
%    \draw [ultra thick] (c_1) -- (c_2);
    \pgadget{(18.5,-8)}{0.25}{0.1}{1.5}{0};  
   \node[filled vertex] (c_1-dash) at  (18.5, -8){};    
    \draw [line width=0.5mm] (c_1) -- (c_1-dash);    
    \draw [dotted, <->, line width=0.5mm] (18.5, -8.5) -- (20, -8.5);
    \node[label=below:\tiny{$P_{m}$}] (p_1) at  (19.2, -8.5){};   
    \node[filled vertex] (c_3) at  (20.5, -8){};
    \draw [line width=0.5mm] (c_2) -- (c_3); 
    
    \node[filled vertex] (c_4) at  (22, -8){};
  %  \draw [ultra thick] (c_3) -- (c_4);    
    \pgadget{(20.5,-8)}{0.25}{0.1}{1}{0};   
    \draw [dotted, <->, line width=0.5mm] (20.5, -8.5) -- (22, -8.5);
    \node[label=below:\tiny{$P_{m-1}$}] (p_2) at  (21.5, -8.5){};
    
    \node[filled vertex] (c_5) at  (22.5, -8){};
    \draw [line width=0.5mm] (c_4) -- (c_5); 

    \node[filled vertex] (c_6) at  (24.5, -8){};   
    \draw [line width=0.5mm] [dotted] (c_5) -- (c_6); 

    \node[filled vertex] (c_7) at  (25, -8){};       
    \draw [line width=0.5mm] (c_6) -- (c_7);   

    \node[filled vertex] (c_8) at  (25.5, -8){};       
    \draw [line width=0.5mm] (c_7) -- (c_8);       

    \node[filled vertex] (c_9) at  (26, -8){};       
    \draw [line width=0.5mm] (c_8) -- (c_9);       

    \node[filled vertex] (c_10) at  (26.5, -8){};       
    \draw [line width=0.5mm] (c_9) -- (c_10);   

    \node[filled vertex] (c_11) at  (27, -8){};       
    \draw [line width=0.5mm] (c_10) -- (c_11);       

    \node[filled vertex] (c_12) at  (27.5, -8){};       
    \draw [line width=0.5mm] (c_11) -- (c_12);       

    \node[filled vertex] (c_13) at  (28, -8){};       
    \draw [line width=0.5mm] (c_12) -- (c_13);       

    \node[filled vertex] (c_14) at  (28.5, -8){};       
    \draw [line width=0.5mm] (c_13) -- (c_14);       

     \node[filled vertex, label=right:\footnotesize{$v_1$}] (c_15) at  (29, -8){};       
     \draw [line width=0.5mm] (c_14) -- (c_15);       

    \node[empty vertex, label=right:\footnotesize{$q$}] (c_9_b) at  (26, -9){};       
     \draw [line width=0.5mm] (c_9_b) -- (c_9);       
     \draw [line width=0.5mm] (c_9_b) -- (c_7);
     \draw [line width=0.5mm] (c_9_b) -- (c_11);     

     \node[empty vertex,, label=right:\footnotesize{$p$}] (c_13_b) at  (28, -9){};       
     \draw [line width=0.5mm] (c_13_b) -- (c_13);       
     \draw [line width=0.5mm] (c_13_b) -- (c_12);
     \draw [line width=0.5mm] (c_13_b) -- (c_14);       

     \node[empty vertex,, label=left:\footnotesize{$r$}] (c_12_u) at  (27.5, -7){};       
     \draw [line width=0.5mm] (c_12_u) -- (c_10);       
     \draw [line width=0.5mm] (c_12_u) -- (c_8);
     \draw [line width=0.5mm] (c_12_u) -- (c_15); 
     
     \draw (c_1) to[bend left=50, line width=0.5mm] (c_15);
  %  \draw[dotted, line width=0.3mm] (u) -- (u');

    \draw [dotted, <->, line width=0.5mm] (18, -10) -- (29, -10);
    \node[label=below:\footnotesize{$m^2$}] (cycle_size) at  (22, -9.75){};

     \node (cycle_left_upper) at (17.5, -5.5) {};
     \node (cycle_right_lower) at (30, -10.5) {};

    \node[rectangle, draw,dashed,inner sep=0pt, line width = 0.2mm][fit= (cycle_left_upper) (cycle_right_lower),  inner xsep=2ex, inner ysep=1.5ex] {}; 
  \end{tikzpicture}}

%% file: dregular.tex
\section{ $d$-regular graphs}

We now extend the hardness result of \BNP\ from connected cubic graphs to connected $d$-regular graphs for every fixed $d \ge 4$. The overall strategy is as follows.
\begin{itemize}
\item First, we construct  a connected $d$-regular graph $H_d$ from a connected cubic graph $G$.

\item Next, we show that the burning number of $H_d$ satisfies $b(H_d) \in \{\, b(G),\, b(G)+1 \,\}$.
    
\item We then prove that $b(H_d)=b(G)+1$ if for every optimal burning sequence of $G$, there exists at least one vertex that is burned in the final step by a unique source.
    
\item Subsequently, we establish that the cubic graph $H$ produced by Construction~\ref{cons:H-combined} indeed satisfies this structural property.
    
\item Since the \BNP\ is NP-complete on the connected cubic graphs obtained from Construction~\ref{cons:H-combined}, it follows that the \BNP\ remains NP-complete on connected $d$-regular graphs for every fixed $d \ge 4$. 
\end{itemize}

\begin{construction}\label{cons:dreg}
Let $G$ be a connected cubic graph on $n$ vertices with $V(G)=\{u_1,\dots,u_n\}$. For an integer $d\ge 4$, construct a graph $H_d$ as follows. Take $d-2$ copies of $G$, denoted by $G_1,\dots,G_{d-2}$. For each $i\in\{1,\dots,n\}$, let $z_{(i,j)}$ denote the copy of $u_i$ in $G_j$. For every fixed $i$, make the set of vertices $\{z_{(i,1)},\dots,z_{(i,d-2)}\}$ a clique.  The resulting graph is denoted by $H_d$.
\end{construction}
\medskip
\begin{remark}
$H_d$ is connected and $d$-regular. Indeed, each vertex has three neighbors inside its copy of $G$ and $d-3$ neighbors inside its clique, giving total degree $d$. %Moreover, $H_d \cong H \square K_{d-2}$.
\end{remark}

The following definition and Observations are based on Construction~\ref{cons:dreg}. 
\begin{definition}[projection of $z_{(i,j)}$ onto $G_k$:]
Let $H_d$ be a graph obtained from $G$ by Construction~\ref{cons:dreg}, and let $G_1,\dots,G_{d-2}$ denote the $d-2$ copies of $G$ used in the construction. For each $i\in[n]$ and $j\in[d-2]$, let $z_{(i,j)}$ denote the vertex in $G_j$ corresponding to $u_i\in V(G)$.
For $1\le k\le d-2$, the \emph{projection of $z_{(i,j)}$ onto $G_k$} is defined as ${Proj}(z_{(i,j)},G_k)=z_{(i,k)}$. For a vertex $v\in V(H_d)$, we denote by ${Proj}(v,G_k)$ its projection onto $G_k$. 
For a sequence $B=(v_1,\dots,v_t)$ of vertices of $H_d$, the \emph{projection of $B$ onto $G_k$}, denoted by ${Proj}(B,G_k)$, is the sequence obtained by replacing each $v_\ell$ ($1 \leq \ell \leq t$) with ${Proj}(v_\ell,G_k)$ while preserving the order.
\end{definition}
Observe that ${Proj}(B,G_k)$ may contain repeated vertices. For instance, if $j_1\neq j_2$, then ${Proj}(z_{(i,j_1)},G_k)={Proj}(z_{(i,j_2)},G_k)=z_{(i,k)}$.

Let $3 \le d'<d$, be an integer. Then by $H_{d'}$, we denote the induced subgraph of $H_d$ containing the first $d'-2$ copies of $G$. Let  $v \in V(H_d)$. Then $Proj(v, H_{d'})$ is defined as follows. Let $v' = Proj(v, H_{d'})$. If $v \in V(H_{d'})$ then $v' = v$, otherwise $v'= Proj(v, G_1)$.  Let $B_d$ be a sequence of vertices from $V(H_d)$. Then the projection of $B_d$ on $H_{d'}$ denoted by $B_d'=Proj(B, H_{d'})$,  is obtained by replacing  every $v \in B_d$ by  $Proj(v, H_{d'})$ in  the same order.

The following observations are based on Construction~\ref{cons:dreg}.
\begin{observation}
\label{obs:distance-in-same-component}
If $m,k \in [d-2]$ and $m \neq k$, $\forall i,j \in [n], d(z_{(i,k)}, z_{(j,m)}) = d(z_{(i,k)}, z_{(j,k)})+1 = d(z_{(i,m)}, z_{(j,m)}) +1$.
\end{observation}

\begin{observation}
\label{obs:distance-in-projection-not-more}
Let $u, v \in V(H_d)$ and $ u' = Proj(u, H_{d'}), v' = Proj(v, H_{d'})$. Then, $d(u',v') \leq d(u, v)$. 
\end{observation}

\begin{proof}
For any $u, v \in V(H_d)$, there can be three cases:\begin{enumerate*}
\item $u,v \in V(H_d)\setminus V(H_d')$.

\item $u,v \in V(H_d)\cap V(H_d')$.

\item $u\in V(H_d)\cap V(H_d')$ and $v \in V(H_d)\setminus V(H_d')$.
\end{enumerate*} 
Initially, assume that $u,v \in V(H_d)\setminus V(H_d')$ and  $u\in V(G_i)$ and $v \in V(G_j)$, where $G_i, G_j$ are the $i^{th}$ and $j^{th}$ copy of $G$ respectively used in the construction of $H_d$ and $d'+1 \leq i, j \leq d$.  Then by the definition of $Proj(v,H_d')$, both $u'$ and $v'$ belong to $V(G_1)\subset V(H_{d'})$. Therefore, by Observation~\ref{obs:distance-in-same-component}, $d(u',v') \leq d(u, v)$. Now assume that, $u,v \in V(H_d)\cap V(H_d')$. Then, again by the definition of $Proj(v,H_d')$, $u'=u$ and $v'=v$. Therefore,  $d(u',v') =d(u, v)$. Finally, consider $u\in V(H_d)\cap V(H_d')$ and $v \in V(H_d)\setminus V(H_d')$. This means that $u' = u$ and $v' \in V(G_1)$. Then by Observation~\ref{obs:distance-in-same-component}, $d(u',v') \leq d(u, v)$. Thus, the observation follows.
\end{proof}

\begin{observation}
\label{obs:projected-seq-burns-H-d'}
Let $B_d=(b_1,\dots,b_t)$ be a burning sequence for $H_d$ and let $B_{d'} = Proj(B_d,H_{d'})$. Then $B_{d'}$ (possibly containing repeated vertices) burns $H_{d'}$ completely.
\end{observation}

\begin{proof}
Let $B_d=(b_1,\dots,b_t)$ be a burning sequence for $H_d$. By definition of graph burning, every vertex $v \in V(H_d)$ is burned by time $t$, and if $v$ is burned by the source $b_i$, then $i + d_{H_d}(b_i,v) \le t$.

Let $v \in V(H_{d'})$. Since $B_d$ burns $H_d$, there exists an index $i$ such that $i + d_{H_d}(b_i,v) \le t$. Let $b_i' = Proj(b_i,H_{d'})$. By Observation~\ref{obs:distance-in-projection-not-more},
$d_{H_{d'}}(b_i',v) \le d_{H_d}(b_i,v)$. Hence $i + d_{H_{d'}}(b_i',v) \le i + d_{H_d}(b_i,v) \le t$.
Therefore, $v$ is burned in $H_{d'}$ by time $t$ under the sequence $B_{d'}$. Since $v$ is arbitrary, every vertex of $H_{d'}$ is burned by $B_{d'}$.
\end{proof}

\begin{lemma}
\label{lem:duplicates-at-the-end}
Let $B_{d'} = Proj(B_d, H_{d'})$, where $B_d$ is an optimal burning sequence for $H_d$. If $B_{d'}$ contains duplicate vertices, then the duplicates must be the last two sources of $B_{d'}$.
\end{lemma}

\begin{proof}
Let $b = |B_d|$. Suppose that $B_{d'}$ contains duplicate vertices. Then there exist $z_{(i,j_1)}, z_{(i,j_2)} \in B_d$ such that either $j_1 = 1$ and $j_2 > d'-2$, or $j_1 > d'-2$ and $j_2 = 1$, or 
$j_1, j_2 > d'-2$. In each case, both vertices project to $z_{(i,1)}$ in $B_{d'}$, and hence $z_{(i,1)}$ appears twice in $B_{d'}$.

Let $B_d[{pos}_1] = z_{(i,j_1)}$ and $B_d[{pos}_2] = z_{(i,j_2)}$, where both vertices project to 
$z_{(i,1)}$ in $B_{d'}$ and ${pos}_1 < {pos}_2$. By Construction~\ref{cons:dreg}, all vertices 
$z_{(i,1)},\dots,z_{(i,d-2)}$ form a clique in $H_d$. Hence $z_{(i,j_1)}$ and $z_{(i,j_2)}$ are adjacent in $H_d$.

Suppose that ${pos}_1 < {pos}_2 - 1$. When $B_d$ is used to burn $H_d$, the vertex $z_{(i,j_1)}$ is $b_{{pos}_1}$. Since the two vertices are adjacent, the fire from $z_{(i,j_1)}$ reaches $z_{(i,j_2)}$ at step 
${pos}_1 + 1$. Because ${pos}_1 + 1 < {pos}_2$, the vertex $z_{(i,j_2)}$ is already burned before step 
${pos}_2$, contradicting the definition of a burning sequence, which requires that a vertex chosen as a source at a given step must be unburned at the beginning of that step. 
Therefore, ${pos}_2 = {pos}_1 + 1$, i.e., the two sources must occupy consecutive positions in $B_d$.

Now assume that these two consecutive sources do not occupy the last two positions of $B_d$. Thus,  they occur at positions ${pos}$ and ${pos}+1$ with ${pos}+1 < b$. Since the two vertices are adjacent, when the source at position ${pos_1}$ is burned, it burns the vertex at position ${pos}+1$ in the next step. Therefore, at the moment the vertex $z_{(i,j_2)}$ is selected as a source at step ${pos}+1$, it is burned at that same step by the fire originating from $z_{(i,j_1)}$. Hence, selecting $z_{(i,j_2)}$ as a source does not contribute to burning any vertex strictly earlier than it would already be burned by previously burned sources.

Since ${pos}+1 < b$, there exists at least one later source in $B_d$. By shifting all sources after position ${pos}+1$ one step earlier and removing $z_{(i,j_2)}$ from the sequence, we obtain a shorter burning sequence that still burns $H_d$ completely. This contradicts the optimality of $B_d$. Therefore, ${pos}+1 = b$, and the two sources $z_{(i,j_1)}$ and $z_{(i,j_2)}$ must occupy the last two positions of $B_d$. Consequently, the duplicates in $B_{d'}$ arise precisely from the last two sources.
\end{proof}

Since $|B_{d'}| = |B_d|$ and, by Observation~\ref{obs:projected-seq-burns-H-d'}, every vertex of $H_{d'}$ is burned by $B_{d'}$ no later than it is burned by $B_d$, it follows that $B_{d'}$ burns $H_{d'}$ within $|B_d|$ steps. The only potential obstruction to $B_{d'}$ being a valid burning sequence is the possible presence of duplicate vertices in $B_{d'}$. If $B_{d'}$ contains duplicates, then by Lemma~\ref{lem:duplicates-at-the-end}, the duplicates occur only in the last two positions of $B_{d'}$. 
Suppose that the last two entries of $B_{d'}$ are both equal to $z_{(i,1)}$. Let $B^1_{d'}$ be the sequence obtained from $B_{d'}$ by deleting the final occurrence of $z_{(i,1)}$. Thus $|B^1_{d'}| = |B_d|-1$.
If $B^1_{d'}$ burns $H_{d'}$ completely, then it is a valid burning sequence for $H_{d'}$ and we are done. 
Otherwise, let $v \in V(H_{d'})$ be a vertex that remains unburned after executing $B^1_{d'}$. Define $B^2_{d'}$ to be the sequence obtained by appending $v$ to the end of $B^1_{d'}$. Then $|B^2_{d'}| = |B_d|$.

\begin{claim}
$B^2_{d'}$ is a burning sequence for $H_{d'}$.
\end{claim}

\begin{proof}
Assume, for contradiction, that some vertex $u \in V(H_{d'})$ remains unburned after applying $B^2_{d'}$ on $H_d'$. Let $b_d = |B_d|$. Since $B_d$ is a burning sequence for $H_d$, every vertex of $H_d$, and hence every vertex of $H_{d'} \subseteq H_d$, is burned by time $b_d$ under $B_d$. Observe that the source chosen at the final step of a burning sequence does not spread fire, since it is burned at time $b_d$. Therefore, every other vertex of $H_d$ is burned by one of the first $b_d-1$ sources of $B_d$. In particular, the vertex $u$ is burned in $H_d$ by some source among the first $b_d-1$ entries of $B_d$. Let this source be $z_{(i,k_1)}$, selected at step $t \le b_d-1$. Then $t + d_{H_d}(z_{(i,k_1)},u) \le b_d$. We then have two cases.

\begin{enumerate}
\item $k_1 \le d'-2$.
In this case, $z_{(i,k_1)} \in V(H_{d'})$. Thus $z_{(i,k_1)}$ appears in $B^2_{d'}$ at the same position $t$. 
Because distances inside $H_{d'}$ coincide with those in $H_d$ for vertices of $H_{d'}$, we have
$t + d_{H_{d'}}(z_{(i,k_1)},u) = t + d_{H_d}(z_{(i,k_1)},u) \le b_d$. Hence $u$ is burned by time $b_d$ under $B^2_{d'}$, a contradiction.

\item  $k_1 > d'-2$.
In this case, $z_{(i,k_1)} \notin V(H_{d'})$ and is projected to $z_{(i,1)}$ in $B_{d'}$. Therefore, in $B^2_{d'}$, the vertex $z_{(i,1)}$ appears at the $t^{th}$ position. By Observation~\ref{obs:distance-in-projection-not-more}, $d_{H_{d'}}(z_{(i,1)},u) \le d_{H_d}(z_{(i,k_1)},u)$.

Therefore, $u$ is burned by time $b_d$ under $B^2_{d'}$, again a contradiction.
\end{enumerate}

In both cases we reach a contradiction. Therefore no vertex of $H_{d'}$ remains unburned after executing $B^2_{d'}$, and hence $B^2_{d'}$ is a valid burning sequence for $H_{d'}$.
\end{proof}

We thus obtain a burning sequence for $H_{d'}$ derived from $B_{d}$ as described in Algorithm~\ref{alg:alg2}.
\RestyleAlgo{boxruled}
\begin {algorithm} [h]
\caption{Construction of a burning sequence for $H_{d'}$ from $B_{d}$}
\label{alg:alg2}
\begin {enumerate}[label=\arabic*.]
\item $p = |B_d|, B_{d'} =  Proj(B_d, H_{d'}) = (x_1, x_2 \dots, x_p)$.
\item If $B_{d'}$ does not have duplicates \\
\hspace*{5mm} return $B_{d'}$ \\
Else\\
\hspace*{5mm} $B^1_{d'} = (x_1, x_2 \dots x_{(p-1)})$.\\
\hspace*{5mm} If $B^1_{d'}$ is a valid burning sequence for $H_{d'}$\\
\hspace*{10mm} return $B^1_{d'}$\\
\hspace*{5mm} Else\\
\hspace*{10mm} $B^2_{d'}=B^1_{d'}\cup v$ where $v$ is a vertex unburned after $(p-1)^{th}$ step\\ 
\hspace*{10mm} return $B^2_{d'}$
\end {enumerate}
\vspace*{2mm}
\end {algorithm}

\begin{lemma}
\label{lem:burning-number-limits}
For $\forall d \geq 3,\ b(G) \leq b(H_d) \leq b(G)+1$.  
\end{lemma}

\begin{proof}
Since $H_3 \cong G$, the statement holds for $d=3$. Assume $d>3$. Let $B_d$ be an optimal burning sequence of $H_d$. Project $B_d$ onto $H_3 \cong G$.
By Observation~\ref{obs:projected-seq-burns-H-d'}, the projected sequence burns $G$. Then by Algorithm~\ref{alg:alg2}, we obtain a valid burning sequence for $G$ of length at most $|B_d|$.
Hence $b(G) \le b(H_d)$.

Let $B$ be an optimal burning sequence of $G_1 \cong G$. Assign $B$ on $H_d$ by burning the corresponding vertices in the first copy $G_1$. If a vertex $u_i$ is burned at step $t$ in $G_1$, then every vertex $z_{(i,k)}$ in its clique is adjacent to $z_{(i,1)}$ and therefore becomes burned at step $t+1$. Thus, every vertex of $H_d$ is burned no later than one step after its corresponding vertex in $G_1$ is burned. Consequently, $b(H_d) \le b(G) + 1$. Combining both inequalities gives $b(G) \le b(H_d) \le b(G)+1$.
\end{proof}

The following observation is a consequence of Lemma~\ref{lem:burning-number-limits}.
\begin{observation}
\label{obs:monotonic}
$b(H_3 = G) \leq b(H_4) \leq b(H_5) \leq \dots \leq b(G)+1$.
\end{observation}

Let $G$ be a graph with a burning sequence $B$. We define the following sets with respect to $G$ and $B$.
\begin{enumerate}
\item $BL(G, B) = \{v \in V(G) | v \text{ is burned in the last step}\}$.
\item $UB(G, B)= \{v \in V(G) | v \text{ is burned by a unique burning source} \}$
\end{enumerate}

\begin{lemma}
\label{lem:b-h4-one-more}
Let $G$ be a connected cubic graph such that for every optimal burning sequence $B$ of $G$,
$BL(G,B) \cap UB(G,B) \neq \emptyset$. Let $H_4$ be a 4-regular graph obtained from $G$ by Construction~\ref{cons:dreg}. Then $b(H_4) = b(G) + 1$.
\end{lemma}
\begin{proof}
Let $b=b(G)$. By Lemma~\ref{lem:burning-number-limits}, $b(H_4)\in\{b,b+1\}$. Assume, for contradiction, that $b(H_4)=b$. Let $B_H$ be an optimal burning sequence of $H_4$ of length $b$. Let $B_G$ be an optimal burning sequence of $G_1\cong G$ of length $b$ obtained by Algorithm~\ref{alg:alg2}. By assumption, there exists
$z_{(i,1)}\in BL(G_1,B_G)\cap UB(G_1,B_G)$. Let $z_{(j,1)}$ be the unique source in $B_G$ that burns $z_{(i,1)}$. We consider two cases.

\begin{enumerate}
\item  $z_{(i,1)}\neq z_{(j,1)}$. Suppose the two vertices $z_{(i,1)}$ and $z_{(i,2)}$ are burned by a single source $z_{(j,m)}$, $m\in\{1,2\}$. If $m=1$, then the fire from $z_{(j,1)}$ burns $z_{(i,2)}$ only in the $(b+1)^{\text{th}}$ step, contradicting the assumption $b(H_4)=b$. If $m=2$, then in order to burn $z_{(i,1)}$ by $z_{(j,2)}$ in the last step, $z_{(j,2)}$ must burn $z_{(i,2)}$ in the $(b-1)$-th step. When the projection of $B_H$ is taken onto $G_1$, it yields the source $z_{(j,1)}$ that burns
$z_{(i,1)}$ in the $(b-1)$-th step, contradicting $z_{(i,1)} \in BL(G_1, B_G)$. Hence, $z_{(i,1)}$ and $z_{(i,2)}$ must be burned by distinct sources in $B_H$.

Now we claim that both $z_{(i,1)}$ and $z_{(i,2)}$ are  burned in the $b^{th}$ step of $B_H$. If either $z_{(i,1)}$ or $z_{(i,2)}$ is burned in the $(b-1)^{\text{th}}$ step, then $z_{(i,1)}\notin BL(G_1,B_G)\cap UB(G_1,B_G)$, a contradiction. Also, the distinct sources that burn $z_{(i,1)}$ and $z_{(i,2)}$ cannot be $z_{(j,1)}$ and $z_{(j,2)}$, as $d(z_{(i,1)},z_{(j,1)})=d(z_{(i,2)},z_{(j,2)})$ which will enforce them to appear at the same position in $B_H$, which is impossible. Thus, $z_{(i,1)}$ and $z_{(i,2)}$ are burned by distinct sources $z_{(p,m)}$ and $z_{(q,m')}$ with $p\neq q$. Their projections in $B_G$ both burn $z_{(i,1)}$ in the $b^{\text{th}}$ step, contradicting $z_{(i,1)}\in UB(G_1,B_G)$.

\item $z_{(i,1)}=z_{(j,1)}$. Then $z_{(i,2)}$ must be burned in the $b^{\text{th}}$ step of $B_H$; otherwise, $z_{(i,1)}$ would be burned earlier in $B_G$. Moreover, all neighbors of $z_{(i,1)}$ and $z_{(i,2)}$ must also be burned in the $b^{\text{th}}$ step. If any neighbor is burned earlier, its projection in $G_1$ would burn $z_{(i,1)}$ in the $b^{\text{th}}$ step, contradicting $z_{(i,1)}\in UB(G_1,B_G)$. If all such neighbors are burned in the last step, then at most one of $z_{(i,1)}$ and $z_{(i,2)}$ can be burned, contradicting that $B_H$ burns all vertices of $H_4$.
\end{enumerate}
In both cases, we obtain a contradiction. Hence, $b(H_4)=b+1$.
\end{proof}

We now show that every connected cubic graph $H$ constructed using Construction~\ref{cons:H-combined} satisfies the property stated in Lemma~\ref{lem:b-h4-one-more}; that is, for every optimal burning sequence $B$ of $H$, $BL(H,B)\cap UB(H,B)\neq\emptyset$.

Let $G$ be a connected cubic graph. Let $G'$ and $ H$ be the graphs constructed using Construction~\ref{cons:G'} and Construction~\ref{cons:H-combined}. Let $B$ be an optimal burning sequence for $H$. By Lemma~\ref{lem:b-equals-threshold}, $|B| = k'+cn'+3$ where $n' = \vert V(G') \vert$ and 
$k' = \beta(G')$.

Let \textit{StartBlock} denote the subsequence $B[1, \dotsc, k']$, \textit{MiddleBlock} denote the subsequence $B[k'+1, \dotsc, s+h+1]$ and \textit{EndBlock} denote the subsequence $B[k'+h+2, \dotsc, k'+cn'+3]$.  Note that,  $\forall p, 1 \leq p \leq k', StartBlock[p] = B[p], \forall q, 1 \leq q \leq h+1, MiddleBlock[q] = B[k'+q]$ and $\forall r, 1 \leq r \leq cn'-h+2, EndBlock[r] = B[k'+h+1+r]$. Also note that, $|EndBlock| = l_1+l_2 = cn'-h+2, |MiddleBlock| = m = h+1$ and $|MiddleBlock|+|EndBlock| = (h+1)+(cn'-h+2) = cn'+3$.

Let $Owners = \{u \vert u \in V(G'), \exists i$ such that $1 \leq i \leq k'$ and $B[i] \in Dom_u\}$. In other words, $Owners$ is the set of \textit{owner} for the vertices in $StartBlock$ that belong to $InsideDomains$.

\begin{lemma}\label{lem:owners-vc}
Let $H$ be the graph obtained from $G'$ by Construction~\ref{cons:H-combined}, and let $B$ be an optimal burning sequence of $H$ with $|B| = k' + cn' + 3$. If $Owners$ set forms a vertex cover of $G'$, then at least one vertex in $trunk'(C)$ belongs to $BL(H, B) \cap UB(H, B)$. 
\end{lemma}

\begin{proof}
Let $m := cn' + 3$.  
Recall that $|MiddleBlock \cup EndBlock| = m$ and $|trunk'(C)| = m^2 - 5 $. Assume for contradiction that $BL(H,B)\cap UB(H,B)=\emptyset$. Then every vertex burned in the last step is burned by at least two distinct sources.

Since $Owners$ is a vertex cover for $G'$, the vertices in $StartBlock$ belong to $InsideDomains$. Let $p$ denote the number of vertices of $trunk'(C)$ burned by sources in $StartBlock$. Since $|trunk'(C)| = m^2 - 5$,
the remaining $m^2-5-p$ of vertices of trunk'(C) must be burned by $m$ sources in $MiddleBlock\cup EndBlock$. 
On a path or a cycle, a source with $i-1$ burning steps available to spread the fire from it can burn at most $2(i-1)+1$ vertices forming a path of $2i-1$ burned vertices.

If no vertex burned in the last step is uniquely burned, each such interval contributes at most $2i-3$ new burned vertices of $trunk'(C)$. Counting every overlapped endpoint once in an interval, we take the contribution of each interval to be at most $(2i-3)+1 = 2i-2$ vertices of $trunk'(C)$. Hence, the total number of trunk'(C) vertices burned by $MiddleBlock\cup EndBlock$ is at most
\[
\sum_{i=1}^{m} (2i-2) = 2\sum_{i=1}^{m} i - 2m = m(m+1) - 2m = m^2 - m.
\]
Thus, we have, 
\[
m^2 - 5 - p \le m^2 - m .
\]
Hence,
\[
p + 5 \ge m .
\tag{$\ast$}
\]

Each vertex in a minimum vertex cover is incident to at least one edge exclusively; otherwise, that vertex can be removed from the minimum vertex cover, which violates the minimality of the minimum vertex cover. Thus, each vertex of $StartBlock$ belongs to a minimum vertex cover of $G'$ and must burn the tips of the $BTP$-gadgets corresponding to at least one exclusive edge. Consequently, the vertex $w\in StartBlock$,
that is closer to the end vertex $v_{m^2}$ of the $C$-gadget is at a distance at least $\frac{n'}{2} + cn' + 2 - k'$. Note that, $d(x,v_{m^2}) = d(y,v_{m^2}) = \frac{n'}{2} + cn' + 2$, and $w$ can be at most $k'$ edges closer along $P_x$ or $P_y$. Since $|B| = k' + m$ and $m = cn' + 3$, the maximum remaining burning time after reaching $v_{m^2}$ is
\[
(k' + m) - \left(\frac{n'}{2} + cn' + 2 - k'\right) = 2k' - \frac{n'}{2} + 1.
\]

As $G'$ is cubic, every vertex cover satisfies $k' \le \frac{2n'}{3}$.
Hence
\[
2k' - \frac{n'}{2} + 1 \le \frac{4n'}{3} - \frac{n'}{2} + 1 = \frac{5n'}{6} + 1.
\]

From a vertex of a path with $t$ remaining steps, at most $2t+1$ vertices can be burned. Therefore, each source in $StartBlock$ can burn at most
\[
2\!\left(\frac{5n'}{6}+1\right)+1 \le \frac{10n'}{6} + 3
\]
vertices of $trunk'(C)$.

Since $trunk'(C)$ is a cycle with a unique end vertex $v_{m^2}$, every source in $StartBlock$ can enter the $C$-gadget only through $v_{m^2}$ due to the structure of $H$. Consequently, the fire induced on $trunk'(C)$ by any such source spreads symmetrically along the cycle from $v_{m^2}$. Therefore, the sets of vertices of $trunk'(C)$ burned by distinct sources in $StartBlock$ are {vertices} of the cycle centered at $v_{m^2}$. Further, if fire from two sources, say $s_1, s_2 \in StartBlock$, enter $C$-gadget and if the fire from $s_1$ reaches $v_{m^2}$ before the fire from $s_2$, then the vertices in $C$-gadget burned by $s_2$ is a subset of the vertices in $C$-gadget burned by $s_1$ and therefore,  these intervals are nested. Thus, we obtain $p  \le \frac{10n'}{6} + 3$. With $m = cn' + 3$  and by the choice of the constant $c$ in the construction, for any $n'$ we have $p + 5 < m$, contradicting $(\ast)$.

Therefore $BL(H,B)\cap UB(H,B)\neq\emptyset $.
\end{proof}

Now we consider the case when {$Owners$} is not a minimum vertex cover and prove that in this case also $BL(H,B)\cap UB(H,B)\neq\emptyset $.

\noindent {\bf Sketch of the proof:} 
We start with the assumption that $BL(H,B)\cap UB(H,B) = \emptyset $. Since $Owners$ is not a vertex cover of $G'$, there exists at least one uncovered edge in $G'$. Hence, the corresponding $BTP$-gadget has no source from $StartBlock$. By {Lemma~\ref{lem:if-one-unrepresented-then-two-unrepresented}, burning such a gadget requires at least three sources from $MiddleBlock$ and to compensate for this extra use of sources (relative to the burning requirements of the $C$-gadget), at least two sources from $StartBlock$ must be placed in $OutsideDomains$. Consequently, $|Owners| \le k'-2$. Therefore, at least two edges of $G'$ are uncovered, and the corresponding two $BTP$-gadgets must both be burned using only sources from $MiddleBlock \cup EndBlock$. {By Lemma~\ref{lem:two-unrepresented-edges-BL-cap-UB-not-empty}, we get $BL(H,B)\cap UB(H,B)\neq\emptyset$, a contradiction to our assumption}. {This contradiction shows that $BL(H,B)\cap UB(H,B)\neq\emptyset$ even when $Owners$ is not a minimum vertex cover}}

\begin{lemma}\label{lem:fivecn-implies-two-outside}
Let $B'$ be a burning sequence of length $s+cn'+3$, where $0\le s\le k'$. If at least $5cn'$ vertices of $trunk'(C)$ are burned by sources in $StartBlock$, then at least two sources in $StartBlock$ belong to $OutsideDomains$.
\end{lemma}

\begin{proof}
By Observation~\ref{obs:max-vertices-in-C-burned-by-insidedomain}, a source in $StartBlock\cap InsideDomains$ can burn fewer than $2cn'$ vertices of $trunk'(C)$. A source in $OutsideDomains$, even in the most favorable configuration of burning sources, that is, the first source placed on a vertex belonging to $trunk'(C)$, can burn at most $1+2(k'+cn'+2) <3cn'$ vertices of $trunk'(C)$, since $k'\le 2n'/3$. Thus, a single source in $StartBlock$ that belongs to $OutsideDomains$ cannot burn $3cn'$  vertices. Therefore, at least two sources in $StartBlock$ must belong to $OutsideDomains$.
\end{proof}

\begin{lemma}\label{lem:three-middleblock-fivecn}
If at least three sources of $MiddleBlock$ are placed in some gadget $BTP_{uv}$, then at least $5cn'$ vertices of $trunk'(C)$ are burned by sources in $StartBlock$. Moreover, at least two sources of $StartBlock$ belong to $OutsideDomains$.
\end{lemma}

\begin{proof}
Any source placed inside $BTP_{uv}$ is at distance at least $cn'+n'/2+1$ from $C$. Hence its fire cannot infiltrate $C$ before $cn'+n'/2+1$ steps, that is, the fire from sources in $MiddleBlock$ placed inside $BTP_{uv}$ cannot infiltrate the C gadget. Applying Lemma~\ref{lem:middbloc-insidedomain-trunk-startblock}
with $t=3$, the number of trunk vertices that must be burned by $StartBlock$ is at least
\[
3^2+12-6h+6cn'-5 = 6cn'-6h+16.
\]
For $4\le c<8$ and $n'\ge 6$, this quantity is at least $5cn'$. The second assertion follows from
Lemma~\ref{lem:fivecn-implies-two-outside}.
\end{proof}

\begin{lemma}
\label{lem:tips_by_middleblock'-endblock}
The sources in $MiddleBlock[2,\dotsc, h+1] \cup EndBlock$ can burn at most $5cn'$ $tip \in V(H_{core})$.
\end{lemma}
\begin{proof}
Let $StartBlock \cup MiddleBlock$ be reorganized as follows. Let $StartBlock'$ comprise of $StartBlock$ followed by $MiddleBlock[1]$ and $MiddleBlock'$ comprise of $MiddleBlock[2,\dotsc,h+1]$. Let $s' = |StartBlock'|, m' = |MiddleBlock'|$. Clearly, $s' = k'+1, m' = (h+1)-1=h$. Thus, $B$ is concatenation of sequences namely,  $StartBlock', MiddleBlock'$ and $EndBlock$ that is, $B = StartBlock'\circ MiddleBlock'  \circ EndBlock$. By Lemma~\ref{lem:max-leaves-in-endblock}, the sources in $EndBlock$ can burn at most $l_1+l_2 = cn'-h+2 < cn'$ $tip$s. By Lemma~\ref{lem:max-leaves-in-middleblock}, the sources in $MiddleBlock'$ can burn at most $2^m-1 = 2^h -1 = 4cn'-1$ tips. Thus, the sources in $MiddleBlock' \cup EndBlock$ together can burn at most $4cn'-1+cn' < 5cn'$ $tip$s that belong to some $BTP$-gadgets in $H_{core}$.
\end{proof}

\begin{lemma}
\label{lem:if-one-unrepresented-then-two-unrepresented}
If there exists one \textit{unrepresented} edge and $BL(H,B)\cap UB(H,B)=\emptyset$, then there exist at least two \textit{unrepresented} edges.
\end{lemma}

\begin{proof}
Let $(u,v)\in E(G')$ be an unrepresented edge. Then by no source in $StartBlock$ corresponds to $(u,v)$ (recall the definition of unrepresented edges), and therefore the gadget $BTP_{uv}$ must be burned entirely by
sources in $MiddleBlock\cup EndBlock$.

Since $|Tips_{uv}|=4cn'$, all $4cn'$ tips of $BTP_{uv}$ must be burned by these sources. Moreover, because
$BL(H,B)\cap UB(H,B)=\emptyset$, every tip burned in the last step must also be burned by at least one additional source. We distinguish according to the position of $MiddleBlock[1]$.

\begin{enumerate}
\item $MiddleBlock[1]\in\{r_{uv},r_{vu}\}$. Without loss of generality let $MiddleBlock[1]=r_{uv}$. Then its fire reaches all vertices of $Tips_{uv}$ in the last step. Hence, each of the $4cn'$ tips must also be burned by another source due to the assumption. By Lemma~\ref{lem:max-leaves-in-endblock}, the sources in $EndBlock$ can burn fewer than $cn'$ tips.Therefore the sources in $MiddleBlock[2,\dots,h+1]$ must burn more than $3cn'$ tips. By Lemma~\ref{lem:max-tips-from-source-in-middleblock}, $MiddleBlock[2]$ can burn at most $2cn'$ and $MiddleBlock[1]$ can burn at most $cn'$ and every later source in $MiddleBlock$ can burn fewer than $cn'$ tips. Hence at least three sources of $MiddleBlock$ must lie inside $BTP_{uv}$.

\item  $MiddleBlock[1]\in V(BTP_{uv})\setminus\{r_{uv},r_{vu}\}$.
If a source from $MiddleBlock[2,\dots,h+1]$ is placed outside $BTP_{uv}$, it cannot burn any tip of $BTP_{uv}$. Thus burning of $Tips_{uv}$ is by a fire spread  from sources placed inside $BTP_{uv}$.
If $MiddleBlock[2]$ is placed at a child of $r_{uv}$ or $r_{vu}$, its fire reaches $2cn'$ tips in the last step. These $2cn'$ tips must get burned once more by other source(s). Since $EndBlock$ burns fewer than $cn'$ tips, at least $cn'+1$ additional tips must be burned by sources in $MiddleBlock[3,\dots,h+1]$. By Lemma~\ref{lem:max-tips-from-source-in-middleblock},$MiddleBlock[3]$ can burn at most $cn'$ tips. Therefore, at least one source in $MiddleBlock[4,\dotsc, h+1]$ is required to burn $tips_{uv}$. Thus, at least three sources of $MiddleBlock$ lie inside $BTP_{uv}$.

If $MiddleBlock[2]$ is placed deeper inside $BTP_{uv}$, it burns at most $cn'$ tips. Since $EndBlock$ burns fewer than $cn'$ tips, a third source in $MiddleBlock$ is necessary to burn the remaining tips the first time as B is a burning sequence for H. Again, at least three sources lie inside $BTP_{uv}$.

\item $MiddleBlock[1]\notin V(BTP_{uv})$. Then all $4cn'$ tips must be burned by $MiddleBlock[2,\dots,h+1]\cup EndBlock$. Since $EndBlock$ burns fewer than $cn'$ tips, the sources in $MiddleBlock[2,\dots,h+1]$ must burn more than $3cn'$ tips. As mentioned in case 1, to burn at least $3cn'+1 tips_{uv}$, at least three sources of $MiddleBlock$ must lie inside $BTP_{uv}$.

\end{enumerate}
In all cases, at least three sources of $MiddleBlock$ are placed inside $BTP_{uv}$. By Lemma~\ref{lem:three-middleblock-fivecn}, at least two sources in $StartBlock$ belong to $OutsideDomains$. Hence $|Owners|\le k'-2$, and therefore there exist at least two unrepresented edges.
\end{proof}

\begin{lemma}
\label{lem:two-unrepresented-edges-BL-cap-UB-not-empty}
If there are at least two unrepresented edges in $H$, then $BL(H, B) \cap UB(H, B) \neq \emptyset$.
\end{lemma}
\begin{proof}
Let $(u,v), (e,f) \in E(G')$ be two \textit{unrepresented} edges. Since $B$ is a burning sequence for $H$, the sources in $MiddleBlock \cup EndBlock$ must burn $Tips_{uv}$ and $Tips_{ef}$. For the sake of contradiction, assume $BL(H, B) \cap UB(H, B) = \emptyset$. There are following possibilities. 
\begin{enumerate}
\item $MiddleBlock[1] \in \{r_{uv}, r_{vu}, r_{ef}, r_{fe}\}$: Without loss of generality, let $MiddleBlock[1]$ is placed at $r_{uv}$. The fire from $r_{uv}$ will reach $Tips_{uv}$ in the last step. By assumption, $BL(H, B) \cap UB(H, B) = \emptyset$. Therefore, every $tips_{uv} \in Tips_{uv}$ must get burned by another source. Therefore, the sources in $MiddleBlock[2,\dotsc, h+1] \cup EndBlock$ must burn $4cn'$ $Tips_{uv}$ and $4cn'$ $Tips_{ef}$, that is, overall $8cn'$ $tip$s. By Lemma~\ref{lem:tips_by_middleblock'-endblock}, the sources in $MiddleBlock[2,\dotsc, h+1] \cup EndBlock$ can burn at most $5cn'$ $tip$s. This is a contradiction.
\item $MiddleBlock[1] \in V(BTP_{uv}) \cup V(BTP_{ef}) \setminus \{r_{uv}, r_{vu}, r_{ef}, r_{fe}\}$: In this case, $MiddleBlock[1]$ is placed at a vertex inside $BTP_{uv}$ or $BTP_{ef}$ except at any of the vertex in $\{r_{uv}, r_{vu}, r_{ef}, r_{fe}\}$. Without loss of generality, let $MiddleBlock[1]$ is placed at $lchild(r_{uv})$. Let it be denoted by $w$. The fire at $w$ will reach $tips_w$ after $cn'+1$ steps. Thus, $2cn'$ $tip_{uv}$ will be burned in the penultimate step. Therefore, the sources in $MiddleBlock[2,\dotsc, h+1] \cup EndBlock$ must burn the remaining $2cn'$ $tip_{uv}$ and $4cn'$ $tip_{ef}$. However, by Lemma~\ref{lem:tips_by_middleblock'-endblock}, the sources in $MiddleBlock[2,\dotsc, h+1] \cup EndBlock$ can burn at most $5cn'$ $tip$s. This is a contradiction.
\item $MiddleBlock[1] \notin BTP_{uv}$: In this case, the sources $MiddleBlock[2,\dotsc, h+1]\cup EndBlock$ are expected to burn $8cn' (= 4cn'\ tip_{uv} + 4cn'\ tip_{ef})$ $tip$s. By Lemma~\ref{lem:tips_by_middleblock'-endblock}, the sources in $MiddleBlock[2,\dotsc, h+1] \cup EndBlock$ can burn at most $5cn'$ $tip$s. This is a contradiction.
\end{enumerate}
Since we get contradiction in every possible scenario, our assumption must be wrong. Therefore, $BL(H, B) \cap UB(H, B) \neq \emptyset$.
\end{proof}

\begin{lemma}
\label{lem:ownwers-not-vc}
Let $Owners$ not be a vertex cover of $G'$. Then, $BL(H, B) \cap UB(H,B) \neq \emptyset$.
\end{lemma}

\begin{proof}
There is at least one edge, say $(u,v) \in E(G')$ that is not represented by the sources in $StartBlock$. Thus, only the sources in $MiddleBlock \cup EndBlock$ are responsible for burning $BTP_{uv}$. For the sake of contradiction, let $BL(H, B) \cap UB(H, B) = \emptyset$. By Lemma~\ref{lem:if-one-unrepresented-then-two-unrepresented}, there must be two \textit{unrepresented} edges. Then, by Lemma~\ref{lem:two-unrepresented-edges-BL-cap-UB-not-empty}, $BL(H, B) \cap UB(H, B) \neq \emptyset$. This is a contradiction.
\end{proof}

\begin{lemma}
\label{lem:vertex-burned-uniquely-in-last-step}
Let $H$ be a connected cubic graph obtained from a connected cubic graph $G$ by Construction~\ref{cons:H-combined}. Then $BL(H, B) \cap UB(H,B) \neq \emptyset$.
\end{lemma}
\begin{proof}
 There are following two possibilities.
\begin{enumerate}
\item $Owners$ is a vertex cover of $G'$: By Lemma~\ref{lem:owners-vc}, $BL(H, B) \cap UB(H, B) \neq \emptyset$.
\item $Owners$ is not a vertex cover of $G'$: By Lemma~\ref{lem:ownwers-not-vc}, $BL(H, B) \cap UB(H, B) \neq \emptyset$.
\end{enumerate}
Therefore, $BL(H, B) \cap UB(H,B) \neq \emptyset$.
\end{proof}

Let $G'$ be the graph obtained from a connected cubic graph $G$ using Construction~\ref{cons:G'} and
let $H$ be the connected cubic graph obtained from $G'$ by Construction~\ref{cons:H-combined}. Let   $H_d$ be the $d$-regular graph obtained from $H$ by Construction~\ref{cons:dreg}.
Then, we have the following lemmas.
\begin{lemma}
\label{lem:b-hd-one-more}
Let $H$ be a connected cubic graph obtained from Construction~\ref{cons:H-combined} and $H_d$ be a $d$-regular graph obtained from $H$ by Construction~\ref{cons:dreg}. Then $b(H_d) = b(H)+1 = k'+cn'+4, \forall d \geq 4$.
\end{lemma}
\begin{proof}
By Lemma~\ref{lem:vertex-burned-uniquely-in-last-step}, for any optimal burning sequence $B$ for $H$, $BL(H,B) \cap UB(H,B) \neq \emptyset$. By Lemma~\ref{lem:b-h4-one-more}, for $d =4, b(H_d) = b(H)+1$. By Observation~\ref{obs:monotonic}, $\forall d \geq 4, b(H_d) = b(H)+1$.
\end{proof}

\begin{lemma}
\label{lem:dregburning}
$b(H_d) = k'+cn'+4$ if and only if \textit{vertex covering number} of $G'$ is $k'$.
\end{lemma}

\begin{proof}
\noindent
\textbf{($\Rightarrow$)} 
Assume that the vertex covering number of $G'$ is $k'$.  By Lemma~\ref{lem:b-equals-threshold}, we have 
$b(H) = k' + cn' + 3$. By Lemma~\ref{lem:b-hd-one-more}, it follows that $b(H_d) = b(H) + 1 = k' + cn' + 4$.

\noindent
\textbf{($\Leftarrow$)} 
Assume that $b(H_d) = k' + cn' + 4$.  By Lemma~\ref{lem:b-hd-one-more}, we obtain $b(H) = b(H_d) - 1 = k' + cn' + 3$.  By Lemma~\ref{lem:b-equals-threshold}, the vertex covering number of $G'$ equals 
$b(H) - cn' - 3$. Substituting $b(H) = k' + cn' + 3$, we obtain
$(k' + cn' + 3) - cn' - 3 = k'$. Hence, the vertex covering number of $G'$ is $k'$.
\end{proof}

\begin{theorem}\label{thm:DRegNPC}
The \BNP\ is \NPC\ for connected $d$-regular graphs with $d>3$.
\end{theorem}

\begin{proof}
Let $G$ be a connected cubic graph on $n$ vertices. From $G$, construct $G'$ using Construction~\ref{cons:G'}.
From $G'$, construct the connected cubic graph $H$ using Construction~\ref{cons:H-combined}. Finally, from $H$, construct the connected $d$-regular graph $H_d$ using Construction~\ref{cons:dreg}. Clearly, $H_d$ is connected and $d$-regular, and all constructions are computable in polynomial time.

By Lemma~\ref{lem:G'vertexcover}, the vertex covering number of $G'$ is $k'$ if and only if $\beta(G)=k'-1$.
By Lemma~\ref{lem:dregburning}, the vertex covering number of $G'$ is $k'$ if and only if $b(H_d)=k'+cn'+4$.
Hence, $\beta(G)=k'-1 \quad \text{if and only if} \quad b(H_d)=k'+cn'+4$.

Thus, we obtain a polynomial-time reduction from \MVC\ on connected cubic graphs to \BNP\ on connected 
$d$-regular graphs. Since \MVC\ is \NPC\ for connected cubic graphs, it follows that
\BNP\ is \NPC\ on connected $d$-regular graphs for every fixed integer $d \ge 3$.
\end{proof}

By Theorem~\ref{thm:cubicAPX}, Lemma~\ref{lem:b-equals-threshold} and Lemma~\ref{lem:dregburning}, it can be seen that the \BNP\ on $d$-regular graphs is APX-hard.

\dregularthm*